\documentclass [honours, printsupervisor, printschool, copyrightpage, titlesmallcaps] {uqthesis}

\pdfoutput=1 

\newcommand {\figwidth} {100mm}
\usepackage{play}
\usepackage[grey,times]{quotchap}
\usepackage{makeidx}

\usepackage[square,comma,numbers,sort&compress]{natbib}
\usepackage[nottoc]{tocbibind}  
\usepackage{amsmath,amsfonts,amssymb} % this is handy for mathematicians and physicists
\usepackage{graphicx} % standard graphics package for inclusion of
\usepackage{epsfig} % older graphics package for inclusion of eps
\newif\ifpdf
\ifx\pdfoutput\undefined
	\pdffalse    % we are not running pdfLaTeX
\else
	\pdfoutput=1 % we are running pdfLaTeX
	\pdftrue
\fi

\ifpdf
	\DeclareGraphicsExtensions{.pdf}  % this command defined in graphicx
\else
	\DeclareGraphicsExtensions{.ps}
\fi

\usepackage[utf8]{inputenc}
\usepackage[T1]{fontenc}
\usepackage{textcomp}
\usepackage{url}
\usepackage{graphicx}
\usepackage{float}
\usepackage{booktabs}
\usepackage{enumitem}
\usepackage{siunitx}
\usepackage[colorlinks=true
,urlcolor=blue
,anchorcolor=blue
,citecolor=blue
,filecolor=blue
,linkcolor=blue
,menucolor=blue
,pagecolor=blue
,linktocpage=true
,pdfproducer=medialab
,pdfa=true
]{hyperref}

\usepackage{caption}
\usepackage{subcaption}
\usepackage{quotchap}

\usepackage[activate={true},final,tracking=true,kerning=true,spacing=true,factor=1100,stretch=20,shrink=20]{microtype}
\microtypecontext{spacing=nonfrench}

\usepackage{fancyhdr}
\usepackage{geometry}

\usepackage{parskip}

\usepackage{emptypage}
\usepackage{multicol}
\usepackage[dvipsnames]{xcolor}

\usepackage{amsmath, amsfonts, mathtools, amsthm, amssymb}
\usepackage{physics}
\usepackage{mathrsfs}
\usepackage{cancel}
\usepackage{bm}

\usepackage{centernot}

\newcommand\N{\ensuremath{\mathbb{N}}}
\newcommand\R{\ensuremath{\mathbb{R}}}
\newcommand\Z{\ensuremath{\mathbb{Z}}}

\let\implies\Rightarrow

\let\epsilon\varepsilon
\let\phi\varphi

\makeatother
\usepackage{pgfplots}
\pgfplotsset{compat=1.18}
\usepackage{tikz-3dplot}

\usetikzlibrary{arrows.meta,bending}
\usetikzlibrary{angles,quotes}
\usetikzlibrary{decorations.pathmorphing,patterns}
\usepgfplotslibrary{fillbetween}
\usepackage{tensor}

\usepackage{mdframed}
\mdfsetup{skipabove=20pt,skipbelow=\topsep,leftmargin=\parindent}
\theoremstyle{definition}

\newtheorem*{note}{Note}

\newtheorem*{claim}{Claim}

\newtheoremstyle{enumstyle}{0}{}{}{}{\bfseries}{:}{ }{\thmname{#1} \thmnumber{#2}\thmnote{#3}}
\theoremstyle{enumstyle}

\newmdtheoremenv[nobreak=true,skipabove=10pt,skipbelow=\topsep]{definition}{Definition}%[section]
\newmdtheoremenv[nobreak=true]{theorem}{Theorem}%[section]
\newmdtheoremenv[nobreak=true]{corollary}{Corollary}[section]
\newmdtheoremenv[nobreak=true]{lemma}{Lemma}[section]

\makeatletter
\def\thm@space@setup{%
  \thm@preskip=\parskip \thm@postskip=0pt
}

\newcommand{\too}{\longrightarrow}

\renewenvironment{proof}
{ 
    \vspace{6 pt}
    \begin{mdframed}[
        linewidth=1pt,
        skipabove=\topskip,
        skipbelow=0pt,
        innertopmargin=0pt,
        innerbottommargin=6pt,
        bottomline=false,
        topline=false,
        rightline=false]%
    \noindent \textit{\textbf{Proof.}}  
}
{%
    \qed
\end{mdframed}
    \vspace{6 pt}
}

\mdfdefinestyle{exampledefault}{%
rightline=true,innerleftmargin=10,innerrightmargin=10,
frametitlerule=true,frametitlerulecolor=black,
frametitlebackgroundcolor=lime,
frametitlerulewidth=2pt,
shadow=true,shadowsize=4pt,
}

\usepackage{fontawesome}

\usepackage[capitalize]{cleveref}
\allowdisplaybreaks
\author{Cian Luke Martin}

\begin{document}

\frontmatter

%%%%% Acknowledgements, titlepage, abstract, list of publications
% Logo located at top of title page.
\logo{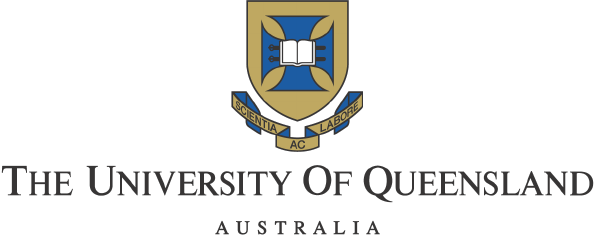}
\logoscale{1}

% Document & author titles.
\title{Perturbative Quantization of Modified Maxwell Electrodynamics}
\author{Cian Luke Martin}
\authorqual{}
\supervisor{Dr Gabriele Tartaglino Mazzucchelli}
\department{Physics}
\school{Mathematics and Physics}

% Construct document with the following order. Comment out as necessary.
\titlepage

\chapter{Abstract}

The standard model of particle physics is incomplete, with problems such as dark matter, quantum gravity and matter-antimatter asymmetry remaining unsolved despite extensive research. The most common modifications of the standard model aim to account for undiscovered physics through the addition of hypothetical particles to the standard model. A less common route in this endeavour is the introduction of new interactions between known particles. Further, the photon does not directly interact with itself in the standard model. However, a new nonlinear model called Modified Maxwell electrodynamics, or \textit{ModMax} for short, predicts a photon capable of interacting with itself without breaking the notable symmetries of Maxwell's theory: conformal invariance and electromagnetic duality.

ModMax has been studied extensively at the classical level with applications in strongly coupled condensed matter systems, however remains largely untouched in a quantum context. As such, it was the central aim of this project to perform the perturbative quantization of this theory. Using the background field method and dimensional regularization, I obtained novel corrections beyond what the classical theory predicts. By calculating the effective action, I showed that these corrections vanish in a constant background field, and are not of the form of the classical theory for a varying background field. 

Motivated by the form of the corrections obtained for ModMax, I applied the method I developed to quantize ModMax to its two dimensional analogue theory. This was the secondary aim of the project, and I obtained the effective action for a general background by evaluating all one loop Feynman diagrams, as well as the separate infinite series of two vertex diagrams. Lastly, I considered an alternative approach to quantization using auxiliary fields that captured the nonlinearity, and demonstrated that it is not possible to study ModMax in this fashion without breaking Lorentz symmetry. As few theories of nonlinear electrodynamics have been explored on the quantum domain, this quantization of ModMax represents a forward step in this endeavour. While ModMax's theoretical effects remain on a scale unreachable by current experimental techniques, I nonetheless characterized ModMax's predictions and its analogue's behaviour on the quantum domain.

\chapter{Acknowledgements}

I would like to thank Gabriele Tartaglino Mazzucchelli, my supervisor, for his guidance and support throughout this project. My development as an undergraduate physics student was undoubtedly shaped by Gabriele's teaching, direction and collaboration. I'd also like to thank Christian Ferko for invaluable feedback and discussions in the editing of this thesis.
I am endlessly grateful to my family and friends for their support throughout my studies.

\newpage
\textbf{Midnight's Burden}\\ \\

    It's Midnight's lonesome burden, \\
    to hold here those awake. \\
    Some whose sleep they stirred in, \\
    and some who did not break. \\

    To carry us bound waking, \\
    entrapped by our own thoughts. \\
    A prison of self making, \\
    we lie awake distraught. \\

    Midnight said to me last time, \\ 
    \textit{I see you far too much.} \\
    \textit{Let sleep take you, in your prime,} \\
    \textit{and enjoy dawn's soft touch.} \\

    Next time I meet you, Midnight, \\
    I'll tell you how it felt, \\
    To see that first golden light, \\
    and let my worries melt.

\tableofcontents

\listoffigures
%\listoftables
\chapter{List of Symbols}

% please change this list to suit your thesis

The following list is neither exhaustive nor exclusive, but may be helpful.
\begin{list}{}{%
\setlength{\labelwidth}{24mm}
\setlength{\leftmargin}{35mm}}
\item[$d$] number of spacetime dimensions (either $1+3$ or $1+1$)
    \item[$\mu$, $\nu$, $\rho$, $\tau$\dotfill] $d$-vector (usually $4$-vector) indices running $(0,1,2,3)$ 
    \item[$i$, $j$, $k$ \dotfill] in $d=4$, 3-vector indices running $(1,2,3)$ 
        \item[$i$, $j$, $k$ \dotfill]  in $d=2$, boson indices running $(1,\cdots,N)$, $N\in \N$
    \item[$g^{\mu \nu}$] metric tensor (always Minkowski or Euclidean)
        \item [$\R$, $\N$, $\Z$] real numbers, natural numbers, and integers respectively
        \item [$\mathcal{O}\left( \gamma^{n} \right) $] up to order $\gamma^{n}$, higher powers are discarded
\end{list}

\mainmatter

%%%%% Introduction
\chapter{Review: Introduction}

\section{Preface}

While common extensions to the standard model in the search for new physics often add new theoretical particles, a less common route in this endeavour is the addition of novel interactions between known particles. An emerging category of such extensions is models of \textit{nonlinear electrodynamics}, that is, adding self-interactions of photons. Such interactions break the principle of superposition that arises from the linearity of Maxwell's equations and thus are phenomenologically significant only at extreme scales.

These models have been extensively studied at the classical level to solve problems in cosmology and supergravity as well as in strongly coupled condensed matter systems where photon self-interactions would contribute significantly \cite{Sorokin_2022}. However, they remain largely unstudied in the quantum domain due to the difficulty the nonlinearity introduces to quantization procedures.

There are a number of ways to extend electrodynamics while preserving different properties of interest within classical electromagnetism. The prototypical example of such an extension is the Born-Infeld theory, proposed in 1934 to solve the infinite self-energy of an electron \cite{Born1934, Born1934_2, Born1935_3, bb1988}. Born and Infeld achieved this by introducing a maximum possible electric field strength in their theory. However this modification introduces a characteristic energy scale (the maximum field strength) which breaks the scale invariance present in Maxwell electrodynamics. Note that Born-Infeld theory still preserves a symmetry of Maxwell's equations called electromagnetic duality.

It was long thought that there were no possible extensions to Maxwell electrodynamics that would preserve both of the notable present symmetries: scale invariance and electromagnetic duality. In recent literature \cite{Bandos_2020, Sorokin_2022, Lechner_2022} however, a novel modification to Maxwell's theory of electromagnetism (electrodynamics) was discovered that achieves this: \textit{Modified Maxwell electrodynamics} or \textit{ModMax} for short. It was further proved that ModMax is the only theory that achieves this, namely, it is the unique nonlinear extension that preserves both of the notable symmetries of Maxwell's original theory: conformal invariance and electromagnetic duality. The beauty inherent in the unique preservation of these symmetries aside, ModMax is also of particular interest as such symmetries can lead to novel observable implications when the theory is quantized.  

Additionally, as we expect such symmetries to be respected in classical limits, it is of great interest whether such symmetries are fundamental or broken at the quantum level. While the domain of effect of such extensions is beyond current experimental techniques, the presence of nonlinear effects represents a conceptual shift in how we describe electromagnetism worthy of our study.

\section{Introduction}

In this thesis, we perform the perturbative quantization of ModMax and calculate such first quantum corrections that arise within this theory. We also generalize our argument to other higher derivative theories in $1+1$ spacetime dimensions.

The process of quantization, the transfer of a classical theory to the quantum domain, begins with formulating a Lagrangian, an object which completely specifies the theory and the equations of motion it predicts. If the Lagrangian describes non-interacting particles, then it is often able to be solved exactly for the equations of motion. However, among interacting theories, very few admit an exact solution and thus we must employ the use of perturbation theory \cite{Schroeder,GA}. To achieve this, we consider the interaction to have a small effect relative to the free evolution of the particles and expand in an increasing number of interactions. This is well suited to nonlinear theories of electrodynamics where the additional interaction term is separable, as the strength of the interaction is necessarily extremely small due to lack of classical observation.

However, due to the nonlinearity of ModMax and the non-analytic nature of the square root it contains, perturbation theory alone cannot yield a quantum version of this theory. Therefore, we must expand the interaction term about some non-zero background, rather than the vacuum. In most nonlinear field theories, the weak field limit reduces to Maxwell's equations, which would allow one to truncate higher powers of fields in such an expansion immediately. However, ModMax only reduces to Maxwell's theory in the \textit{non-interacting limit}, namely when the strength of the self-interaction goes to zero, which is distinct. Instead, we make use of the \textit{background field method}, where we expand about a fixed classical background field. This fixed non-zero classical background provides a valid point to expand about, and will reduce correctly in the limit to Maxwell's equations.

\section{Classical Electromagnetism}

Beginning with Maxwell's theory, Maxwell's Lagrangian is expressible as
\begin{align}
    \mathcal{L} = -\frac{1}{4}F^{\mu \nu} F_{\mu \nu}
,\end{align}
where the field strength $F_{\mu \nu}$ is defined by
\begin{align}
    F_{\mu \nu} &= \partial_\mu A_\nu - \partial_\nu A_{\mu}
,\end{align}
for four vector potential $A_{\mu}$, with $\mu = 0,1,2,3$. We can also write the electric and magnetic fields explicitly with derivatives of this potential such that for $i = 1,2,3$ we have
\begin{align}
    E_i = \partial_0 A_i - \partial_i A_0, && B_i = -\epsilon_{ijk} \partial^j A^k
,\end{align}
where we use Einstein notation in which summation over repeated indices is implied.

\begin{definition}
    The \textbf{Hodge dual} of the field strength tensor $F_{\mu \nu}$ is defined as
    \begin{align}
        \widetilde{F}^{\mu \nu} = \frac{1}{2} \epsilon^{\mu \nu \rho \tau} F_{\rho \tau}
    ,\end{align}
    where $\epsilon^{\mu \nu \rho \tau} = -\epsilon^{\mu \nu \tau \rho}$ is the Levi-Civita symbol that is antisymmetric under all index exchanges.
\end{definition}

Applying the Euler-Lagrange equation,
\begin{align}
    \partial_\mu \pdv{\mathcal{L}}{\left( \partial_\mu A_\nu \right) } &=  \pdv{\mathcal{L}}{A_\nu}
,\end{align}
leads to the equations of motion
\begin{align}
    \partial_\mu F^{\mu \nu} = 0 && \partial_\mu \widetilde{F}^{\mu \nu} = 0,
    \intertext{which can be written in the more familiar form}
    \pdv{\vb{E}}{t} = \grad \times \vb{B}, &&
    \pdv{\vb{B}}{t} = - \grad \times \vb{E}, \nonumber\\
    \grad \cdot \vb{B} = 0, &&
    \grad \cdot \vb{E} = 0 
,\end{align}
which are the (sourceless) Maxwell's equations. 

\section{Symmetries of Maxwell's Equations}
Maxwell's equations have two symmetries of note that are preserved uniquely by ModMax: electromagnetic duality and conformal invariance.

One can notice that under an $SO\left( 2 \right) \simeq U(1)$ transformation (i.e. a 2D rotation) parametrised by an angle $\alpha \in [0,2\pi)$,
\begin{align}
    \mqty( F'^{\mu \nu} \\ \widetilde{F}'^{\mu \nu} )  &=  \mqty( \cos \alpha & \sin \alpha \\ - \sin \alpha & \cos \alpha ) \mqty( F^{\mu \nu} \\ \widetilde{F}^{\mu \nu} ) \label{so2}
,\end{align}
that Maxwell's equations of motion are invariant. This is a symmetry called \textit{electromagnetic duality} (EM-duality) that Maxwell's theory possesses. Notice that this symmetry holds only \textit{on-shell}, that is, it occurs when the equations of motion are applied. ModMax preserves this symmetry at this level as well \cite{Sorokin_2022,Lechner_2022}.

Writing the field strength tensor and its dual explicitly we see the element wise exchange (up to sign) of electric and magnetic fields with
\begin{align}
    F^{\mu \nu} = \mqty( 0 & -E_x & -E_y & -E_z \\ E_x & 0 & -B_z & B_y \\ E_y & B_z & 0 & -B_x \\ E_z & -B_y & B_x & 0 ), && \widetilde{F}^{\mu \nu} = \mqty( 0 & -B_x & -B_y & -B_z \\ B_x & 0 & E_z & -E_y \\ B_y & -E_z & 0 & E_x \\ B_z & E_y & -E_x & 0 )
.\end{align}

Additionally, Maxwell's theory has no dependence on a length or energy scale and thus has a global symmetry of scale invariance. With the addition of special conformal transformations, this becomes \textit{conformal invariance}. This is a global symmetry of the Lagrangian, not just the equations of motion which is important as it allows us to apply Noether's theorem and derive conserved quantities. This group of symmetries includes all transformations that preserve angles and thus includes the Poincar\'e group (which is Lorentz transformations and 3D spatial rotations), dilations (zooming in/out) and special conformal transformations. The latter two transform the coordinates according to
\begin{align}
    x^{\mu} &\to \lambda x^{\mu}, \nonumber \\
    x^{\mu} &\to \frac{x^{\mu} - \lambda^{\mu} x^2}{1 -2 \lambda_\mu x^{\mu}+ \lambda^2 x^2}
,\end{align}
where $\lambda^{\mu}$ parametrizes the transformation.

More generally, in Minkowski space where the metric is given by
\begin{align}
    \eta_{\mu \nu} = \mqty( 1 & 0 & 0 & 0 \\ 0 & -1 & 0 & 0 \\ 0 & 0 & -1 & 0 \\ 0 & 0 & 0 & -1 )
,\end{align}
we can write an infinitesimal distance as
\begin{align}
    \dd{^2s} = \eta_{\mu \nu} \dd{x}^{\mu} \dd{x}^{\nu}
.\end{align}
Under a conformal transformation parametrized by $\Omega \left( x \right) $, this distance transforms as
\begin{align}
    \dd{^2s'} = e^{\Omega\left( x \right) }\dd{^2s}
,\end{align}
which preserves the relative angles of vectors (as shown in \cref{fig:conformal}).

%\begin{figure}[h]
%    \centering
%    \includegraphics[width=\figwidth/2]{figures/cgb.png}
%    \begin{tikzpicture}[baseline=-22.0ex]
%        \draw[->] (0,0) -- (1,0);
%    \end{tikzpicture}
%    \includegraphics[width=\figwidth/2]{figures/cga.png}
%    \caption{A special conformal transformation of a grid. Notice that the right angle intersections of all grid lines is preserved after the transformation.}
%    \label{fig:cg}
%\end{figure}

% figures/conformal2.tex
\begin{figure}[h]
    \centering
    \includegraphics[width=\figwidth]{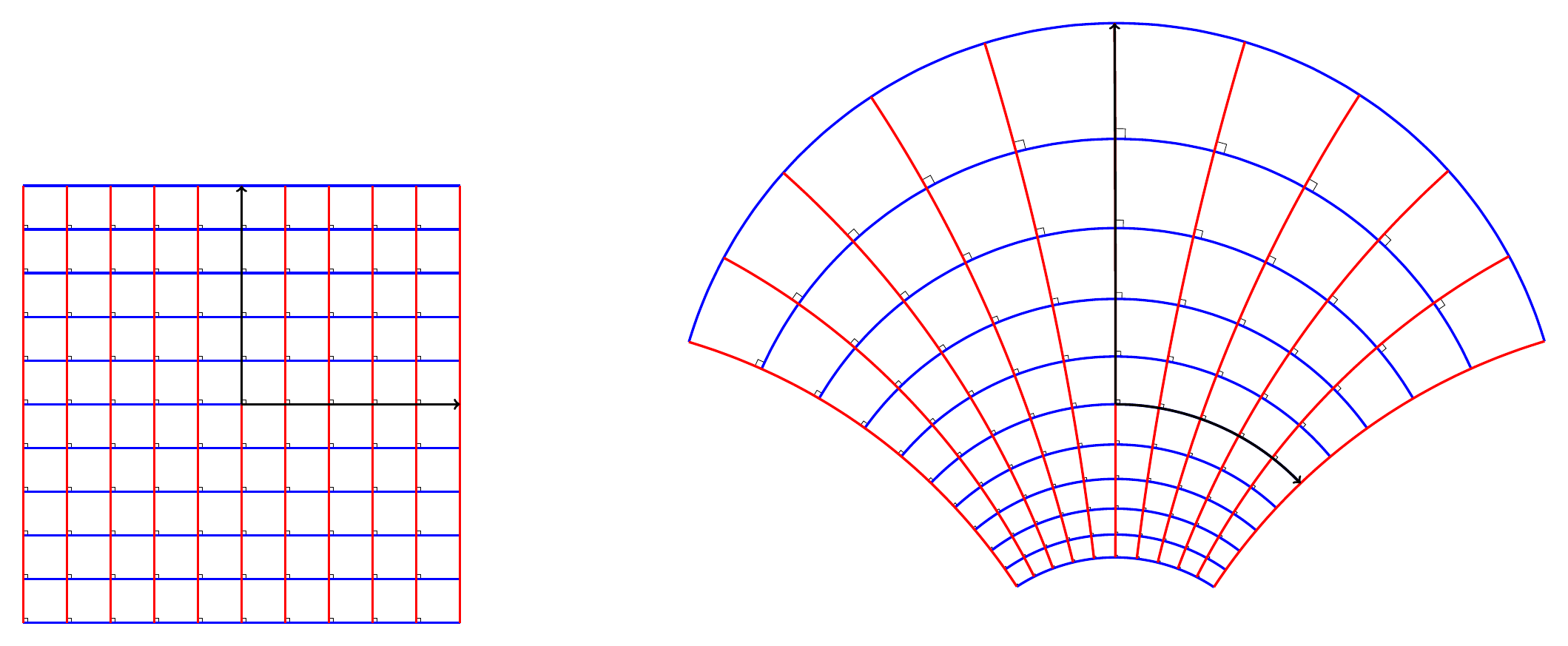}
    \caption{A special conformal transformation of a grid. Notice that the right angle intersections of all grid lines is preserved after the transformation.}
    \label{fig:conformal}
\end{figure}

\begin{note}
    The special conformal transformation can also be interpreted as a coordinate inversion composed with a translation and then a second coordinate inversion.
\end{note}

\section{General Theories of Nonlinear Electrodynamics}

When investigating extensions to Maxwell's Lagrangian, there are strong restrictions on the form of the candidates and their nonlinearity. The theory must be Lorentz invariant to agree with experimental observations, and thus must be built out of Lorentz invariant operators. The only two independent Lorentz invariant combinations of operators that can be made with the field strength are Maxwell's Lagrangian
\begin{align}
    S &\equiv -\frac{1}{4}F_{\mu \nu} F^{\mu \nu} = \frac{1}{2} \left( \vb{E}^2 - \vb{B}^2 \right)  
    \intertext{which is a scalar, and}
    P &\equiv -\frac{1}{4}F_{\mu \nu} \widetilde{F}^{\mu \nu} = \vb{E} \cdot \vb{B}
,\end{align}
which is a pseudo-scalar (i.e. negates under parity transformations). All higher order combinations of $F^{\mu \nu}$, such as $F^{\mu \nu} \tensor{F}{_{\nu}^{\rho}} F_{\rho \mu}$, are expressible in terms of these two invariants, and terms including more derivatives, such as $\partial^{\mu} S \partial_\mu S$ can lead to systems associated with unphysical `ghost fields' \cite{Sorokin_2022,Ferko_2022}. 

Therefore, the most general form of a nonlinear electrodynamics Lagrangian we focus on is some function of these quantities, $\mathcal{L}\left( S,P \right) $. Note that this form contains no restrictions on the symmetries of the theory and in general can break both of the notable symmetries of Maxwell's equations: conformal symmetry and electromagnetic duality. Note that as Maxwell's theory is a free theory, it has infinitely many symmetries. However, in an effort to narrow the search scope, we can identify that if we want our nonlinear extension to maintain conformal invariance, it must transform under a rescaling (by a constant $a$) of $S \to a^{-4} S$ and $P \to a^{-4} P$ as
\begin{align}
    \mathcal{L} \left( a^{-4} S, a^{-4}P \right) &= a^{-4} \mathcal{L}\left( S, P \right)
.\end{align}
This factor of $a^{-4}$ is cancelled by the transformation of $\dd{^{4}x} \to a^{4}\dd{^{4}x}$ in the action integral to leave the theory invariant under this transformation.

Likewise, if we want our nonlinear extension to maintain electromagnetic duality, its equations of motion should be invariant under the generalization of the $SO\left( 2 \right) $ rotation we saw in \cref{so2} for Maxwell's equations,
\begin{align}
    \mqty( -2 \pdv{\mathcal{L}\left( F' \right) }{F'_{\mu \nu}} \\ \widetilde{F}'_{\mu \nu} ) &= \mqty( \cos \alpha & \sin \alpha \\ - \sin \alpha & \cos \alpha ) \mqty( -2  \pdv{\mathcal{L}}{F_{\mu \nu}} \\ \widetilde{F}_{\mu \nu} )
.\end{align}

\section{Classical ModMax}

In recent work \cite{Bandos_2020}, it was shown that there is a unique family of Lagrangians which satisfy these two constraints and thus preserve these two symmetries of Maxwell's equations. This family of Lagrangians is ModMax, and is a family of theories as it satisfies these conditions for any real value of a dimensionless constant $\gamma \in \R$ that parametrizes the family.

As expected, we can write the ModMax Lagrangian in terms of the two Lorentz invariants $S$ and $P$ with
\begin{align}
    \mathcal{L} &= S \cosh \gamma + \sinh \gamma \sqrt{S^2 + P^2},
    \intertext{or equivalently, using the definitions of $S$ and $P$,}
    \mathcal{L} &= -\frac{\cosh \gamma}{4} F_{\mu \nu} F^{\mu \nu} + \frac{\sinh \gamma}{4} \sqrt{\left( F_{\mu \nu} F^{\mu \nu} \right)^2 + \left( F_{\mu \nu} \widetilde{F}^{\mu \nu} \right)^2  }  \nonumber \\
    &= -\frac{\cosh \gamma}{2} \left( \vb{E}^2 - \vb{B}^2 \right)  + \frac{\sinh \gamma}{2} \sqrt{\left( \vb{E}^2 - \vb{B}^2 \right)^2 + 4\left( \vb{E} \cdot \vb{B}\right)^2  }
,\end{align}
where $\gamma$ is interpreted as the dimensionless coupling constant that determines the strength of the nonlinear self-interaction in the second term. Notice that when $\gamma = 0$ we recover Maxwell's Lagrangian as the nonlinear term disappears as $\sinh \left( 0 \right) = 0$.

While this family of Lagrangians possesses the notable symmetries of Maxwell's equations (conformal symmetry and electromagnetic duality) for all values of $\gamma \in \R$, for $\gamma < 0$ the theory predicts faster than light propagation of photons which violate causality. Therefore, we take $\gamma > 0$ for which causality is preserved \cite{Sorokin_2022,Bandos_2020}.

Applying the Euler-Lagrange equations, we find the equations of motion of the theory to be
\begin{align}
    0 &= \cosh \gamma \partial_\mu F^{\mu \nu} + \sinh \gamma \partial_\mu \left( \frac{SF^{\mu \nu} + P \widetilde{F}^{\mu \nu}}{\sqrt{S^2 + P^2}} \right) \nonumber\\
    \implies \partial_\mu F^{\mu \nu} &= \tanh \gamma \partial_\mu \left( \frac{SF^{\mu \nu} + P \widetilde{F}^{\mu \nu}}{\sqrt{S^2 + P^2}} \right)
.\end{align}

Notice that while these equations are nonlinear, if the field satisfies $P = aS$ for constant $a \in \R$, then they linearise.

Further, as ModMax preserves the conformal symmetry of Maxwell's equations, it is also invariant under conformal transformations. However, more commonly, one makes use of the fact that conformal invariance implies that the stress energy tensor of the theory is traceless.

\begin{definition}
    Given a Lagrangian $\mathcal{L}$, the \textbf{stress energy tensor} \cite{ferko2023stress} is defined by
    \begin{align}
        T_{\mu \nu} &= -2 \pdv{\mathcal{L}}{g^{\mu \nu}} + g_{\mu \nu} \mathcal{L}
    .\end{align}
\end{definition}

\begin{note}
    If $\mathcal{L}$ possesses conformal invariance, then the stress energy tensor is traceless such that $\tensor{T}{^\mu_\mu} = 0$.
\end{note}

Applying this to ModMax, we find that the stress energy tensor can be written as
\begin{align}
    T_{\mu \nu} &= -2 \left( \pdv{\mathcal{L}}{S} \pdv{S}{g^{\mu \nu}} + \pdv{\mathcal{L}}{P} \pdv{P}{g^{\mu \nu}} \right) + g_{\mu \nu} \mathcal{L}
,\end{align}
where
\begin{align}
    \pdv{S}{g^{\mu\nu}} = -\frac{1}{2} \tensor{F}{_{\mu}^{\rho}}F_{\nu \rho}, && \pdv{P}{g^{\mu\nu}} = -\frac{1}{4} \left( \tensor{F}{_{\mu}^{\rho}} \widetilde{F}_{\nu \rho} + \tensor{F}{_{\nu}^{\rho}} \widetilde{F}_{\mu \rho} \right) 
.\end{align}

In fact we have
\begin{align}
    T^{\mu \nu} &= \left( \tensor{F}{^{\mu}_\rho} F^{\nu \rho} - \frac{1}{4} \eta^{\mu \nu} \left( F_{\rho \tau} F^{\rho \tau} \right)  \right) \pdv{\mathcal{L}}{S}\label{eq:set},
    \intertext{where}
    \pdv{\mathcal{L}}{S} &= \cosh \gamma - \sinh \gamma \frac{F_{\mu \nu}F^{\mu \nu}}{\sqrt{\left( F_{\mu \nu}F^{\mu\nu} \right)^2 + \left( F_{\mu \nu} \widetilde{F}^{\mu\nu} \right)^2} } 
.\end{align}

From \cref{eq:set}, we see immediately that $\tensor{T}{^{\mu}_{\mu}} = 0$, and thus ModMax is conformal.

\section{Experimental Observability of ModMax}

Despite the lack of study at the quantum level, classical analysis of ModMax indicates that it predicts a refractive index of the vacuum $n \neq 1$\cite{Sorokin_2022}. This is not unexpected as the nonlinearities that arise within the standard model also predict a vacuum refractive index differing from $n = 1$.

The most precise experimental test of the nonlinearity of the vacuum was recently conducted by the PVLAS experiment (\textit{Polarizzazione del Vuoto con LASer}, ``polarization of the vacuum with laser''). Using a cavity with mirrors, their experiment attempts to observe any interaction of light with itself or with the vacuum (i.e. spontaneous pair production). Due to the extremely small scale of any present nonlinearity, they were able to obtain an upper bound of $\gamma \leq 3 \times 10^{-22}$ with lower bound experiments currently underway \cite{ejlli2020pvlas}. This suggests that if ModMax is an accurate description of our universe's electrodynamics, its nonlinear contribution is very small.

Specifically, the PVLAS experiment proceeds by measuring the birefringence of the vacuum through how far its refractive index differs from $n = 1$  with ModMax predicting
\begin{align}
    \Delta n_{\text{ModMax}} = e^{\gamma} - 1 \approx \gamma + \mathcal{O}\left( \gamma^2 \right) 
.\end{align}

While PVLAS observed a difference in refractive index of $\Delta n_{\text{Obs}} \leq 3 \times 10^{-22}$ \cite{ejlli2020pvlas}, QED predicts $\Delta n_{\text{QED}} \sim 4 \times 10^{-24}$ \cite{Sorokin_2022} and thus there remains two orders of magnitude of parameter space for ModMax to have significant observable contributions.

%%%%% Other chapters in here
\chapter{Review: Quantization of Quantum Electrodynamics}

Prior to considering ModMax in a quantum context, we review the quantization of \textit{quantum electrodynamics} (QED) and the techniques applicable to such a theory: with and without interactions with matter. We also review the effective action which will prove instrumental in the quantization of ModMax.

\section{Quantization of Maxwell's Lagrangian}

Maxwell's Lagrangian translates smoothly to its quantum counterpart. Considering only free photons (i.e. no matter/electrons), we can describe electromagnetic waves with an identical Lagrangian to Maxwell's with
\begin{align}
    \mathcal{L} &= -\frac{1}{4} F_{\mu \nu} F^{\mu \nu}
    \intertext{where this leads to an action $S$ given by}
    S\left[ A \right]  &= \int \dd{^{4}x} \mathcal{L}  \\
     &= \frac{1}{2} \int \dd{^{4}x} A_{\mu}\left( x \right)  \left( \partial^2 g^{\mu \nu} - \partial^{\mu} \partial^{\nu} \right) A_{\nu}\left( x \right).
     \intertext{Note that with the Fourier transform of this field, given by}
     A_{\mu}\left( k \right) &= \int \dd{^{4}x} e^{i k_\nu x^{\nu}} A_{\mu} \left( x \right),
     \intertext{the action integral can be written as}
     S\left[ A \right] &= \frac{1}{2} \int \dd{^{4}x} \dd{^{4}k_1} \dd{^{4}k_2} e^{k_1^{\alpha} x_{\alpha}}A_{\mu}\left( k_1 \right) \left( \partial^2 g^{\mu \nu} - \partial^{\mu} \partial^{\nu} \right) e^{i k_2^{\beta} x_{\beta}} A_{\nu}\left( k_2 \right),  \nonumber
     \intertext{where evaluating the derivatives on the exponential yields}
    S\left[ A \right]  &= \frac{1}{2} \int \dd{^{4}x} \dd{^{4}k_1} \dd{^{4}k_2} e^{k_1^{\alpha} x_{\alpha}}A_{\mu}\left( k_1 \right) \left( -k_2^2 g^{\mu \nu} + k_2^{\mu} k_2^{\nu} \right) e^{i k_2^{\beta} x_{\beta}} A_{\nu}\left( k_2 \right),  \nonumber \\
    \intertext{and grouping the exponentials gives}
    S\left[ A \right] &= \frac{1}{2} \int \dd{^{4}x} \dd{^{4}k_1} \dd{^{4}k_2} A_{\mu}\left( k_1 \right) \left( -k_2^2 g^{\mu \nu} + k_2^{\mu} k_2^{\nu} \right) e^{i \left( k_1 + k_2 \right) \cdot x  } A_{\nu}\left( k_2 \right)  \nonumber
    \intertext{reveals that with $\displaystyle \delta^{4}\left( k_1 + k_2 \right) = \int \dd{^{4}x} e^{i \left( k_1 + k_2 \right) \cdot x}$, integrating over $x$ yields}
   S\left[ A \right]  &= \frac{1}{2} \int \dd{^{4}k_1} \dd{^{4}k_2} A_{\mu}\left( k_1 \right) \left( -k_2^2 g^{\mu \nu} + k_2^{\mu} k_2^{\nu} \right) A_{\nu}\left( k_2 \right) \delta^{4} \left( k_1 + k_2 \right), \nonumber
    \intertext{where integrating over $k_2$ absorbs the $\delta$ function and enforces $k \equiv k_1 = -k_2$ leaving}
     S\left[ A \right] &= \frac{1}{2} \int \dd{^{4}k}  A_{\mu}\left( k \right) \left( -k^2 g^{\mu \nu} + k^{\mu} k^{\nu} \right) A_{\nu}\left( -k \right) \label{eq:momentum_maxwell_action}
.\end{align}
This form of the action will prove useful.

\section{Functional Integrals}

In quantum field theory, there is an analogue of the partition function $Z$ from statistical mechanics called the generating functional $Z\left[ J \right] $ which depends on an arbitrary external source $J\left( x \right) $. $J\left( x \right) $ is the analogue of an external magnetic field $B(x)$. The generating functional is a convenient albeit abstract method to determine correlation functions.

\begin{definition}
    For the electromagnetic field, the \textbf{generating functional} \cite{Schroeder, GA} is given by
    \begin{align}
        Z\left[ J \right] &= \int \mathcal{D}A \exp \left({i S\left[ A \right] + i \int \dd{^{4}x} J_{\mu}\left( x \right) A^{\mu}\left( x \right) }\right)
    ,\end{align}
    where $\displaystyle\int \mathcal{D}A$ is a \textit{functional integral}, that is, it integrates over all possible functions or field configurations $A^{\mu}\left( x \right)$ can take. One can think of this as the continuous analogue of a sum over all possible states that a system can take as in the partition function.
\end{definition}

Namely, taking derivatives of the generating functional yields \textbf{correlation functions} such that
\begin{align}
    \left\langle A_{\mu_1}\left( x_1 \right) \cdots A_{\mu_n}\left( x_n \right) \right\rangle &= \left( -i \right)^{n} \frac{\delta^n Z\left[ J \right] }{\delta J_{\mu_1}\left( x_1 \right) \cdots \delta J_{\mu_n}\left( x_n \right) } \bigg|_{J = 0}
,\end{align}
where the $n$th order correlation function $\left\langle A_{\mu_1}\left( x_1 \right) \cdots A_{\mu_n}\left( x_n \right) \right\rangle$ can be used to obtain probability amplitudes for a given interaction or decay process. One can also represent correlation functions in terms of Feynman diagrams as we will see.

However, one notices that the action $S\left[ A \right]$ in \cref{eq:momentum_maxwell_action}  vanishes for all potentials $A_{\mu}\left( k \right) = k_{\mu} \alpha \left( k \right)$ where $\alpha\left( k \right)$ is any scalar function as
\begin{align}
S\left[ \alpha(x) k \right] &= -\frac{1}{2} \int \dd{^{4}k} \alpha^2(k) k_{\mu} \left( -k^2 g^{\mu \nu} + k^{\mu} k^{\nu} \right) k_{\nu}    \\
 &=  -\frac{1}{2} \int \dd{^{4}k} \alpha^2(k) \left( -k^{4} + k^{4} \right)  \\
 &= 0
.\end{align}

This is problematic for the theory as the partition function evaluated at $J = 0$ (i.e. no external source) leads to
\begin{align}
    Z\left[ 0 \right] &= \int \mathcal{D}A \exp \left( {i S\left[ A \right] } \right) \\
                      &= \int \mathcal{D}A e^{0} \\
    &= \int \mathcal{D}A 1
,\end{align}
which diverges as there are uncountably infinite different possible field configurations $A^{\mu}\left( x \right)$ can take. This divergence in fact arises due to a lack of uniqueness in this theory's description of a given physical field configuration.
\begin{claim}
    Namely, one can shift $A_{\mu}\left( x \right) $ by
    \begin{align}
        A_{\mu}\left( x \right) &\to A_{\mu}\left( x \right) + \partial_\mu \alpha \left( x \right) \label{eq:gauge_1}
    ,\end{align}
    for an arbitrary function $\alpha \left( x \right) $, without changing the physical implications of the theory. This shift is called a \textbf{gauge transformation}.
\end{claim}

\begin{note}
        We suppress the $x$ dependence of $\alpha$ for brevity.
\end{note}

\begin{proof}
    The field strength tensor, $F_{\mu \nu} = \partial_\mu A_\nu - \partial_\nu A_\mu$ transforms under this shift $A_{\mu} \to A_{\mu} + \partial_\mu \alpha$ to
    \begin{align}
        F_{\mu \nu} \to F'_{\mu \nu} &= \partial_\mu \left( A_\nu + \partial_\nu \alpha \right) - \partial_\nu \left( A_\mu + \partial_\mu \alpha \right),
        \intertext{where as partial derivatives commute, yields}
                   F'_{\mu \nu} &= \partial_\mu A_\nu - \partial_\nu A_\mu \\
                   F'_{\mu \nu} &= F_{\mu \nu}
    .\end{align}
    As the equations of motion, or equivalently the electric and magnetic fields can be written in terms of $F_{\mu \nu}$, the field configuration $A_{\mu} + \partial_\mu \alpha$ has identical physical implications and dynamics to $A_{\mu}$.
\end{proof}

More generally we notice that the action itself is invariant under this transformation (as it can be written purely in terms of $F_{\mu \nu}$ as $S = -\frac{1}{4}F_{\mu \nu} F^{\mu \nu}$) such that
\begin{align}
    S\left[ A_{\mu} \right] =  S\left[ A_{\mu} + \partial_\mu \alpha \right] 
.\end{align}

This is referred to as a \textbf{gauge degree of freedom}, and is remedied by \textit{fixing the gauge} which means we only count physically distinct states.

\section{Faddeev-Popov Gauge Fixing}

The cleanest way to achieve this gauge fixing in a path integral approach is through a method pioneered by Faddeev and Popov \cite{Faddeev}. 

\begin{definition}
    A Lagrangian $\mathcal{L}\left[ A \right] $ has \textbf{local gauge symmetry} if it is invariant under a gauge transformation
    \begin{align}
        A_{\mu}\left( x \right) \to A_{\mu} + \partial_\mu \alpha \left( x \right) 
    ,\end{align}
    where $\alpha \left( x \right) $ is again an arbitrary function. Local here refers to the spacetime dependence of $\alpha \left( x \right) $. If it were constant $\alpha \left( x \right) = C \in \R$, it would be a \textbf{global symmetry}.
\end{definition}

To fix the gauge, we define a gauge fixing function $G\left( A \right) $ that is zero for only one of every physical/gauge inequivalent state. It can be chosen to take the form
\begin{align}
    G\left( A \right) \equiv \partial_{\mu} A^{\mu} - \omega\left( x \right) 
,\end{align}
such that $G\left( A \right) = 0$ for only $\partial_\mu A^{\mu} = \omega\left( x \right) $. If the functional integral contained $\delta \left( G\left( A \right)  \right) $, then this would select only unique physical states.

\begin{note}
    Under composition by a function $g : \R^{n} \to \R^{n}$, the delta function satisfies
    \begin{align}
        \int \dd{^{n}x} \delta^{n} \left( g\left( x \right)  \right) f\left( g\left( x \right)  \right) \det \left( \partial_j g_i \right) = \int \dd{^{n}x}\delta^{n} \left( x \right) f\left( x \right) 
    .\end{align}
    If we integrate over a space of functions $\mathcal{D}\alpha$ (which is infinite dimensional) rather than $\R^{n}$, the analogous identity is
    \begin{align}
        \int \mathcal{D}\alpha \delta \left( g \left( \alpha \right)  \right) f\left( g\left( \alpha \right)  \right) \det \left( \fdv{g\left( \alpha \right)}{\alpha} \right)  &= \int \mathcal{D}\alpha \delta\left( \alpha \right) f\left( \alpha \right)
    .\end{align}
    Choosing $f\left( \alpha \right) = 1$, this reduces to
    \begin{align}
        \int \mathcal{D}\alpha \delta \left( g \left( \alpha \right)  \right)  \det \left( \fdv{g\left( \alpha \right)}{\alpha} \right)  &= 1 \label{eq:func_indentity}
    .\end{align}
\end{note}

This identity appears promising and indeed, we can insert it into the functional integral by multiplying by 1, and the delta function will select only physical states. The only cost is the introduction of the determinant. As $g$ must be a function of $\alpha$ we choose 
\begin{align}
    g\left( \alpha  \right) &\equiv G\left( A^{\mu} + \partial^\mu \alpha \left( x \right)  \right)\\
    &= \partial_\mu \left( A^{\mu} + \partial^\mu \alpha \left( x \right)  \right) - \omega \left( x \right)
,\end{align}
and insert the left hand side of \cref{eq:func_indentity} into the functional integral yielding
\begin{align}
    Z\left[ 0 \right] &= \int \mathcal{D}A \exp \left( i S\left[ A \right]  \right) \left( \int \mathcal{D} \alpha \delta \left( g\left( \alpha \right)  \right) \det \left( \fdv{g\left( \alpha \right)}{\alpha} \right)  \right).
    \intertext{From the definition of $g\left( \alpha \right) $, one can see that $\displaystyle\fdv{g\left( \alpha \right) }{\alpha} = \partial_\mu \partial^{\mu}$ is independent of $\alpha$ and thus can be factored out as a constant. We then see}
    Z\left[ 0 \right] &= \det \left( \partial^2 \right) \int \mathcal{D}A \exp \left( i S\left[ A \right]  \right) \left( \int \mathcal{D} \alpha \delta \left( G\left( A^{\mu} + \partial^{\mu}\alpha \right)  \right)   \right).
    \intertext{As $S\left[ A_{\mu} \right] = S\left[ A_{\mu} + \partial_{\mu} \alpha \right] $ since it is gauge invariant, and $\mathcal{D}A_{\mu} = \mathcal{D}\left( A_{\mu} + \partial_{\mu} \alpha \right) $ because the space of functions is similarly invariant, we notice that we can write $Z\left[ 0 \right] $ purely in terms of $A_{\mu} + \partial_\mu \alpha$. However, as it is an integration variable, we can substitute back $A_{\mu} + \partial_\mu \alpha \to A_{\mu}$, removing all explicit $\alpha$ dependence to find}
    Z\left[ 0 \right] &= \det \left( \partial^2 \right) \left( \int \mathcal{D} \alpha    \right) \int \mathcal{D}A \exp \left( i S\left[ A \right]  \right) \delta \left( G\left( A \right)  \right),
    \intertext{namely, that the $\mathcal{D}\alpha$ integral factors out, amounting to an infinite constant (without physical implications), and that we have obtained the $\delta \left( G\left( A \right)  \right) = \delta \left( \partial_\mu A^{\mu} - \omega\left( x \right)  \right)  $ desired to select only physical states. As this expression holds for any fixed function $\omega \left( x \right) $, it holds identically for any normalized linear sum of $\omega\left( x \right) $'s. Faddeev and Popov's essential insight was to integrate over a normalized Gaussian weighted envelope in the function space of $\omega\left( x \right)$ with variance $\xi$ such that}
    Z\left[ 0 \right] &= N\left( \xi \right) \int \mathcal{D} \omega \exp \left( -i \int \dd{^{4}x} \frac{\omega^2}{2\xi} \right) \int \mathcal{D}A \exp \left( i S\left[ A \right]  \right) \delta \left( \partial_\mu A^{\mu} - \omega  \right),
    \intertext{where $N\left( \xi \right) $ ensures the Gaussian is normalized, and absorbs the other constant factors for brevity. Evaluating the $\mathcal{D}\omega$ integral absorbs the delta function enforcing $\partial_\mu A^{\mu} = \omega$ and thus yields}
    Z\left[ 0 \right] &= N\left( \xi \right) \int \mathcal{D}A \exp \left( -i \int \dd{^{4}x} \frac{\left( \partial_\mu A^{\mu} \right) ^2}{2\xi} \right) \exp \left( i S\left[ A \right]  \right)
,\end{align}
where we see that the exponential term modifies the Lagrangian with the addition of a term of the form
\begin{align}
    \mathcal{L}_{\text{gauge fixed}} &= \mathcal{L} - \frac{1}{2\xi} \partial_{\mu} A^{\mu} \partial_\nu A^{\nu} \label{eq:gauge_fixed}
.\end{align}
where $\xi \in \R$ is referred to as the \textbf{gauge parameter} and can be fixed to any desired number. It can be shown that observable quantities will always be independent of your choice of $\xi$; however some choices more significantly simplify calculations.

This arduous derivation of gauge fixing applies not only to QED, but rather any abelian gauge theory, including ModMax. Namely, given an abelian gauge theory with Lagrangian $\mathcal{L}$, subtracting the gauge fixing term in \cref{eq:gauge_fixed} fixes the gauge, leading to well defined observables and behaviour.

Returning to the momentum space action in \cref{eq:momentum_maxwell_action}, we see that after gauge fixing, the action can be written as
\begin{align}
    S\left[ A \right] = \frac{1}{2} \int \dd{^{4}k} A_{\mu}\left( k \right) \left( -k^2 g^{\mu \nu} + \left( 1 - \frac{1}{\xi} \right) k^{\mu} k^{\nu} \right) A_{\nu}\left( -k \right) 
,\end{align}
which no longer vanishes for $A_\mu = \partial_\mu \alpha $, and thus leads to a well defined generating functional.

\section{The Effective Action}

Further drawing on the analogy with statistical mechanics, recall that the free energy $F \left( B \right) $ of a system dependent on a magnetic field $B$ can be obtained from the partition function $Z\left[ B \right] $ by
\begin{align}
    F \equiv - T\ln Z\left[ B \right] 
.\end{align}
Taking the derivative of $F$ with respect to $B$ then yields the magnetization $M$ of the system, 
\begin{align}
    M \equiv -\pdv{F}{B}
,\end{align}
from which the Gibbs free energy $G$ can be found by Legendre transforming $F$ such that
\begin{align}
    G \equiv F - BM
.\end{align}

Each of these thermodynamic quantities has an analogue within quantum field theory. Namely, the generating functional (which takes the place of the partition function) is defined in terms of an external source $J\left( x \right) $, which is the generalization of the external magnetic field $B$. As such, we can define the analogue of the free energy
\begin{align}
    E\left[ J \right] \equiv -i \fdv{J\left( x \right) } \ln Z\left[ J \right]
,\end{align}
for which a further derivative provides the analogue of the magnetization: the expectation value of the field $ C_{\mu} \equiv \left<A_{\mu}\left( x \right)  \right>$ such that
\begin{align}
    \fdv{J\left( x \right) }E\left[ J \right] = -\left<A_{\mu}\left( x \right)  \right> = -C_{\mu}
.\end{align}
Notice that we have taken $C_{\mu}$ to be independent of $x^{\mu}$ which follows from the assumption of translational invariance. This makes the expectation value of the field $C_{\mu}$ a global property (like magnetisation) which characterises the whole system.
\begin{definition}
    Lastly, Legendre transforming $E\left[ J \right] $, we obtain the \textbf{effective action}
\begin{align}
    \Gamma \left[ C \right] \equiv E\left[ J \right] - \int \dd{^{4}x} J^{\mu}\left( x \right) C_{\mu}\left( x \right)
,\end{align}
which is a functional depending only on $C_{\mu}$, the expectation value of the field. This action provides an \textit{effective} classical description of the full quantum theory, and thus is an invaluable tool to obtain observable quantities from otherwise unsolvable theories.
\end{definition}

Taking the derivative of the effective action with respect to $C_{\mu}$, it can be shown that
\begin{align}
    \fdv{C_{\mu}} \Gamma \left[ C_{\mu} \right] = -J\left( x \right)
,\end{align}
and thus in the sourceless case where $J\left( x \right) = 0$, we find
\begin{align}
    \fdv{C_{\mu}} \Gamma \left[ C_{\mu} \right]\bigg|_{J(x)=0} &= 0
.\end{align}
Namely, $C_{\mu} = \left< A_{\mu}\left( x \right)  \right>$ that solves this equation extremizes the action and thus corresponds to a stable solution $A_{\mu}\left( x \right) $ of the original theory. The effective action therefore allows us to study the large scale \textit{effective} behaviour when quantum effects are cumulatively taken into account.

\section{Propagators and Correlation Functions}

While we do not prove it here, the effective action can equivalently be obtained by evaluating all Feynman diagrams with classical external vertices constructable within a theory. One constructs and evaluates diagrams by obtaining the \textit{Feynman rules} of the theory: namely, the factors to include in the calculation for each possible line and vertex that make up a diagram. As we will see here, each line within a diagram represents the propagation of a particle, and each vertex an interaction.

\begin{definition}
    Given a field $A_{\mu}\left( x \right) $, the \textbf{propagator} of that field satisfies
    \begin{align}
        A_{\mu}\left( x \right)  &= \int \dd{^{4}y} D_{\mu \nu}\left( x-y \right)  A^{\nu}\left( y \right) 
    ,\end{align}
    in that it \textit{propagates} the field $A^{\nu}\left( y \right) $ to $x$ through all possible paths. 

    When we draw a Feynman diagram, each internal line represents a propagator corresponding to that field
    \begin{align}
        \vcenter{\hbox{\includegraphics{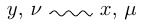}}}
        = D_{\mu \nu}\left( x-y \right)
    ,\end{align}
    propagating it from $y$ to $x$.
\end{definition}

With the above gauge fixing, the propagator for the electromagnetic field can now be found to satisfy
\begin{align}
    \left( -k^2 g^{\mu \nu} + \left( 1- \frac{1}{\xi} \right) k^{\mu} k^{\nu} \right) D_{\nu \rho}\left( k \right) &= i\delta^{\mu}_\rho
,\end{align}
which has solution
\begin{align}
    D_{\nu \rho}\left( k \right) &= \frac{-i}{k^2} \left( g_{\nu \rho} - \left( 1 - \xi \right) \frac{k_{\nu} k_{\rho}}{k^2} \right)
.\end{align}

Setting $\xi = 1$ here, as we are free to do so without affecting observable quantities, is referred to as \textit{Feynman gauge}, and simplifies the form of the propagator\footnote{Formally, to prevent a pole at $k^2 = 0$, one replaces the denominator with $k^2 + i \epsilon$ such that the pole is shifted into the complex plane, and the limit of $\epsilon \to 0$ is taken after contour integration.} greatly to
\begin{align}
    D_{\nu \rho}\left( k \right) &= \frac{-ig_{\nu \rho}}{k^2} 
.\end{align}

We will proceed with this choice of gauge for simplicity.

\section{QED Diagrams}

With the propagator for the photon obtained, one can now look at introducing electrons as described in QED. The full Lagrangian is of the form
\begin{align}
    \mathcal{L}_{\text{QED}} &= \underbrace{-\frac{1}{4} F_{\mu \nu} F^{\mu \nu} + \overline{\psi} \left( i \partial_{\mu} \gamma^\mu - m \right) \psi}_{\mathcal{L}_0} - e \overline{\psi} \gamma^{\mu} A_{\mu} \psi
,\end{align}
where $\psi \left( x \right) $ is a spinor field that describes electrons/positrons, $e$ is the electric charge, and $\gamma^{\mu}$ are the Dirac matrices that satisfy $\left\{ \gamma^{\mu}, \gamma^{\nu} \right\} = 2g^{\mu \nu}$.

The first two terms in this Lagrangian (together forming $\mathcal{L}_0$) describes how photons and electrons freely evolve, and the third determines their interaction.  Like most interacting quantum field theories, this Lagrangian is unsolvable in its exact form and thus we move to using perturbation theory, where we assume the interaction of the electrons and photons is relatively small. This is valid as the electric charge, $e \ll 1$, here determines the strength of this interaction. This allows us to Taylor expand the interaction term such that
\begin{align}
    \exp \left( i \int \dd{^{4}x} \mathcal{L} \right) &= \exp \left( i \int \dd{^{4}x}\mathcal{L}_0\right) \exp \left( i \int \dd{^{4}x} \left( - e \overline{\psi} \gamma^{\mu} A_{\mu} \psi \right)    \right),
    \intertext{with $\exp x \sim 1 + x + \mathcal{O}\left( x^2 \right) $,}
    \exp \left( i \int \dd{^{4}x} \mathcal{L} \right) &= \exp \left( i \int \dd{^{4}x}\mathcal{L}_0\right) \left( 1 - ie \int \dd{^{4}x}    \overline{\psi} \gamma^{\mu} A_{\mu} \psi  + \mathcal{O}\left( e^2 \right) \right) 
,\end{align}
where the $1$ term corresponds to the free evolution of particles, and the nontrivial term corresponds to an interaction vertex involving two fermions and a photon. The presence of a $\overline{\psi}$ and a $\psi$ indicates that a fermion enters and leaves this interaction, and the $A_{\mu}$ field indicates that a photon is involved. In fact, with the identification of these fields, the remaining factors in this term yield, in the language of Feynman diagrams, the form of the QED vertex

\begin{align}
% \feynmandiagram [baseline=(d.base), horizontal=d to b] {
%    a -- [fermion] b [dot] -- [fermion] c,
%b -- [boson,edge label=\(A^\mu\)] d 
%};
\raisebox{-0.5\height}
{\includegraphics[width=0.2\linewidth]{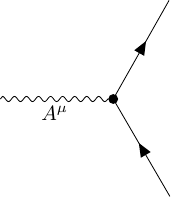}}
= i e \gamma^{\mu}
,\end{align}
where we draw photons as wavy lines and fermions as straight lines with arrows indicating the flow of charge. This factor is to be included in the evaluation of Feynman diagrams whenever this vertex appears.

If we assume that the photon field is an external classical field $C_\mu \left( x \right)$, that is the line exits the diagram, we instead obtain the vertex rule
\begin{align}
%\feynmandiagram [baseline=(d.base), horizontal=d to b] {
%    a -- [fermion] b [dot] -- [fermion] c,
%b -- [boson] d [crossed dot],
%};
\raisebox{-0.5\height}
{\includegraphics[width=0.2\linewidth]{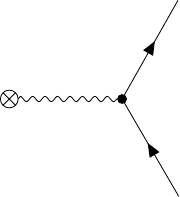}}
= i e \gamma^{\mu} \int \dd{^{4}x} C_{\mu}\left( x \right) 
.\end{align}

To obtain the effective action $\Gamma$, and thus the form of the largest quantum corrections, we need to construct all possible Feynman diagrams formed out of this vertex that contain one loop. Higher numbers of loops contribute less significantly. Observe that with classical action $S\left[ C_{\mu} \right]$, the one-loop corrections consists of the following infinite series of diagrams:
\begin{align}
%\feynmandiagram [horizontal=a to b, layered layout, baseline=(a.base)] {
%    i1 [crossed dot] -- [photon] a [dot]
%  -- [half left, fermion] b
%  -- [half left, fermion] a,
%}; + 
%\feynmandiagram [horizontal=a to b, layered layout, baseline=(a.base)] {
%    i1 [crossed dot]  -- [photon] a [dot]
%  -- [fermion, half left] b [dot]
%  -- [fermion, half left] a,
%  b -- [boson] f1 [crossed dot]
%}; +
%\feynmandiagram [small, horizontal=a to d, spring layout, baseline=(a.base)] {
%    i1 [crossed dot]  -- [photon] a [dot]
%  -- [fermion, quarter right] c [dot]
%  -- [fermion, quarter right] d [dot]
%  -- [fermion, quarter right] a,
%  c -- [photon] f1 [crossed dot],
%  d -- [photon] f2 [crossed dot]
%}; + \cdots
    \Gamma = S[C_\mu] + 
    \raisebox{-0.5\height}{\includegraphics{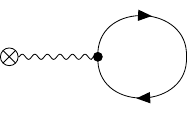}} ~ + ~
    \raisebox{-0.5\height}{\includegraphics{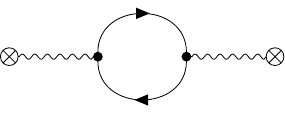}} ~ + ~
    \raisebox{-0.5\height}{\includegraphics{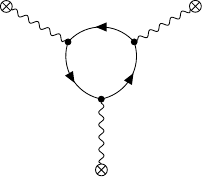}} ~ +  \cdots
.\end{align}

Evaluating this series of diagrams using the QED Feynman rules for the vertex and photon propagator (as derived above) yields an effective action, $\Gamma \left[ C_{\mu}\left( x \right)  \right]$, that can be used to derive physical observables. This action is referred to as \textit{effective} in two senses. First, this is a quantum effective action in that it is an effective classical description of the quantum theory as defined above, such that we have captured the nonlinearity of photons induced by quantum effects. Secondly, it is also a low energy effective action in that one integrates out the dependence on electrons/positrons, and thus this action breaks down for energies much greater than the electron mass (at which interactions such as pair production would occur). In what follows in this thesis, we use effective action in the former sense. This effective action was first derived by Julian Schwinger \cite{Schwinger1951} in 1951. This result was of particular interest as Schwinger also showed that this effective action predicts that at sufficiently high electric field strengths, the vacuum will `decay' by pair producing electrons and positrons. While the field strength required for pair production to occur is unreachable by modern experimentalist techniques, it is widely accepted as a valid prediction of QED.

It is an aim of this project to derive the analogous Feynman rules for ModMax, and to use them to evaluate a similar infinite series of diagrams exactly to find the effective action, $\Gamma\left[ C_{\mu}\left( x \right)  \right] $.

\chapter{ModMax in the Background Field Method}

With the quantization of QED detailed, we seek to quantize ModMax. However, as well will see, ModMax is not amenable to traditional quantization techniques due to the nonlinear and nonanalytic form of the self-interaction. Therefore, I quantize ModMax within the background field method. I show that the one loop effective action for backgrounds with constant field strength exactly vanishes using dimensional regularization as well as the two loop effective action to order $\gamma^2$. I further find that allowing the background to vary, logarithmic divergences emerge.

\section{Background Field Method}

Recall that the ModMax Lagrangian \cite{Sorokin_2022} is given by
\begin{align}
    \mathcal{L} &=  S \cosh \gamma + \sqrt{S^2 + P^2}  \sinh \gamma
,\end{align}
where we recall the respectively scalar and pseudoscalar invariants $S = -\frac{1}{4}F_{\mu \nu} F^{\mu \nu}$ and $ P = -\frac{1}{4} F_{\mu \nu} \widetilde{F}^{\mu \nu}$.

As $\gamma = 0$ recovers Maxwell's Lagrangian, $\mathcal{L} =  S$, the most natural interpretation of the ModMax Lagrangian is that the $S \cosh \gamma$ term provides a Maxwell-like free propagation of the photon, and the $\sqrt{S^2 + P^2} \sinh \gamma$ term is an interaction of the photon field with itself. This interaction is small relative to the free evolution, as $\sinh \gamma \overset{\mathcal{O}\left( \gamma \right)}{\sim}  \gamma \ll 1 \overset{\mathcal{O}\left( \gamma \right)}{\sim}  \cosh \gamma$.

Thus, as this is a self-interacting theory, one may desire to proceed using canonical quantization techniques as applied to QED. However, the nonlinear form of the interaction is not compatible with the familiar perturbative Feynman diagram expansion, where we require positive integer powers of the fields to proceed. The natural impulse is then to Taylor expand the interaction, assuming that the essential physics can be captured at low powers of the field, or equivalently by weak fields (as higher powers would be comparatively negligible). However, ModMax does not reduce to Maxwell's equations in the weak field limit, only in the $\gamma \to 0$ limit, and thus such an approach is ill-suited. Clearly an alternative approach is needed. 

As such, we begin by employing the \textit{background field method} in which we consider a fixed non-zero classical background with quantum fluctuations about this background \cite{Abbott} (See \cref{fig:BFM-splitting}). Mathematically, this is performed by taking the photon field $A_{\mu}$ and decomposing it into a classical field $C_{\mu}$ and a quantum field $a_{\mu}$ such that the quantum oscillations on the classical background are equivalent to the original field with
\begin{align}
    A_{\mu} &= C_{\mu} + a_{\mu}.
\end{align}
\begin{figure}[h]
    \centering
    \includegraphics[width=\figwidth]{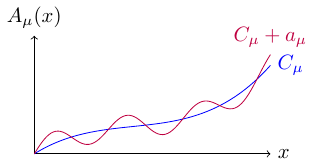}
    \caption{Pictorial depiction of background field method splitting of an arbitrary field $A_{\mu}\left( x \right) $ into a classical background field and a quantum oscillation about the background. The classical field here is varying spatially, but can also be taken to be constant.}
    \label{fig:BFM-splitting}
\end{figure}

This linear splitting of the field leads to a linear splitting in the field strength of
    \begin{align}
    F_{\mu \nu} &= \partial_\mu \left( C_\nu + a_\nu \right) - \partial_\nu \left( C_\mu + a_\mu \right)
    \intertext{into classical and quantum field strengths given by}
   F_{\mu \nu} &= \underbrace{\left( \partial_\mu C_\nu - \partial_\nu C_\mu \right)}_{C_{\mu \nu}} + \underbrace{\left( \partial_\mu a_{\nu} - \partial_\nu a_{\mu} \right)}_{f_{\mu \nu}}   \\
   F_{\mu \nu} &\equiv C_{\mu \nu} + f_{\mu \nu}
,\end{align}
where $C_{\mu \nu}$ is the field strength tensor for the classical field $C_{\mu}$ and $f_{\mu \nu}$ is the field strength tensor for the quantum field $a_{\mu}$. This splitting now allows us to tackle the nonlinear terms in the Lagrangian perturbatively, as we can expand in powers of the quantum terms around a fixed classical background field term.

Additionally, to simplify the problem further, we can assume that the classical field is stationary both in space and time. This reduces the predictive power of the calculations, but simplifies the process dramatically.  This assumption is equivalent to treating $\partial_\mu C_{\nu \rho}$ as negligible in the effective action, $\Gamma \left[ C_{\mu} \right] $. Schwinger \cite{Schwinger1951} made this assumption when calculating the effective action for QED. We will first investigate the case where we consider the background stationary before returning to generalize to the more difficult case.

\section{Taylor Expansion}

We notice that the invariants $S$ and $P$ can thus be decomposed into background and quantum field strength tensors as
\begin{align}
    S&= -\frac{1}{4}F_{\mu \nu} F^{\mu \nu} \nonumber\\
     &= \underbrace{-\frac{1}{4}C_{\mu \nu} C^{\mu \nu}}_{S_C} -\frac{1}{2}C_{\mu \nu} f^{\mu \nu}  \underbrace{- \frac{1}{4}f_{\mu \nu} f^{\mu \nu}}_{S_a} \nonumber \\
    &= S_C -\frac{1}{2}C_{\mu \nu} f^{\mu \nu} + S_a, \label{eq:S}
    \intertext{where we desire to make use of the constant classical field strength condition $\partial_\mu C_{\nu \rho} = 0$. This in fact implies that the cross term vanishes as}
    C_{\mu \nu} \left( \partial^{\mu} a^{\nu}  \right) &= \partial^{\mu} \left( C_{\mu \nu} a^{\nu} \right) - \left( \partial^{\mu} C_{\mu \nu} \right) a^{\nu},
    \intertext{where the first term is a total derivative that amounts to a boundary term in the action (that does not influence the physics) and the second is a derivative of the classical field strength. Therefore, we have $C_{\mu \nu} f^{\mu \nu} = 0$ and can write}
    S = S_C + S_a,
    \intertext{and identically we have}
    P = P_C + P_a
,\end{align}
where $P_C \equiv -\frac{1}{4} C_{\mu \nu} \widetilde{C}^{\mu \nu}$ and $P_a \equiv -\frac{1}{4} f_{\mu \nu} \widetilde{f}^{\mu \nu}$ analogously.

Note that the cross terms do not vanish in $S^2$ and $P^2$ as they are no longer a total derivative and we instead have
\begin{align}
    S^2 &= S_C^2 - \underbrace{S_C C_{\mu \nu} f^{\mu \nu}}_{\mathcal{O}\left( a \right) } + 2 S_C S_a  + \frac{1}{4} C_{\mu \nu} C_{\rho \tau} f^{\mu \nu} f^{\rho \tau}  - \underbrace{S_a C_{\mu \nu} f^{\mu \nu}}_{\mathcal{O}\left( a^3 \right) }+ \underbrace{S_a^2}_{\mathcal{O}\left( a^{4} \right) }.
    \intertext{Note that terms that are linear in the quantum field $\mathcal{O}\left( a \right) $ cumulatively vanish due to the classical equations of motion. As such, these terms and total derivative terms can be ignored. Neglecting terms of order $\mathcal{O}\left( a^{3} \right) $ and greater as they do not contribute at one loop, we are left with}
    S^2 &= S_C^2 -\underbrace{ S_C C_{\mu \nu} f^{\mu \nu}}_{\mathcal{O}\left( a \right) } + 2S_C S_a + \frac{1}{4}C_{\mu \nu} C_{\rho \tau} f^{\mu \nu} f^{\rho \tau}, \\
    P^2 &= P_C^2 - \underbrace{P_C \widetilde{C}_{\mu \nu} f^{\mu \nu}}_{\mathcal{O}\left( a \right)} + \underbrace{2P_C P_a}_{\text{total deriv.}} + \frac{1}{4}\widetilde{C}_{\mu \nu} \widetilde{C}_{\rho \tau} f^{\mu \nu} f^{\rho \tau},
    \intertext{where discarding total derivative and linear terms, we arrive at}
    S^2 + P^2 &= S_C^2 + P_C^2 + 2 S_C S_a +
    \frac{1}{4}\left( C_{\mu \nu} C_{\rho \tau} + \widetilde{C}_{\mu \nu} \widetilde{C}_{\rho \tau} \right) f^{\mu \nu} f^{\rho \tau}
.\end{align}

We see that the first two terms, $S_C^2 + P_C^2$, are purely background dependent and thus serve as a nontrivial classical point to Taylor expand with respect to. Namely, Taylor expanding $\sqrt{S^2 + P^2}$ about $S_C^2 + P_C^2$ we see that
\begin{align}
    \sqrt{S^2 + P^2} &\equiv \sqrt{S_C^2 + P_C^2 + \mathcal{Q}} = \sqrt{S_C^2 + P_C^2}  + \frac{\mathcal{Q}}{2 \sqrt{S_C^2 + P_C^2} } - \frac{\mathcal{Q}^2}{8\left( S_C^2 + P_C^2 \right)^{\frac{3}{2}}} + \mathcal{O}\left( \mathcal{Q}^3 \right),
    \intertext{where the quantum terms are grouped in the object $\mathcal{Q}$ such that}
\mathcal{Q} &\equiv 2 S_C S_a + \underbrace{2 P_C P_a}_{\text{total deriv.}} - \underbrace{\left( S_C C_{\mu \nu} + P_C \widetilde{C}_{\mu \nu} \right) f^{\mu \nu}}_{\mathcal{O}\left( a \right) } + \frac{1}{4}\left( C_{\mu \nu} C_{\rho \tau} + \widetilde{C}_{\mu \nu} \widetilde{C}_{\rho \tau} \right) f^{\mu \nu} f^{\rho \tau},
    \intertext{where the total derivative and linear terms are included as they contribute to higher powers of $\mathcal{Q}$. Namely, up to terms quadratic in the quantum field we have}
    \mathcal{Q}^2 &= \left( S_C^2 C_{\mu \nu} C_{\rho \tau} +  S_C P_C \left(\widetilde{C}_{\mu \nu} C_{\rho \tau} + C_{\mu \nu}\widetilde{C}_{\rho \tau} \right)  + P_C^2 \widetilde{C}_{\mu \nu} \widetilde{C}_{\rho \tau} \right) f^{\mu \nu} f^{\rho \tau} + \mathcal{O}\left( a^3 \right) 
.\end{align}

Therefore with the classical field Lagrangian defined by
\begin{align}
    \mathcal{L}_C &\equiv S_C \cosh \gamma + \sqrt{S_C^2 + P_C^2} \sinh \gamma,
    \intertext{we can thus write the full Lagrangian as}
    \mathcal{L} &= \mathcal{L}_{C} + S_a \cosh \gamma +  \left( \frac{S_C S_a}{\sqrt{S_C^2 + P_C^2}} + B_{\mu \nu \rho \tau} f^{\mu \nu} f^{\rho \tau} \right) \sinh \gamma + \mathcal{O}\left( a^3 \right)
,\end{align}
where I further define
\begin{align}
    B_{\mu \nu \rho \tau} &\equiv \frac{C_{\mu \nu} C_{\rho \tau} + \widetilde{C}_{\mu \nu} \widetilde{C}_{\rho \tau}}{8 \sqrt{S_C^2 + P_C^2}} - \frac{S_C^2 C_{\mu \nu} C_{\rho \tau} + 2S_C P_C C_{\mu \nu}\widetilde{C}_{\rho \tau} + P_C^2 \widetilde{C}_{\mu \nu} \widetilde{C}_{\rho \tau}}{8 \left( S_C^2 + P_C^2 \right)^{\frac{3}{2}} }
,\end{align}
to capture the classical field dependence in the interaction. This tensor coincidentally has the same symmetries as the Riemann curvature tensor $B_{\mu \nu \rho \tau} = B_{\rho \tau \mu \nu}$ and $B_{\mu \nu \rho \tau} = - B_{\nu \mu \rho \tau} = - B_{\mu \nu \tau \rho}$. Here it is constant as it is purely a function of the classical field strength.

The Lagrangian for the quantum field thus suggests the quantum field has a Maxwell-like propagation $S_a \cosh \gamma$ and an interaction vertex quadratic in the quantum field $a_{\mu}$.

As with the QED Lagrangian, we have that we can rearrange $S_a$ into a more useful form with
\begin{align}
    S_a = -\frac{1}{4} f_{\mu \nu} f^{\mu \nu} &= \frac{1}{2} a_{\nu} \left( g^{\mu \nu}\partial_\rho \partial^{\rho}  - \partial^{\mu} \partial^{\nu} \right) a_{\mu},
    \intertext{and using the symmetry of $B_{\mu \nu \rho \tau}$, we have}
    B^{\mu \nu \rho \tau} f_{\mu \nu} f_{\rho \tau} &= 4a_{\nu} \left( B^{\mu \nu \rho \tau} \partial_{\mu} \partial_{\rho} \right) a_{\tau}
,\end{align}
so that the Lagrangian can be expressed as
\begin{align}
    \mathcal{L} &= \mathcal{L}_C +\frac{\cosh \left( \gamma \right) }{2} a_{\nu} \left( g^{\mu \nu}\partial_\rho \partial^{\rho} - \partial^{\mu} \partial^{\nu}\right) a_{\mu} \nonumber\\
    &\phantom{=}~+ \sinh \left( \gamma \right)  a_{\nu} \left( S_C \frac{g^{\mu \nu}\partial_\rho \partial^{\rho} - \partial^{\mu} \partial^{\nu} }{\sqrt{S_C^2 + P_C^2}} + 4B^{\alpha \nu \beta \mu} \partial_{\alpha} \partial_{\beta} \right) a_{\mu}
.\end{align}

As ModMax is an abelian gauge theory, it must be gauge fixed for this Lagrangian to give meaningful results. Applying the Faddeev-Popov procedure results in the addition of the gauge fixing term $\xi a_{\nu} \partial^{\nu} \partial^{\mu} a_{\mu}$ where $\xi$ here can be background dependent. Taking the Feynman gauge equivalent of $\xi = \frac{1}{2}\cosh \gamma - \frac{S_C}{\sqrt{S_C^2 + P_C^2}} \sinh \gamma$ the Lagrangian simplifies to
\begin{align}
    \mathcal{L} &= \mathcal{L}_C +\frac{\cosh \left( \gamma \right) }{2} a_{\nu} g^{\mu \nu}\partial_\rho \partial^{\rho} a_{\mu} + \sinh \left( \gamma \right)  a_{\nu} \left( \frac{S_C}{\sqrt{S_C^2 + P_C^2}}\partial_\rho \partial^{\rho} g^{\mu \nu} + 4B^{\alpha \nu \beta \mu} \partial_{\alpha} \partial_{\beta} \right) a_{\mu}
.\end{align}

With $\partial_\mu \to ik_\mu$, this Lagrangian can be expressed in momentum space as
\begin{align}
    \mathcal{L} &= \mathcal{L}_C -\frac{\cosh \left( \gamma \right) }{2} a_{\nu} k^2 g^{\mu \nu} a_{\mu} - \sinh \left( \gamma \right)  a_{\nu} \left( \frac{S_C}{\sqrt{S_C^2 + P_C^2}} k^2g^{\mu \nu} + 4B^{\alpha \nu \beta \mu} k_{\alpha} k_{\beta} \right) a_{\mu}
.\end{align}

\begin{note}
    We consider the last term in this Lagrangian as an interaction term despite the classical field strength tensors that appear being non-dynamical. Namely, as we have $\partial_\mu C_{\nu \rho} = 0$, the classical field strengths are independent of $x_{\mu}$ and thus can be factored out of the action integral as constants. However this term in the Lagrangian still represents an interaction of the quantum field with a fixed background source.
\end{note}

With the $\sinh \gamma$ term interpreted as an interaction vertex between the classical and quantum fields, the momentum space propagator for the quantum field is then given by
\begin{align}
    -\cosh \gamma k^2 g_{\mu \nu} D^{\nu \rho} &=  i \delta_{\mu}^{\rho} \\
    \implies D^{\nu \rho}  &= \frac{1}{\cosh \gamma} \frac{-i g^{\nu \rho}}{k^2} \label{eq:MMpropagator}
,\end{align}
which is entirely analogous to the propagator for the quantum field in QED. See \cref{sec:propagator} for a derivation of this propagator.

\section{One Loop Effective Action}

With the propagator obtained, we are nearly able to obtain the effective action $\Gamma$ describing quantum corrections to the classical ModMax Lagrangian. Namely, we have that the effective action is equal to the classical action $S\left[ C_{\mu} \right] $ plus quantum corrections given by
\begin{align}
    \Gamma = \underbrace{\int \dd{^{d}x} \mathcal{L}_C}_{S[C_\mu]} + \frac{i}{2} \Tr \left[ \log \left( \frac{\delta^2 \mathcal{L}}{\delta a_{\mu} \delta a_\nu} \right)  \right] + \cdots 
,\end{align}
where we have truncated the expansion at the first term corresponding to one-loop diagrams. In this thesis I compute the one-loop effective action $\Gamma$ using diagrammatic techniques. However, note that heat kernel techniques can also be used and prove effective in the constant background case \cite{Pinelli}.

Evaluating the trace-log of this operator is equivalent to evaluating all Feynman diagrams that are constructable out of the quantum field propagator and the interaction vertex which contain at most one loop. This corresponds to a single infinite series of diagrams given by
\begin{align}
    \Gamma &= S[C_\mu] + \mathcal{D}_1 + \mathcal{D}_2 + \mathcal{D}_3 + \cdots  \nonumber\\
    &= S[C_\mu] +
%\feynmandiagram [horizontal=a to b, layered layout, baseline=(a.base)] {
%    i1 [crossed dot] -- a [dot]
%  -- [photon,half left] b
%  -- [photon,half left] a,
%}; + 
%\feynmandiagram [horizontal=a to b, layered layout, baseline=(a.base)] {
%    i1 [crossed dot]  -- a [dot]
%  -- [photon,half left] b [dot]
%  -- [photon,half left] a,
%  b -- f1 [crossed dot]
%}; +
%\feynmandiagram [small, horizontal=a to d, spring layout, baseline=(a.base)] {
%    i1 [crossed dot]  -- a [dot]
%  -- [photon,quarter right] c [dot]
%  -- [photon,quarter right] d [dot]
%  -- [photon,quarter right] a,
%  c -- f1 [crossed dot],
%  d -- f2 [crossed dot]
%}; + \cdots
    \raisebox{-0.5\height}{\includegraphics{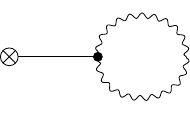}} ~ + ~ 
    \raisebox{-0.5\height}{\includegraphics{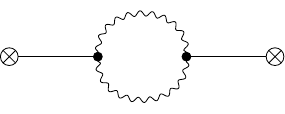}} ~+ ~
    \raisebox{-0.5\height}{\includegraphics{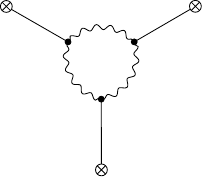}} ~ +  \cdots
,\end{align}
where we draw quantum fields as wavy internal lines and the classical background field as an unlabelled solid external line.

To evaluate such diagrams, we need to derive the vertex factor for ModMax, which together with the propagator forms the Feynman rules for the theory. Reading off the Lagrangian, we see that the vertex factor takes the form
\begin{align} \label{eq:vertex}
    \vcenter{\hbox{
%\begin{tikzpicture}
%    \begin{feynman}
%    \diagram [small,horizontal=a to d] {
%        a -- [photon,momentum=$p\text{, }\nu$] e [dot, label=0:$ $],  
%        b[crossed dot] -- e,
%        e -- [photon, momentum=$p\text{, }\mu$] d, 
%    };
%   \end{feynman}
%\end{tikzpicture}
                  \includegraphics{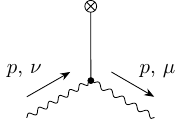}
          }} &= -2\sinh \left( \gamma \right) \left( \frac{S_C}{\sqrt{S_C^2 + P_C^2}} p^2 g^{\mu \nu} + 4 B^{\alpha \nu \beta \mu} p_{\alpha} p_{\beta} \right) 
,\end{align}
where notice that no momentum can flow through the classical fields in the interaction as $\partial_\mu C_{\nu \rho} = 0$.

With this vertex factor obtained, the first diagram in this infinite series is then given by
\begin{align}
    \mathcal{D}_1 = \vcenter{\hbox{
%\begin{tikzpicture}
%    \begin{feynman}
%    \diagram [horizontal=a to b, layered layout] {
%        i1 [crossed dot] -- a [dot] 
%        -- [photon, half left, rmomentum=$k\text{, }\mu$] b 
%         -- [photon, half left, ] a , % rmomentum'=$k\text{, }\mu$
%    }; 
%    \end{feynman}
%\end{tikzpicture}
            \includegraphics{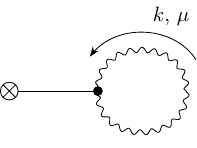}
}} &= \int \frac{\dd{^{d}k}}{\left( 2\pi \right)^{d}} \sinh \left( \gamma \right) \left( \frac{2S_C}{\sqrt{S_C^2 + P_C^2}}  k^2 g^{\mu \nu} + 8 B^{\alpha \nu \beta \mu} k_{\alpha} k_{\beta} \right) D_{\nu \mu} 
,\end{align}
    where the symmetric replacement under the integral of $k_{\alpha } k_{\beta} \to \frac{k^2}{d} g_{\alpha \beta}$\footnote{Note that this replacement is only valid as the only symmetric tensor present is $g^{\alpha \beta}$, and the constant $\frac{1}{d}$ can be seen by contracting both sides with $g_{\alpha \beta}$.} and the propagator derived above yield
\begin{align}
    \mathcal{D}_1 &= -i\tanh \left( \gamma \right)\int \frac{\dd{^{d}k}}{\left( 2\pi \right)^{d}}  \left( \frac{2S_C}{\sqrt{S_C^2 + P_C^2}} k^2 g^{\mu \nu} + \frac{8}{d} B^{\alpha \nu \beta \mu} g_{\alpha \beta} k^2 \right) \frac{g_{\nu \mu}}{k^2} \label{eq:sym_rep_used}     \\
    &= -i\tanh \left( \gamma \right) \left(  2d\frac{S_C}{\sqrt{S_C^2 + P_C^2}}  +  \frac{8}{d}\tensor{B}{^{\alpha \nu}_{\alpha \nu}} \right) \int \frac{\dd{^{d}k}}{\left( 2\pi \right)^{d}}  1
.\end{align}
Notice the lack of $k$ dependence in this integral. Proceeding regardless, this contraction simplifies greatly with the use of the identity
\begin{align}
    -\frac{1}{4}\widetilde{C}_{\mu \nu} \widetilde{C}^{\mu \nu} &= -S_C,
    \intertext{to}
   \tensor{B}{^{\alpha \nu}_{\alpha \nu}} &= -\frac{S_C^3 - 2S_C P_C^2 + S_C P_C^2}{2 \left( S_C^2 + P_C^2 \right)^{\frac{3}{2}}} \nonumber \\
 &= -\frac{S_C^3-S_C P_C^2}{2 \left( S_C^2 + P_C^2 \right)^{\frac{3}{2}}}
.\end{align}

However, as the integral is independent of $k$, it is scaleless and thus while $\Lambda^{d}$-divergent with a naive cutoff, using dimensional regularization we can show it is exactly zero.

\section{Dimensional Regularization}

Divergent quantities that one would naively expect to be physical are a common occurrence in quantum field theory. Such divergences usually arise at some number of loops, and require \textit{renormalization}, the process of recovering finite physical results from such theories by redefining constants in the Lagrangian \cite{Schroeder,Renormalization,LeBellac}. Such redefinitions absorb the divergences that arise.

However, first one must characterize the divergence, the process of which is referred to as \textit{regularization}. This process is not unique and there are many \textit{regulators} one can make use of. For example, for the above integral, the naive method is called \textit{cutoff regularization} where one introduces a maximum momentum scale $k^2 \leq \Lambda^2$ such that our integral now reads,
\begin{align}
    \int \frac{\dd{^{d}k}}{\left( 2\pi \right)^{d}} 1 \too \int_{-\Lambda}^{\Lambda} \frac{\dd{^{d}k}}{\left( 2\pi \right)^{d}} 1
.\end{align}
One would then have $\Lambda$ dependence in the Lagrangian, with the intention of taking the limit of $\Lambda \to \infty$ to recover the original theory.

\begin{note}
    The final result of regularization is expected to be independent of the regulator used. If two regulators lead to different observable quantities, often a symmetry is being broken by one or both regulators.
\end{note}

While cutoff regularization is approachable, it is not the most elegant method. This is largely caused by the introduction of a characteristic length scale $\Lambda$ in our otherwise scaleless theory. 

Instead, we will make use of \textit{dimensional regularization} where we consider the dimension to be a free parameter $d \in \R \setminus \N$. The natural numbers $\N$ are excluded from the domain here as they lead to divergences and thus undefined values. However, the limit in which $d \to n \in \N$ (usually $d = 4$) is well defined. We then appeal to \textit{analytic continuation}, where as a given calculation is an analytic function of the dimension $d$, we define the value at $d=4$ to agree with the limit $d \to 4$. This is a valid regulator, and introduces no characteristic scale. We use dimensional regularization throughout this thesis as it causes a large number of otherwise difficult integrals to vanish.

\begin{claim}
    In fact, returning to the integral at hand, using dimensional regularization, we can show that
    \begin{align}
        \int \frac{\dd{^{d}k}}{\left( 2\pi \right)^{d}} 1 &\to 0
    ,\end{align}
    as $d \to n \in \N$.
\end{claim}

\begin{proof}
    We begin by considering $d \in \mathbb{R} \setminus \mathbb{N}$ dimensional Euclidean space (rather than Minkowski) for convenience. Taking the integral
    \begin{align}
        \int \frac{\dd{^{d}k}}{\left( 2\pi \right)^{d}} 1 &=  \int \frac{\dd{^{d}k}}{\left( 2\pi \right)^{d}} \frac{k^2}{k^2 + m^2} + \int \frac{\dd{^{d}k}}{\left( 2\pi \right)^{d}} \frac{m^2}{k^2 + m^2} 
    .\end{align}
    These are both specific cases of a known integral \cite{Renormalization}, with general form given by
    \begin{align}
        \int \frac{\dd{^{d}k}}{\left( 2\pi \right)^{d}} \frac{k^{2\beta}}{\left( k^2 + m^2 \right)^{\alpha}} &= \frac{\Gamma \left( \beta + \frac{d}{2} \right) \Gamma \left( \alpha -\beta - \frac{d}{2} \right) }{\left( 4\pi \right)^{\frac{d}{2}} \Gamma \left( \alpha \right) \Gamma \left( \frac{d}{2} \right)  } m^{2 \left( \frac{d}{2} - \alpha + \beta \right) }
    ,\end{align}
    where $\Gamma \left( n + 1 \right) = n!$ is the \textit{Gamma function}, the generalization of the factorial to the real numbers using analytic continuation.
    With $\alpha = \beta = 1$ we see that
    \begin{align}
        \int \frac{\dd{^{d}k}}{\left( 2\pi \right)^{d}} \frac{k^2}{k^2 + m^2} &= \frac{\Gamma \left( 1 + \frac{d}{2} \right) \Gamma\left( -\frac{d}{2} \right)  }{\left( 4\pi \right)^{\frac{d}{2}} \Gamma \left( \frac{d}{2}\right) } m^{d},
        \intertext{where $\Gamma \left( x + 1 \right) = x\Gamma \left( x \right) $ implies}
     \int \frac{\dd{^{d}k}}{\left( 2\pi \right)^{d}} \frac{k^2}{k^2 + m^2} &= \left( \frac{d}{2} \right) \frac{\Gamma \left( \frac{d}{2} \right) \Gamma\left( -\frac{d}{2} \right)  }{\left( 4\pi \right)^{\frac{d}{2}} \Gamma \left( \frac{d}{2}\right) } m^{d},
        \intertext{and for $\beta = 0$ we find}
        m^2 \int \frac{\dd{^{d}k}}{\left( 2\pi \right)^{d}} \frac{1}{k^2 + m^2} &= m^2\frac{\Gamma\left( \frac{d}{2} \right) \Gamma \left( 1 - \frac{d}{2} \right)  }{\left( 4\pi \right)^{\frac{d}{2}} \Gamma\left( \frac{d}{2} \right) } m^{d-2},
        \intertext{where $\Gamma \left( x + 1 \right) = x\Gamma \left( x \right) $ similarly implies}
        m^2 \int \frac{\dd{^{d}k}}{\left( 2\pi \right)^{d}} \frac{1}{k^2 + m^2} &= \left( -\frac{d}{2} \right) \frac{\Gamma\left( \frac{d}{2} \right) \Gamma \left( - \frac{d}{2} \right)  }{\left( 4\pi \right)^{\frac{d}{2}} \Gamma\left( \frac{d}{2} \right) } m^{d}
    ,\end{align}
    which is the negative of the previous integral, and thus their sum vanishes as desired for any $d \in \R \setminus \N$. By appealing to the analytic continuation of this result as $d \to 4$, it holds identically for $d=4$. \\
\end{proof}

In fact, the above argument is generalizable to show that any integral of the form
\begin{align*}
    \int \frac{\dd{^{d}k}}{\left( 2\pi \right)^{d}} k^{2 \alpha}
,\end{align*}
vanishes in dimensional regularization for $\alpha \in \Z$. Intuitively, this is because the integral has no characteristic external scale dependence which is central in dimensional regularization \cite{Smirnov}.

\section{Generalization to $n$th order diagrams}

Notice that the above result relies only on the momentum dependence of the integral.

For a general $n$th order one loop diagram however, all insertions of additional vertices do not change the momentum dependence of the integral. Namely, as we will have $n$ propagators (\cref{eq:MMpropagator}) $D_{\mu \nu} \propto \frac{1}{k^2}$ and $n$ vertices (\cref{eq:vertex}) $\propto k^2$ which are all equal by momentum conservation (as we assumed no momentum flow through the classical fields), the integral will be momentum independent as was the case for the 1 vertex diagram, $\mathcal{D}_1$. 

In fact, as the vertex factor contains $k_{\alpha} k_{\beta}$ rather than $k^2$ we will first have to perform the symmetrization of $2n$ momenta that leads to the replacement
\begin{align}
    \prod_{i=1}^{n} k^{\alpha_{2i-1}} k^{\alpha_{2i}} \to \frac{k^{2n} \left( d -2 \right)!! (2n-1)!!}{\left( d-2 +2n \right)!! } g^{(\alpha_{1}\alpha_2} \cdots g^{\alpha_{2n} \alpha_{2n+1})} \label{eq:symmetrization}
,\end{align}
where $g^{(\mu \nu} g^{\rho \tau)} = \frac{1}{3} \left( g^{\mu \nu} g^{\rho \tau} + g^{\mu \rho} g^{\nu \tau} + g^{\mu \tau} g^{\nu \rho} \right)$ is symmetrization with normalization, and $n!! = n \cdot (n-2) \ldots 4 \cdot 2$ is the double factorial.

Therefore all diagrams $\mathcal{D}_n$ have an integrand independent of momentum in the same fashion as $\mathcal{D}_1$, and thus vanish in dimensional regularization by the same argument.

We therefore conclude that the perturbative one-loop effective action $\Gamma$, with constant background field strength $C_{\mu\nu}$, vanishes through dimensional regularization. This implies that there are no 1-loop corrections to the classical theory under these assumptions.

\section{Two Loops}

The only two loop diagram to order $\gamma$ (i.e. having one vertex as each carries $\sinh \gamma \overset{\mathcal{O}\left( \gamma \right) }{=} \gamma$) is quartic in the quantum field such that
\begin{align*}
    \mathcal{D}_{\text{2 loops}} &= 
    \vcenter{\hbox{
%\begin{tikzpicture}
%    \begin{feynman}
%    \diagram [horizontal=a to c] {
%        a -- [photon, half left, looseness=1.5, rmomentum=$p$] b [dot, label=0:$ $], 
%        a -- [photon, half right, looseness=1.5, ] b,
%        b -- [photon, half left, looseness=1.5, ] c,
%        b -- [photon, half right, looseness=1.5, rmomentum'=$q$] c;
%    }; 
%    \end{feynman}
%\end{tikzpicture}
            \includegraphics{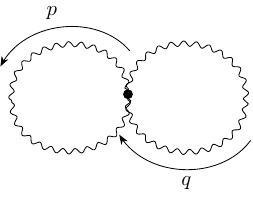}
}}
.\end{align*}
The classical field line has been suppressed for clarity (and due to the fact that it carries no momenta). To evaluate this diagram, we need to return to the expansion of the square root and consider terms up to quartic order such that
\begin{align}
    \sqrt{S_C^2 + P_C^2 + \mathcal{Q}}  &= \sqrt{S_C^2 + P_C^2}  + \frac{\mathcal{Q}}{2 \sqrt{S_C^2 + P_C^2}} - \frac{\mathcal{Q}^2}{8 \left( S_C^2 + P_C^2 \right)^{\frac{3}{2}}}\nonumber\\
    &\phantom{=}~+ \frac{\mathcal{Q}^3}{16 \left( S_C^2 + P_C^2 \right)^{\frac{5}{2}}} - \frac{5\mathcal{Q}^{4}}{128 \left( S_C^2 + P_C^2 \right)^{\frac{7}{2}}} + \mathcal{O}\left( \mathcal{Q}^{5} \right) 
.\end{align}

From $\mathcal{Q}$ and $a_\mu$ power counting, we have that the $\mathcal{O}\left( a^{4} \right) $ term in the Lagrangian is
\begin{align}
    \frac{\mathcal{L}_{a^4}}{\sinh \gamma} &= \frac{1}{2\sqrt{S_C^2 + P_C^2} }S_a^2 \nonumber\\
                      &\phantom{=}~- \frac{1}{8\left( S_C^2 + P_C^2 \right)^{\frac{3}{2}}}\bigg[ \left( 2S_C S_a + 2P_C P_a + \frac{1}{4}\left( C_{\mu \nu} C_{\rho \tau} + \widetilde{C}_{\mu \nu} \widetilde{C}_{\rho \tau} \right) f^{\mu \nu} f^{\rho \tau} \right)^2\nonumber\\
                      &\phantom{=}~+ 2\left( \left( S_C C_{\mu \nu} + P_{C}\widetilde{C}_{\mu \nu} \right) f^{\mu \nu} \right)  \left( \left( S_a C_{\alpha \beta}  + P_a \widetilde{C}_{\alpha \beta} \right) f^{\alpha \beta}  \right) \bigg]  \nonumber\\
                      &\phantom{=}~+\frac{3}{16\left( S_C^2 + P_C^2 \right)^{\frac{5}{2}}} \left( \left( S_C C_{\alpha \beta} + P_C \widetilde{C}_{\alpha \beta} \right) f^{\alpha \beta} \right)^2\nonumber\\
                      &\phantom{=}~\phantom{+}\times\left( 2S_C S_a + 2P_C P_a + \frac{1}{4}\left( C_{\mu \nu} C_{\rho \tau} + \widetilde{C}_{\mu \nu} \widetilde{C}_{\rho \tau} \right) f^{\mu \nu} f^{\rho \tau}  \right) \nonumber \\
                      &\phantom{=}~-\frac{5}{128 \left( S_C^2 + P_C^2 \right)^{\frac{7}{2}}}\left( \left( S_C C_{\mu \nu} + P_C \widetilde{C}_{\mu \nu} \right) f^{\mu \nu}  \right)^{4} 
.\end{align}

It is unapproachable to ascertain the exact form of the background dependence. However, as the quantum fields only appear in the field strengths, we notice that we can write the vertex
\begin{align}
    \mathcal{L}_{a^{4}} &= \sinh \left( \gamma \right) A_{\mu \nu \rho \tau \alpha \beta \sigma \kappa} f^{\mu \nu} f^{\rho \tau} f^{\alpha \beta} f^{\sigma \kappa},
    \intertext{where $A$ is a background dependent, momentum independent tensor that captures the structure of the vertex. Its exact form is not important as we will see. Observe that this can be written, for some tensor $A'$, as}
    \mathcal{L}_{a^{4}} &= \sinh \left( \gamma \right) A'_{\mu \nu \rho \tau \alpha \beta \sigma \kappa} \partial^{\mu} a^{\nu} \partial^{\rho} a^{\tau} \partial^{\alpha} a^{\beta} \partial^{\sigma} a^{\kappa}
.\end{align}

As each field has a derivative acting on it, we can express this vertex as
\begin{align}
    \vcenter{\hbox{
%\begin{tikzpicture}
%    \begin{feynman}
%    \diagram [small,horizontal=a to d] {
%        a -- [photon, momentum=$p\text{, }\mu$] e [dot, label=0:$ $],  
%        b -- [photon, momentum=$q\text{, }\nu$] e, 
%        e -- [photon, rmomentum'=$k\text{, }\rho$]c;
%        e -- [photon, rmomentum'=$r\text{, }\sigma$]d;
%};
%    \end{feynman}
%\end{tikzpicture}
            \includegraphics{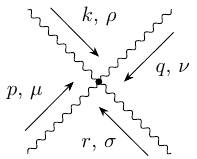}
}} &= \sinh \left( \gamma \right)  \tensor{B}{^{\alpha}_{\mu}^{\beta}_{\nu}^{\kappa}_{\rho}^{\tau}_{\sigma}} p_{\alpha} q_{\beta} k_{\kappa} r_{\tau} 
,\end{align}
where $\tensor{B}{^{\alpha}_{\mu}^{\beta}_{\nu}^{\kappa}_{\rho}^{\tau}_{\sigma}}$ is also independent of the momenta and captures both the background dependence (from the $A'$ tensor) and the symmetrization of momenta from different contractions. The explicit form of $B$ is also not important. Therefore the double loop diagram can be written with symmetry factor $S = 8$ \cite{Schroeder} as
\begin{align}
    \mathcal{D}_{\text{2 loops}} &= \frac{\sinh \gamma}{8} \tensor{B}{^{\alpha}_{\mu}^{\beta}_{\nu}^{\kappa}_{\rho}^{\tau}_{\sigma}} \int \frac{\dd{^{d}p} \dd{^{d}q}}{\left( 2\pi \right)^{2d}} p_{\alpha} p_{\beta} q_{\kappa} q_{\tau} D^{\mu \rho} D^{\nu \sigma},
    \intertext{where with the propagator from \cref{eq:MMpropagator} given by $D^{\mu \rho} = \frac{1}{\cosh \gamma} \frac{-i g^{\mu \rho}}{k^2}$, we have}
    \mathcal{D}_{\text{2 loops}} &= \frac{\coth \gamma}{8\cosh \gamma}\tensor{B}{^{\alpha}_{\mu}^{\beta}_{\nu}^{\kappa}_{\rho}^{\tau}_{\sigma}} \int \frac{\dd{^{d}p} \dd{^{d}q}}{\left( 2\pi \right)^{2d}} p_{\alpha} p_{\beta} q_{\kappa} q_{\tau} \frac{g^{\mu \rho}g^{\nu \sigma}}{p^2 q^2},
     \intertext{where applying the metric tensors yields}
    \mathcal{D}_{\text{2 loops}} &= \frac{\coth \gamma}{8\cosh \gamma}\tensor{B}{^{\alpha}_{\mu}^{\beta}_{\nu}^{\kappa}^{\mu}^{\tau}^{\nu}}\int \frac{\dd{^{d}p} \dd{^{d}q}}{\left( 2\pi \right)^{2d}} \frac{p_{\alpha} p_{\beta} q_{\kappa} q_{\tau} }{p^2 q^2},
     \intertext{where as before symmetry implies we can take $p_{\alpha} p_{\beta} \to \frac{p^2}{4} g_{\alpha \beta}$ and identically for $q$ yielding}
    \mathcal{D}_{\text{2 loops}} &= \frac{\coth \gamma}{128\cosh \gamma}\tensor{B}{_{\beta}_{\mu}^{\beta}_{\nu}_{\tau}^{\mu}^{\tau}^{\nu}}  \int \frac{\dd{^{d}p} \dd{^{d}q}}{\left( 2\pi \right)^{2d}} 1
,\end{align}
for which, as we found above, cutoff regularization suggests a polynomial divergence of the form $\Lambda^{2d}$, however the dimensional regularization result derived above allows us to conclude it vanishes.

Notice that we can insert a vertex along either of these loops which will not change the momentum structure of the diagram and get us to $\gamma^2$ order
\begin{align*}
\vcenter{\hbox{
%\begin{tikzpicture}
%    \begin{feynman}
%    \diagram [horizontal=a to c] {
%        a [dot] -- [photon, half left, looseness=1.5, rmomentum=$p$] b [dot, label=0:$ $], 
%        a -- [photon, half right, looseness=1.5, momentum'=$p$] b, %, momentum'=$k\text{, }\rho$
%        b -- [photon, half left, looseness=1.5, rmomentum=$q$] c, 
%        b -- [photon, half right, looseness=1.5, ] c ; % rmomentum'=$k\text{, }\mu$
%    };
%    \end{feynman}
%\end{tikzpicture}
        \includegraphics{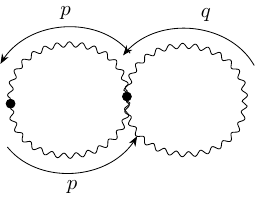}
}}
,\end{align*}
and thus we conclude that this diagram will also vanish by similar arguments.

The only other diagram arising at order $\gamma^2$ has two cubic vertices, which will come from a term in the Lagrangian of the form
\begin{align}
   \mathcal{L}_{a^3} &= A'_{\mu \nu \rho \tau \alpha \beta} \partial_{\mu} a_{\nu} \partial_{\rho} a_{\tau} \partial_{\alpha} a_{\beta}
,\end{align}
and leads to a vertex of the form
\begin{align}
\vcenter{\hbox{
%\begin{tikzpicture}
%    \begin{feynman}
%     \diagram [small,horizontal=a to d] {
%         e [dot] -- [photon,rmomentum'=$\alpha\text{, }k^3$] c ;
%         a -- [boson, momentum=$\mu\text{, }k^1$] e,  
%         b -- [boson, momentum=$\rho\text{, }k^2$] e, 
%};
%    \end{feynman}
%\end{tikzpicture}
        \includegraphics{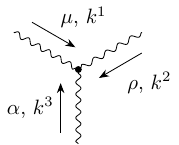}
}} =
\sinh \left( \gamma \right) \tensor{B}{^{\mu}_{\nu}^{\rho}_{\tau}^{\alpha}_{\beta}} k^1_{\mu} k^2_{\rho} k^3_{\alpha}
.\end{align}

This diagram does not immediately vanish and the full calculation using dimensional regularization is shown in \cref{sec:2_loop_2_vertex_calc}. The diagram evaluates to

\begin{align}
    \mathcal{D}_{\text{Cubic}} &=  
\vcenter{\hbox{
%\begin{tikzpicture}
%    \begin{feynman}
%    \diagram [horizontal=b to c] {
%        c [dot] -- [photon, half right, momentum'=$p$] b [dot, label=0:$ $], 
%        b -- [photon,momentum'=$q$] c, 
%        c -- [photon, half left, rmomentum=$p-q$] b, 
%}; 
%    \end{feynman}
%\end{tikzpicture}
    \includegraphics{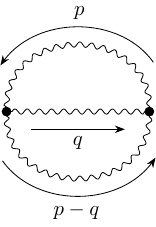}
}} \nonumber\\
                               &=  \left( -i \right)^3 \frac{\sinh^2 \gamma}{2 \cosh^3 \gamma} \tensor{B}{^{\mu_1}_{\nu_1}^{\rho_1}_{\tau_1}^{\alpha_1}_{\beta_1}} B^{\mu_2 \nu_1 \rho_2 \tau_1 \alpha_2 \beta_1} \times \nonumber \\
&\phantom{=}\int \frac{\dd{^{d}q}}{\left( 2\pi \right)^{d}} \frac{q_{\rho_1} q_{\rho_2} q_{\alpha_1} q_{\alpha_2}}{q^2} \frac{i\Gamma\left( 1 - \frac{d}{2} \right)}{\left( 4\pi \right)^{\frac{d}{2}} \left( d-1 \right)  }\left[ q_{\mu_1} q_{\mu_2} q^{d-4} \left( 1 - \frac{d}{2} \right) - \frac{dg_{\mu_1 \mu_2}}{8}  q^{d-2}  \right]
.\end{align}

Ignoring the prefactors to focus on the structure of the divergence, we take $d = 4 + 2\epsilon$ with the intention of taking the limit $\epsilon \to 0$ such that $d \to 4$. Using the expansion of
\begin{align}
    \Gamma \left( -1-\epsilon \right) = \frac{1}{\epsilon} - \gamma + 1
,\end{align}
we find that the divergent part of this integral is
\begin{align}
      \mathcal{D}_{\text{Cubic}} &\propto \frac{1}{\epsilon}\int \frac{\dd{^{4}q}}{\left( 2\pi \right)^{4}} \frac{q_{\rho_1} q_{\rho_2} q_{\alpha_1} q_{\alpha_2}}{q^2} \frac{i}{\left( 4\pi \right)^{2} 6}\left[ 2q_{\mu_1} q_{\mu_2} + g_{\mu_1 \mu_2}q^{2}  \right]
.\end{align}

This term is \textit{logarithmically divergent}, and similarly to before, due to the presence of the prefactor tensors, does not resemble the original Lagrangian. Notice that when $\epsilon > 0$ so that $d \neq 4$, we are left with a symmetrizable integral over $q$ that will vanish identically. By analytic continuation, in this regularization scheme we therefore conclude that the integrals also vanish at $d = 4$.

Therefore the two loop effective action also vanishes at minimum to order $\gamma^2$. Proceeding further in this manner is impractical but we expect that a generalization can be made to suggest that the effective action should vanish at all orders in $\gamma$ and at all loops. Allowing derivatives to act on the field strength such that $\partial_{\mu} C_{\nu \rho} \neq 0$ contrary to what was assumed here, will lead to logarithmic divergences and even constant diagrams as we will see below that will not vanish in dimensional regularization. 

\section{Varying Backgrounds}

Proceeding in a similar manner to before, we perform the background field method splitting, only now without discarding terms of the form $\partial_\mu C_{\nu \rho}$ (as well as higher derivatives). This leads to an entirely analogous Lagrangian
\begin{align}
    \mathcal{L} &= \mathcal{L}_C + S \cosh \gamma + \sinh \gamma a_{\mu} P^{\mu \nu} a_{\nu},
    \intertext{where the background field dependence within $P^{\mu\nu}$ has become more complex with}
    P^{\mu \nu} &=  \left( S_C \partial^{\rho} + \partial^{\rho} S_C \right)  \partial_\rho g^{\mu \nu} - \partial^{\nu} S_C \partial^{\mu} - \partial_\alpha P_C \epsilon^{\alpha \nu \rho \mu} \partial_\rho + 2 \tensor{B}{^\tau^\mu^\rho^\nu} \partial_{\tau} \partial_{\rho} + 2 \partial_{\tau} B^{\tau \mu\rho \nu} \partial_{\rho}
,\end{align}
where trailing partial derivatives act on the quantum field $a_{\nu}$.

When we held $\partial_\mu C_{\nu \rho} = 0$ we were able to factor the classical field dependence out of the integral as it was necessarily independent of $x^{\mu}$. However in the general varying background case, this is no longer possible. Nonetheless, if we consider $P^{\mu \nu}\left( x \right) $ to represent the cumulative effect of the background rather than just a composite operator, then we can obtain Feynman rules for this theory. Namely, we see in the interaction vertex (with the prefactor $\sinh \gamma$), that we have two factors of the photon field $a_{\mu}$ and one $P_{\mu \nu}$. As such, our interaction vertex has two photons (represented by wavy lines) and one cumulative classical background field $P_{\mu \nu}$ (represented by a coiled line) which leads to a vertex factor
\begin{align}
    \vcenter{\hbox{
%\begin{tikzpicture}
%    \begin{feynman}
%    \diagram [small,vertical=a to d] {
%        e [dot] -- [gluon,rmomentum'=$r\,\,\mathrm{=}-p-q$] c [crossed dot,label=2:$\mu\text{,~}\nu$];
%        a -- [photon, momentum=$p$] e,  
%        b -- [photon, momentum=$q$] e, 
%    };
%    \end{feynman}
%\end{tikzpicture}
            \includegraphics{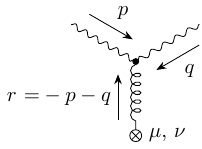}
}} &= -i\sinh \gamma \int \dd{^{4}x}P_{\mu \nu} \left( x \right)
,\end{align}
entirely analogously to the QED case. With this vertex factor, we proceed to evaluate the leading corrections in the infinite series of diagrams for a generic classical background field.

The only one-vertex diagram constructable from this vertex vanishes by an identical argument as in the constant background case, and thus we move to the one loop two vertex diagram, which we call $\mathcal{D}_2$ and draw
\begin{align}
    \mathcal{D}_2 = \vcenter{\hbox{
%\begin{tikzpicture}
%    \begin{feynman}
%    \diagram [layered layout, horizontal=a to b] {
%        z [crossed dot,label=$\mu\nu$] -- [gluon, momentum=$q$] a [dot]-- [photon, half left, looseness=1.5, rmomentum=$\ell$] b [dot], 
%        a -- [photon, half right, looseness=1.5, momentum'=$\ell + q$] b ,
%        b -- [gluon,momentum=$q$] c [crossed dot,label=$\rho\tau$];
%    };
%    \end{feynman}
%\end{tikzpicture}
            \includegraphics{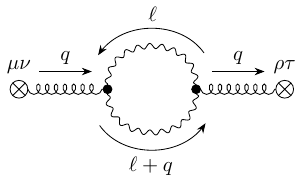}
}}
.\end{align}

Applying the derived Feynman rules (i.e. the vertex factor and propagator) in momentum space we see that this diagram yields
\begin{align}
     \mathcal{D}_2 &= \frac{\sinh^2 \gamma}{\cosh^2 \gamma} \int \frac{\dd{^{d}\ell} \dd{^{d}q}}{\left( 2\pi \right)^{2d}} P_{\mu \nu}\left( q,\ell \right) \frac{g^{\mu\rho} g^{\nu \tau}}{\ell^2 \left( \ell^2 + q^2 \right) } P_{\rho \tau} \left( -q,\ell \right).
    \intertext{Using dimensional regularization, we find by a similar argument as above that the leading order (in external momentum $q$) divergent contribution is of the form}
\mathcal{D}_2 &\propto \Gamma \left( -\frac{d}{2} \right) g_{(\mu \rho} g_{\alpha \beta)}  \int \frac{\dd{^{d}q}}{\left( 2\pi \right)^{d}} B^{\mu \nu \rho \tau} \tensor{B}{^\alpha_\nu^{\beta}_\tau} q^{d},
    \intertext{where the remaining free indices are contracted with symmetrized metric tensors. With $d = 4+2\epsilon$ and taking $\epsilon \to 0$ such that $d \to 4$, we see}
    \mathcal{D}_2 &\propto \left( \frac{1}{\epsilon} \right) g_{(\mu \rho} g_{\alpha \beta)} \int \frac{\dd{^{4}q}}{\left( 2\pi \right)^{4}} B^{\mu \nu \rho \tau} \tensor{B}{^\alpha_\nu^{\beta}_\tau} q^{4}
,\end{align}
which is not only divergent as $\epsilon \to 0$, but is also not of the form of the original Lagrangian. If this result was of the form of the original Lagrangian, then we can interpret such a diagram as suggesting a redefinition of a constant (such as $\gamma$) within the original Lagrangian, to absorb this divergence. However, as this is not the case, this suggests that ModMax is likely not a physical theory on the quantum domain, without modification. Note that it is possible that there is a structure hidden in such divergences which is of the form of the original Lagrangian, which may only appear when evaluated to all loops or all orders in $\gamma$. As this is infeasible to ascertain directly due to the increased complexity with order, we move to study ModMax's two dimensional analogue which may elucidate further insight.

Nonetheless, we have obtained the effective action for ModMax in both the static and varying background case, thus characterizing the predictions of the theory on the quantum level.

\chapter{Scalar Field Analogue in Two Dimensions}

With the quantization of ModMax characterized through the two-loop effective action, I study the scalar analogue of ModMax in $d = 2$ and show that the methods applied to ModMax generalized appropriately. In particular, I focus on generalizing the varying background case. Two dimensional conformal field theory is a highly active area of research, as conformal symmetry is extremely restrictive and thus powerful in $1+1$ spacetime dimensions. We expect nontrivial results to emerge at a lower order than in $d=4$ and thus such a toy model can provide insight into the underlying physics.

Similarly, upon successfully quantizing a conformal field theory, the corrections to the classical theory that arise can lead to the discovery of novel conformal field theories at the classical level. This phenomenon can be referred to as a classical conformal field theory \textit{generated} quantum mechanically. However, if the conformal symmetry is broken at the quantum level, referred to as an \textit{anomaly}, this is equally of interest as this often leads to observable predictions of the theory. The corrections I obtain appear to respect conformal symmetry, however further investigation is required into the properties of such generated classical conformal field theories.

Through the generalization of the approach I developed above, I find that the one loop effective action for backgrounds with constant field strength identically vanishes as in the original theory. By allowing the background to vary, I find that the first corrections to the classical theory arise at the order $\mathcal{O}\left( \gamma^2 \right) $, and resemble a QED-like operator dependent on the classical field. Motivated by the experimental bound of $\gamma \leq 3 \times 10^{-22}$, I further calculated the contribution from all diagrams up to order $\mathcal{O}\left( \gamma^2 \right)$, namely the infinite series of diagrams containing up to 2 vertices. The results of this section appear in \cite{paper}.

\section{Background Field Method}

Recall that in $d = 4$, the ModMax Lagrangian \cite{Sorokin_2022,Lechner_2022} is given by
\begin{align}
    \mathcal{L} &=  S \cosh \gamma + \sqrt{S^2 + P^2}  \sinh \gamma
.\end{align}

By dimensional reduction from ModMax in $d = 4$, one can define in $d = 2$ for $N \geq 2$ scalar boson fields $\phi^{i}$, (where $i \in \{1,\cdots,N\} $ labels the bosons\footnote{All $i,j \in \{1,\cdots,N\} $ indices are exclusively superscripts here to distinguish them from 3-vector indices used in Chapter 1.}) the Lagrangian \cite{Ferko_2022_2,Conti_2022,Barbashov1966}
\begin{align}
    \mathcal{L} = \frac{\cosh \gamma}{2} \partial_{\mu} \phi^{i}\partial^{\mu} \phi^{i} + \frac{\sinh \gamma}{2} \sqrt{2 \left( \partial_\mu \phi^{i} \partial^{\nu} \phi^{i} \right) \left( \partial_\nu \phi^{j} \partial^{\mu} \phi^{j} \right) - \left( \partial_\mu \phi^{i} \partial^{\mu} \phi^{i} \right)^2   } 
,\end{align}
where we have implicit summation over $i,j = 1, \ldots, N$ which label the bosons. If $N = 1$, the theory becomes trivial and reduces to $\mathcal{L} = e^{\gamma} \partial_\mu \phi \partial^{\mu} \phi$. The analogue of the field strength for this theory, $\forall N \in \N$, is
\begin{align}
    \tensor{\phi}{_{\mu}^{i}} &\equiv \partial_\mu \phi^{i}
,\end{align}
which lets us write the Lagrangian as
\begin{align}
    \mathcal{L} = \frac{\cosh \gamma}{2} \tensor{\phi}{_{\mu}^{i}} \tensor{\phi}{^{\mu}^{i}}+ \frac{\sinh \gamma}{2} \sqrt{2 \left( \tensor{\phi}{_{\mu}^{i}}\tensor{\phi}{^{\nu}^{i}} \right) \left( \tensor{\phi}{_{\nu}^{j}} \tensor{\phi}{^{\mu}^{j}} \right) - \left( \tensor{\phi}{_{\mu}^{i}} \tensor{\phi}{^{\mu i}} \right)^2  } 
,\end{align}
where the ModMax-like structure is more apparent.

Employing the background field method \cite{Abbott}, we split the field $\phi^{i}$ into a classical (background) field $C^{i}$ and a quantum field $Q^{i}$ such that
\begin{align}
    \phi^{i} &= C^{i} + Q^{i} \\
    \implies \tensor{\phi}{_{\mu}^{i}} &= \partial_\mu \left( C^i + Q^i \right) \\
    &\equiv \tensor{C}{_{\mu}^{i}} + \tensor{Q}{_{\mu}^{i}}
,\end{align}
where $\tensor{C}{_{\mu}^{i}}$ is the field strength tensor for the classical field $C^{i}$ and $\tensor{Q}{_{\mu}^{i}}$ is the field strength tensor for the quantum field $Q^{i}$.

We notice that one can choose analogues of the invariants $S$ and $P$ which can similarly be decomposed as
\begin{align}
    S&\equiv \tensor{\phi}{_{\mu}^{i}} \tensor{\phi}{^{\mu i}} \\
    &= \underbrace{\tensor{C}{_{\mu}^{i}} \tensor{C}{^{\mu i}}}_{S_C} + 2 \tensor{C}{_{\mu}^{i}} \tensor{Q}{^{\mu i}} + \underbrace{\tensor{Q}{_{\mu}^{i}} \tensor{Q}{^{\mu i}}}_{S_Q},\\
    \tensor{P}{_\mu^\nu} &\equiv \tensor{\phi}{_{\mu}^{i}} \tensor{\phi}{^{\nu i}} \\
     &= \underbrace{\tensor{C}{_{\mu}^{i}} \tensor{C}{^{\nu i}}}_{P_C} + \tensor{Q}{_{\mu}^{i}} \tensor{C}{^{\nu i}} + \tensor{C}{_{\mu}^{i}} \tensor{Q}{^{\nu i}} + \underbrace{\tensor{Q}{_{\mu}^{i}} \tensor{Q}{^{\nu i}}}_{P_Q}
,\end{align}
where the subscript on these invariants denotes the field they correspond to.
\begin{note}
    In this chapter, we do not assume $\partial_\mu \tensor{C}{_{\nu}^{i}} = 0$. Nonetheless, terms linear in $Q^{i}$ will still vanish in the computation of the effective action, as the adjoining classical fields collectively satisfy the equations of motion \cite{Schroeder}.
\end{note}

Observe that $S^2$ and $P^2$ can be written as
\begin{align}
    S^2 &= S_C^2 + \underbrace{4S_C \tensor{C}{_{\mu}^{i}} \tensor{Q}{^{\mu i}}}_{\mathcal{O}\left( Q \right) } + 2 S_C S_Q  + 4\tensor{C}{_{\mu}^{i}}Q^{\mu i} \tensor{C}{_\nu^{j}} Q^{\nu j}  + \underbrace{4S_Q \tensor{C}{_{\mu}^{i}} Q^{\mu i}}_{\mathcal{O}\left( Q^3 \right) }+ \underbrace{S_Q^2}_{\mathcal{O}\left( Q^{4} \right) },
    \intertext{where neglecting higher order terms we are left with}
    S^2 &= S_C^2 + \underbrace{4S_C \tensor{C}{_{\mu}^{i}} \tensor{Q}{^{\mu i}}}_{\mathcal{O}\left( Q \right) } + 2 S_C S_Q  + 4\tensor{C}{_{\mu}^{i}}Q^{\mu i} \tensor{C}{_\nu^{j}} Q^{\nu j}, \\
    P^2 &= \tensor{C}{_{\mu}^{i}} \tensor{C}{^{\mu j}} \tensor{C}{_{\nu}^{i}} \tensor{C}{^{\nu j}} + \underbrace{4\tensor{C}{_{\mu}^{i}} \tensor{C}{^{\mu}^{j}} \tensor{C}{^{\nu}^{i}} \tensor{Q}{_{\nu}^{j}}}_{\mathcal{O}\left( Q \right) } + 2\tensor{C}{_{\mu}^{i}} C^{\nu i} \tensor{Q}{_{\nu}^{j}} \tensor{Q}{^{\mu j}}\nonumber \\
    &\phantom{=}~+ 2\tensor{Q}{_{\mu}^{i}} C^{\nu i} \tensor{Q}{_{\nu}^{j}} C^{\mu j} + 2\tensor{Q}{_{\mu}^{i}} C^{\nu i} \tensor{C}{_{\nu}^{j}} Q^{\mu j} + \mathcal{O}\left( Q^3 \right).
    \intertext{Combining these terms, we see that the term underneath the square root is given by}
    \implies 2P^2 - S^2 &= 2P_C^2 -S_C^2 +\underbrace{ 8\tensor{C}{_{\mu}^{i}} \tensor{C}{^{\mu}^{j}} \tensor{C}{^{\nu}^{i}} \tensor{Q}{_{\nu}^{j}}-4S_C \tensor{C}{_{\mu}^{i}} \tensor{Q}{^{\mu i}} }_{\mathcal{O}\left( Q \right) } - 2 S_C S_Q  - 4\tensor{C}{_{\mu}^{i}}Q^{\mu i} \tensor{C}{_\nu^{j}} Q^{\nu j}  \nonumber\\
    &\phantom{=}~+ 4\tensor{C}{_{\mu}^{i}} C^{\nu i} \tensor{Q}{_{\nu}^{j}} \tensor{Q}{^{\mu j}} + 4\tensor{Q}{_{\mu}^{i}} C^{\nu i} \tensor{Q}{_{\nu}^{j}} C^{\mu j} + 4\tensor{Q}{_{\mu}^{i}} C^{\nu i} \tensor{C}{_{\nu}^{j}} Q^{\mu j}
.\end{align}
Taylor expanding the square root about the solely background dependent terms $2P_C^2 - S_C^2$ we see that
\begin{align}
    \sqrt{2P^2 - S^2} &\equiv \sqrt{2P_C^2 - S_C^2 + \mathcal{Q}} = \sqrt{2P_C^2 - S_C^2}  + \frac{\mathcal{Q}}{2 \sqrt{S_C^2 + P_C^2} } - \frac{\mathcal{Q}^2}{8\left( S_C^2 + P_C^2 \right)^{\frac{3}{2}}} + \mathcal{O}\left( \mathcal{Q}^3 \right),
    \intertext{where $\mathcal{Q}$ is not to be confused with the quantum field $Q^{i}$ and is defined such that}
    \mathcal{Q} &= \underbrace{ 8\tensor{C}{_{\mu}^{i}} \tensor{C}{^{\mu}^{j}} \tensor{C}{^{\nu}^{i}} \tensor{Q}{_{\nu}^{j}}-4S_C \tensor{C}{_{\mu}^{i}} \tensor{Q}{^{\mu i}} }_{\mathcal{O}\left( Q \right) } - 2 S_C S_Q  - 4\tensor{C}{_{\mu}^{i}}Q^{\mu i} \tensor{C}{_\nu^{j}} Q^{\nu j}  \nonumber\\
    &\phantom{=}~+ 4\tensor{C}{_{\mu}^{i}} C^{\nu i} \tensor{Q}{_{\nu}^{j}} \tensor{Q}{^{\mu j}} + 4\tensor{Q}{_{\mu}^{i}} C^{\nu i} \tensor{Q}{_{\nu}^{j}} C^{\mu j} + 4\tensor{Q}{_{\mu}^{i}} C^{\nu i} \tensor{C}{_{\nu}^{j}} Q^{\mu j},
    \intertext{and up to terms quadratic in the quantum field we have}
    \mathcal{Q}^2 &=  16S_C^2 \tensor{C}{_{\mu}^{i}}Q^{\mu i} \tensor{C}{_{\nu}^{j}} Q^{\nu j} - 32 S_C \tensor{C}{_{\mu}^{i}}Q^{\mu i} \tensor{C}{_{\nu}^{j}} C^{\nu k} C^{\rho j} \tensor{Q}{_{\rho}^{k}}\nonumber \\
                  &\phantom{=}~+ 64 \left( \tensor{C}{_{\nu}^{j}} C^{\nu k} C^{\rho j} \tensor{Q}{_{\rho}^{k}} \right)^2 + \mathcal{O}\left( Q^3 \right) 
.\end{align}

Therefore with the classical field Lagrangian defined by
\begin{align}
    \mathcal{L}_C &\equiv \frac{\cosh \gamma}{2}S_C  + \frac{\sinh \gamma}{2}\sqrt{2P_C^2 - S_C^2},
    \intertext{we can thus write the full Lagrangian as}
    \mathcal{L} &= \mathcal{L}_{C} + \mathcal{L}_Q\nonumber \\
    &= \mathcal{L}_{C} + \frac{\cosh \gamma}{2}S_Q   +  \frac{\sinh \gamma}{2} \left( \frac{\mathcal{Q}}{2 \sqrt{2P_C^2 - S_C^2} } - \frac{\mathcal{Q}^2}{8 \left( 2P_C^2 - S_C^2 \right)^{\frac{3}{2}} } \right) 
.\end{align}

The Lagrangian for the quantum field thus suggests the quantum field has a free kinetic term, $S_Q \cosh \gamma$, as well as an interaction vertex obtained perturbatively.

We notice that we can express the quantum Lagrangian in the form 
\begin{align}
    \mathcal{L}_Q &= Q^{\mu i}\left( \frac{\cosh\gamma}{2}g_{\mu \nu} \delta^{ij} +  \tensor{P}{_{\mu \nu}}^{ij}\right)Q^{\nu j},
    \intertext{or equivalently, using integration by parts to move the derivative in $Q^{\mu i} = \partial^\mu Q^i$, as}
                 \mathcal{L}_Q &=  - Q^{i} \left( \frac{\cosh \gamma}{2} \delta^{ij} \partial^2 +  \sinh \gamma \partial^{\mu} \tensor{P}{_{\mu \nu}^{ij}}\partial^\nu + \sinh \gamma \tensor{P}{_{\mu \nu}^{ij}} \partial^\mu \partial^\nu\right)  Q^{j}, \label{eq:SFQuantum_Lagrangian} 
                  \intertext{where}
    \tensor{P}{_{\mu \nu}^{ij}} &=  - \bigg( \frac{-2S_C g_{\mu \nu} \delta^{ij} - 4 \tensor{C}{_{\mu}^{i}} \tensor{C}{_{\nu}^{j}} + 4\tensor{C}{_{\mu}^{k}} \tensor{C}{_\nu^{k}} \delta^{ij} + 4 \tensor{C}{_{\mu}^{j}} \tensor{C}{_{\nu}^{i}} + 4 \tensor{C}{_{\rho}^{i}}C^{\rho j} g_{\mu \nu}}{4 \sqrt{2P_C^2 - S_C^2} }\nonumber \\
                                &\phantom{=}\hspace{1.6cm} -\frac{16S_C^2 \tensor{C}{_{\mu}^{i}} \tensor{C}{_{\nu}^{j}} - 32 S_C \tensor{C}{_{\mu}^{i}} \tensor{C}{_{\rho}^{k}} \tensor{C}{^{\rho}^{j}}\tensor{C}{_{\nu}^{k}} + 64 \tensor{C}{_{\rho}^{k}} \tensor{C}{^{\rho i}}\tensor{C}{_{\mu}^{k}} \tensor{C}{_{\tau}^{m}} \tensor{C}{^{\tau j}} \tensor{C}{_{\nu}^{m}}} {16 \left( 2 P_C^2 - S_C^2 \right)^{\frac{3}{2}}} \bigg) 
,\end{align}
and we consider $\tensor{P}{_{\mu \nu}^{ij}}(x)$ to be a composite operator representing the cumulative effect of the classical field. We can interpret the first term in \cref{eq:SFQuantum_Lagrangian} as a kinetic term for the field $Q^{i}$ and the second and third terms encode the interactions of the classical (through $\tensor{P}{_{\mu \nu}^{ij}}$) and quantum fields $Q^{i}$ induced by the perturbative expansion.

\section{Feynman Rules}
The propagator for the quantum field is
\begin{align}
        \tensor{D}{^{ij}}  &= \frac{1}{\cosh \gamma} \frac{-i \delta^{ij}}{k^2 } \label{eq:SFpropagator}
.\end{align}

I will draw quantum fields as solid lines and the cumulative effect of the background fields as a single coiled line. 

Reading off the Lagrangian \cref{eq:SFQuantum_Lagrangian}, consider first contracting $Q^{i}$ with the incoming quantum field with momenta $p$. This means $Q^{j}$ will contract with the momenta $q$ field. The $ Q^{i}\tensor{P}{_{\mu \nu}^{ij}} \partial^{\mu} \partial^{\nu} Q^{j}$ term will thus carry $q^{\mu} q^{\nu}$ and the $Q^{i} \partial^{\mu} \tensor{P}{_{\mu \nu}^{ij}} \partial^{\nu} Q^{j}$ term will contribute $r^{\mu} q^{\nu}$ as the first derivative now acts on the classical field that has momenta $r$.

Performing this identically for the other possible contraction, we see that the interaction vertex, which we call $g^{ij} $, takes the form
\begin{align} \label{eq:SFvertex}
     g^{ij} &\equiv \vcenter{\hbox{
%\begin{tikzpicture}
%    \begin{feynman}
%    \diagram [small,horizontal=a to d] {
%        e [dot] -- [gluon,rmomentum'=$r\,\,\mathrm{=}-p-q$] c [crossed dot,label=2:$\mu\text{,~}\nu\text{;~}m\text{,~}n$];
%        a -- [momentum=$i\text{, }p$] e,  
%        b -- [momentum=$j\text{, }q$] e, 
%    };
%    \end{feynman}
%\end{tikzpicture}
             \includegraphics{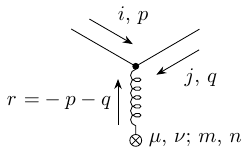}
}} \\
   g^{ij} &=i \tensor{P}{_{\mu \nu}^{mn}} \left( \delta^{im} \delta^{jn}q^{\mu} q^{\nu} + \delta^{in} \delta^{jm} p^{\mu} p^{\nu} \right) \nonumber \\
    &\phantom{=}~ + i \tensor{P}{_{\mu \nu}^{mn}}\left( \delta^{im} \delta^{jn}r^{\mu} q^{\nu}  + \delta^{in} \delta^{jm} r^{\mu} p^{\nu} \right) \nonumber \\
    &= i\sinh \gamma \tensor{P}{_{\mu \nu}^{mn}}\left( \delta^{im} \delta^{jn}\left( q^{\mu} + r^{\mu}  \right) q^{\nu}  + \delta^{in} \delta^{jm} \left( p^{\mu} + r^{\mu} \right)  p^{\nu} \right) \nonumber
    \intertext{where using $r^{\mu} = -p^{\mu} - q^{\mu}$ yields}
   g^{ij} &= -i\sinh \gamma \tensor{P}{_{\mu \nu}^{mn}}\left( \delta^{im} \delta^{jn} p^{\mu}q^{\nu}  + \delta^{in} \delta^{jm} q^{\mu} p^{\nu} \right) \\
    &= -i \sinh \gamma \left( \tensor{P}{_{\mu \nu}^{ij}} p^{\mu} q^{\nu} + \tensor{P}{_{\mu \nu}^{ji}} q^{\mu} p^{\nu} \right) 
    \intertext{and as $\tensor{P}{_{\mu \nu}^{ij}} = \tensor{P}{_{\nu \mu}^{ji}}$, we can equivalently write this as}
   g^{ij} &= -2i \sinh \gamma \tensor{P}{_{\mu \nu}^{ij}}  p^{\mu} q^{\nu}
.\end{align}

\begin{note}
    Formally, one should also integrate over the classical field momenta $r^{\mu}$ in this vertex factor, however for clarity and concision it is omitted. Similarly, the $\tensor{P}{_{\mu \nu}^{mn}}$ factor will often be excluded from diagrammatic evaluation for brevity such that we are using an uncontracted vertex factor $\left( \widetilde{g}_{ij}^{mn} \right)^{\mu \nu} $ defined such that formally
    \begin{align}
        g^{ij} &= \int \dd{^{d}r} \tensor{P}{_{\mu \nu}^{mn}}\left( r \right)  \left( \widetilde{g}_{mn}^{ij} \right)^{\mu \nu},
        \intertext{where}
\left( \widetilde{g}_{ij}^{mn} \right)^{\mu \nu} &= -2i \sinh \gamma \delta^{m}_i \delta^{n}_j p^{\mu} q^{\nu}
    .\end{align}
\end{note}

\section{Background-Varying One-Loop Diagrams}
As the Feynman rules are entirely analogous to the 4D ModMax case, solely with the addition of indices $i,j \in \{1,\cdots,N\} $ that sum over bosons, we can conclude immediately that if the background field is held constant, then the one loop effective action will vanish. As such, we move to the more general case, in which we do not impose, $\partial_\mu \tensor{C}{_{\nu}^{i}}= 0$, thus allowing the background field to arbitrarily vary.

The first diagram in the perturbative expansion is
\begin{align}
    \mathcal{D}_1 &= \vcenter{\hbox{
%\begin{tikzpicture}
%    \begin{feynman}
%    \diagram [layered layout, vertical=a to c] {
%        a -- [ half left, looseness=1.5, rmomentum=$\ell$] b [dot], 
%        a -- [half right, looseness=1.5] b , % rmomentum'=$k\text{, }\mu$
%        b -- [gluon,rmomentum=$p$] c [crossed dot, label={2:$\mu,~\nu;m,~n$}]; %\tensor{P}{_{\mu \nu}^{mn}} 
%    };
%    \end{feynman}
%\end{tikzpicture}
            \includegraphics{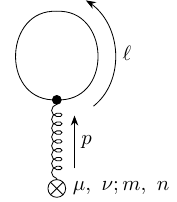}
}} \nonumber \\
   \mathcal{D}_1 &= \sinh \gamma \delta^{im} \delta^{jn} \int \frac{\dd{^{d}\ell}}{\left( 2\pi \right)^{d}} \left( 2i \ell^{\mu} \ell^{\nu} \right) D^{ij},
    \intertext{where the replacement $\ell^{\mu} \ell^{\nu} \to \frac{\ell^2}{d} g^{\mu \nu}$ and the propagator derived above yield}
   \mathcal{D}_1 &= -\frac{2\tanh \gamma}{d} \delta^{im} \delta^{in} \int \frac{\dd{^{d}\ell}}{\left( 2\pi \right)^{d}} g^{\mu \nu}
,\end{align}
which vanishes in dimensional regularization. The first nontrivial diagram is
\begin{align} \label{D2_diagram}
    \mathcal{D}_2 \; = \vcenter{\hbox{
%\begin{tikzpicture}
%    \begin{feynman}
%    \diagram [layered layout, horizontal=a to b] {
%        z [crossed dot,label=$\mu\nu;ij$] -- [gluon, momentum=$q$] a [dot]-- [ half left, looseness=1.5, rmomentum=$\ell;k$] b [dot], 
%        a -- [half right, looseness=1.5, momentum'=$\ell + q;l$] b ,
%        b -- [gluon,rmomentum=$q$] c [crossed dot,label=$\rho\tau;mn$];
%    };
%    \end{feynman}
%\end{tikzpicture}
            \includegraphics{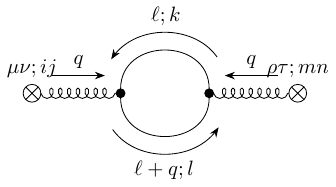}
}} \; ,
\end{align}
which we can express in momentum space as
\begin{align}
    \mathcal{D}_2 =-4\tanh^2 \gamma \int \frac{\dd{^{d}q}}{\left( 2\pi \right)^{d}} \tensor{P}{_{\mu \nu}^{ij}}\left( -q \right) \mathcal{I}^{\mu \nu \rho \tau}_2  \tensor{P}{_{\rho \tau}^{ij}}\left( q \right) \label{I2_to_D2}
,\end{align}
where all the loop dependence is contained within the simpler integral
\begin{align}
    \mathcal{I}^{\mu \nu \rho \tau}_2 = \int \frac{\dd{^{d}\ell}}{\left( 2\pi \right)^{d}} \frac{\left( \ell + q \right)^{\nu}\ell^{\mu} \left( \ell + q \right)^{\tau} \ell^{\rho}}{\ell^2 \left( \ell + q \right)^2} \label{simpler_integral_body}
.\end{align}

To study the structure of the divergence of $\mathcal{D}_2$ it suffices to evaluate $\mathcal{I}^{\mu \nu \rho \tau}_2$. The divergent contribution with $d=2+2\epsilon$ is analogous to the ModMax case with
\begin{align} 
    \mathcal{I}^{\mu \nu \rho \tau}_2 &= \left( \frac{1}{\epsilon} \right) \frac{-i}{24\left( 4\pi \right)} \bigg[\nonumber\\
    &\phantom{=}q^{2}  \left( g^{\mu \nu} g^{\rho \tau}  + g^{\mu \rho} g^{\nu \tau} + g^{\mu \tau} g^{\nu \rho} \right)\nonumber \\
    &\phantom{=}~+2 \left( q^{\nu} q^{\mu} g^{\tau \rho} + g^{\mu \tau} q^{\nu} q^{\rho}+ g^{\nu \mu} q^{\tau} q^{\rho} + g^{\nu \rho} q^{\mu} q^{\tau} \right) \nonumber \\
    &\phantom{=}~+4\left( g^{\mu \rho} q^{\nu} q^{\tau} + g^{\nu \tau} q^{\mu} q^{\rho} \right) \nonumber\bigg] \label{divergent_final_body}
.\end{align}

For the full dimensional regularization calculation of this diagram see \cref{sec:scalar_field_1_loop_2_vertex_calc}.

Therefore, with this integral evaluated, we see from \cref{I2_to_D2} that the power of the momentum dependence of $\mathcal{D}_2$ can be captured (ignoring index structure) with
    \begin{align}
        \mathcal{D}_2 \sim \left( \frac{1}{\epsilon} \right) \int \frac{\dd{^{d}q}}{\left( 2\pi \right)^{d}} \tensor{P}{_{\mu \nu}}^{ij}\left( -q \right) q^2 \tensor{P}{_{\rho \tau}^{ij}}\left( q \right) 
    .\end{align}

    However, there is no such $\sim q^2\left(\tensor{P}{_{\mu \nu}^{ij}} \right)^2$ term in the original Lagrangian, which bodes poorly for renormalization. Namely, as $\frac{1}{\epsilon}$ diverges in the limit $\epsilon \to 0$, to obtain finite predictions from this theory, one would introduce a counter term to the Lagrangian, which removes this divergence. When the corrections are of the form of the original Lagrangian, this is physically well motivated as it corresponds to a redefinition of constants in the Lagrangian. However, as the form of this correction is not present in the original Lagrangian, this interpretation does not apply as the corrections are novel.

    Nonetheless, I have obtained the divergent part of the one-loop effective action up to order $\mathcal{O}\left( \gamma^2 \right) $. Truncating at an arbitrary order in $\gamma$ is well motivated due to the small experimental bound of $\gamma \leq 3 \times 10^{-22}$, however for completeness we proceed with the generalization of the argument developed to an $n$-vertex diagram.

\section{\texorpdfstring{$n$}{n} Vertex Diagram}

%Rather than truncating our calculation at one loop, we can notice that the order/power of $\gamma$ that a given diagram contributes is identical to the number of vertices in the diagram, as to first order the vertex contributes $\sinh \gamma \overset{\mathcal{O}\left( \gamma \right) }{=} \gamma$ and each propagator has no effect with $\cosh \gamma \overset{\mathcal{O}\left( \gamma \right) }{=} 1$.
%
%As such, we find that by evaluating all diagrams containing two vertices (and any number of loops), we can find the all-loops effective action to order $\mathcal{O}\left( \gamma^2 \right)$.

With the two-vertex diagram evaluated, to complete the one-loop effective action, we seek to evaluate all remaining diagrams containing one loop. Fortunately, there is only one diagram, $\mathcal{D}_n$, for each number of vertices $n$. As such, we proceed with the generalization of the above method. 

For an $n$ vertex diagram, we label the external momenta as $q_{i}$ for $i \in \left\{ 0,\cdots,n-1 \right\}$ with momentum conservation implying
\begin{align}
    q_{n-1} = - \sum_{i=0}^{n-2} q_{i}
.\end{align}

As we did with $\mathcal{D}_2$ in \cref{simpler_integral_body}, we will strip off various factors of $\tensor{P}{_\mu_\nu^i^j}$ to write
\begin{align}
    \mathcal{D}_n &= \frac{\left( - 2 \tanh ( \gamma ) \right)^n}{n} \int \left( \prod_{j=0}^{n-2} \frac{\dd{^{d}q_{j}}}{\left( 2\pi \right)^{d}} \tensor{P}{_{\alpha_{2j} \alpha_{2j+1}}^{i_{j} i_{j+1}}}\left( q_j \right) \right) \tensor{P}{_{\alpha_{2n-2} \alpha_{2n-1}}^{i_{n-1} i_0}}\left( q_{n-1} \right) \left( {\mathcal{I}}_n \right)^{\alpha_0 \ldots \alpha_{2n-1}} \label{Dn_defn}
,\end{align}
where ${\mathcal{I}}_n$ is the simpler integral
\begin{align}
    \left( {\mathcal{I}}_n \right)^{\alpha_0 \ldots \alpha_{2n-1} } &= \int \frac{d^d \ell}{( 2\pi )^d} \, \left( \prod_{i = 0}^{n-1} \left( \ell + \sum_{j=0}^{i} q_{j}  \right)^{-2} \right) \left( \prod_{k=0}^{n-1} \left( \ell + \sum_{j=0}^{k-1} q_{j} \right)^{\alpha_{2k}} \left( \ell + \sum_{j=0}^{k-1} q_{j}  \right)^{\alpha_{2k+1}} \right) \label{In_defn}
.\end{align}
We will further break up $\mathcal{I}_n$ into pieces and evaluate each piece in turn. Let us write the integrand of \cref{In_defn} as a product of propagators and vertex factors such that
\begin{align}
    \left( {\mathcal{I}}_n \right)^{\alpha_0 \alpha_1 \ldots \alpha_{2n-1} } = \int \frac{d^d \ell}{( 2\pi )^d} \mathcal{P}_n \mathcal{V}^{\alpha_0 \alpha_1 \ldots \alpha_{2n-1} },  \label{In_PV}
\end{align}
where the propagators are captured by the symbol
\begin{align}
    \mathcal{P}_n &= \prod_{i = 0}^{n-1} \left( \ell + \sum_{j=0}^{i} q_{j}  \right)^{-2} ,
\end{align}
and the vertex factors are captured by $\mathcal{V}^{\alpha_0 \alpha_1 \ldots \alpha_{2n-1} }$.

We begin by simplifying the propagators. Observe that in general \cite{Renormalization, Schroeder}, we can write the product of $n$ propagators using a Feynman parameterization as
\begin{align}
    \prod_{i=0}^{n-1} A_{i}^{-1} &= \int_0^{1} \left( \prod_{i=0}^{n-1} \dd{x}_{i} \right) \delta \left( \sum_{i=0}^{n-1} x_{i} - 1  \right) \frac{\left( n - 1 \right)!}{\left[ \sum_{i}^{}x_{i} A_{i}  \right]^{n} }
.\end{align}

The product of propagators inside the loop can thus be expressed as
\begin{align}
    \mathcal{P}_n = \left( n - 1 \right)! \int_0^{1} \left( \prod_{i=0}^{n-1} \dd{x}_{i} \right) \delta \left( \sum_{i=0}^{n-1} x_{i} - 1  \right) \left[ \sum_{i}^{}x_{i} \left( \ell + \sum_{j=0}^{i} q_{j} \right)^2   \right]^{-n},
\end{align}
    where as $\sum_{i}^{} x_{i} = 1$, we can expand and reduce the square bracketed term to

\begin{align}
    &\left[ \sum_{i}^{}x_{i} \left( \ell + \sum_{j=0}^{i} q_{j} \right)^2   \right]^{-n} \nonumber \\
    &=  \left[ \ell^2  + \sum_{i}^{}x_{i} \left( 2\ell_{\mu} \sum_{j=0}^{i} q_{j}^{\mu}+ \left( \sum_{j=0}^{i} q_{j} \right)^2  \right)   \right]^{-n}  \nonumber\\
    &=  \left[ \left( \ell + \sum_{i}^{} \sum_{j=0}^{i} x_{i} q_{j}  \right)^2 - \left( \sum_{i}^{} \sum_{j=0}^{i} x_{i} q_{j} \right)^2  + \sum_{i}^{}x_{i}   \left( \sum_{j=0}^{i} q_{j} \right)^2  \right]^{-n},
\end{align}
    where we have completed the square. Further, under the translation $\ell^{\mu} \to \ell^{\mu} - \sum_{i}^{} \sum_{j=0}^{i} x_{i} q_{j}^{\mu}$ becomes
\begin{align}
  \left[ \sum_{i}^{}x_{i} \left( \ell + \sum_{j=0}^{i} q_{j} \right)^2   \right]^{-n}  &=  \left[ \ell^2 - \left( \sum_{i}^{} \sum_{j=0}^{i} x_{i} q_{j} \right)^2  + \sum_{i}^{}x_{i}   \left( \sum_{j=0}^{i} q_{j} \right)^2  \right]^{-n} \\
    &=  \left[ \ell^2 - \Delta^2  \right]^{-n},
    \intertext{where we have defined}
    \Delta^2 &\equiv \left( \sum_{i}^{} \sum_{j=0}^{i} x_{i} q_{j} \right)^2  - \sum_{i}^{}x_{i}   \left( \sum_{j=0}^{i} q_{j} \right)^2
.\end{align}

We now turn our attention to the vertex factors as captured by $\mathcal{V}^{\alpha_0 \alpha_1 \cdots \alpha_{2n-1}}$. Each vertex factor yields two factors of momenta in the numerator. As each external vertex around the loop adds $q_{j}$, and the vertex factor is the product of the momenta entering and leaving the vertex, in total we will have
\begin{align}
   \mathcal{V}^{\alpha_0 \alpha_1 \cdots \alpha_{2n-1}} &= \prod_{k=0}^{n-1} \left( \ell + \sum_{j=0}^{k} q_{j} \right)^{\alpha_{2k}} \left( \ell + \sum_{j=0}^{k} q_{j}  \right)^{\alpha_{2k+1}},
    \intertext{where the $k=0$ term gives us the $\ell$ terms. Under the translation identified for the denominator to be quadratic in $\ell$, this numerator is translated to}
   \mathcal{V}^{\alpha_0 \alpha_1 \cdots \alpha_{2n-1}} &= \prod_{k=0}^{n-1} \left( \ell - \sum_{i=0}^{n-1} \sum_{j=0}^{i} x_{i} q_{j}  + \sum_{j=0}^{k} q_{j} \right)^{\alpha_{2k}} \left( \ell - \sum_{i=0}^{n-1} \sum_{j=0}^{i} x_{i} q_{j}+ \sum_{j=0}^{k} q_{j}  \right)^{\alpha_{2k+1}}.
    \intertext{We expand this product in descending powers of $\ell$ as only powers $\ell^{2n}$ and $\ell^{2n-2}$ will lead to divergent terms. We have that}
   \mathcal{V}^{\alpha_0 \alpha_1 \cdots \alpha_{2n-1}} &= \prod_{k=0}^{n-1} \ell^{\alpha_{2k}} \ell^{\alpha_{2k+1}} + \sum_{a=0}^{2n-1} \sum_{b>a}^{2n-1} \left( \prod_{c \neq a,b}^{2n-1} \ell^{\alpha_c} \right)  f\left( x,q,a \right)^{\alpha_{a}} f\left( x,q,b  \right)^{\alpha_{b}} + \mathcal{O}\left( \ell^{2n-4} \right)
,\end{align}
where we have defined for brevity
\begin{align}
    f\left( x,q,a \right)^{\alpha_a} &\equiv \left( \sum_{i=0}^{n-1} \sum_{j=0}^{i} x_{i} q_{j}  + \sum_{j=0}^{\left\lfloor  \frac{a+1}{2}\right\rfloor} q_{j} \right)^{\alpha_{a}}
.\end{align}
% Given a=2k or 2k-1 we need to get k
% floor (a+1)/2 achieves this
Next we will replace products of uncontracted loop momenta using the generalized symmetrization rule of equation (\ref{eq:symmetrization}). To ease notation, let us write
\begin{align}
    g^{\alpha_0 \ldots \alpha_{2n-1}} = g^{(\alpha_0 \alpha_1 } \cdots g^{\alpha_{2n-2} \alpha_{2n-1})} \label{shorthand_metric}
,\end{align}
for the symmetrized combination of derivatives appearing in this expression. When no confusion is possible, we will also write $g^{\{ \alpha \}}$ for (\ref{shorthand_metric}), where $\{ \alpha \}$ is understood to refer to a multi-index $\{ \alpha \} = \alpha_0 \ldots \alpha_{2n-1}$. With this notation, the replacement rule becomes
\begin{align}
    \prod_{i=1}^{n} \ell^{\alpha_{2i-1}} \ell^{\alpha_{2i}} \to \frac{\ell^{2n} \left( d -2 \right)!! (2n-1)!!}{\left( d-2 + 2n \right)!!} g^{\alpha_{0}\ldots \alpha_{2n-1}}
.\end{align}
This transforms the vertex factors into
\begin{align}
    \mathcal{V}^{\alpha_0 \alpha_1 \ldots \alpha_{2n-1} } &=  \frac{\ell^{2n} \left( d - 2 \right)!!(2n-1)!!}{ \left( d  - 2 + 2n \right)!!} g^{\alpha_0 \cdots \alpha_{2n-1}} \nonumber\\
                                                          &\phantom{=}~+ \frac{\ell^{2n-2} \left( d - 2 \right)!! \left( 2n -3 \right)!!}{\left( d - 4 + 2n \right)!!}\sum_{a=0}^{2n-1} \sum_{b>a}^{2n-1} g^{\{ \alpha \neq \alpha_a, \alpha_b \}}   f^{\alpha_a} \left( x,q,a \right)  f^{\alpha_{b}} \left( x,q,b \right) \nonumber\\
                                                          &\phantom{=}~+ \mathcal{O}\left( \ell^{2n-4} \right) \label{vertex_intermediate}
,\end{align}
where $g^{\{ \alpha \neq \alpha_a, \alpha_b \}}$ refers to symmetrization over all possible index values except $\alpha_a$ and $\alpha_b$.

Therefore, recombining the pieces, we can express the divergent terms in $\mathcal{D}_n$ as
\begin{align}\hspace{-10pt}
    &\left( {\mathcal{I}}_{n} \right)^{\alpha_0 \alpha_1 \ldots \alpha_{2n-1}} \nonumber \\
    &\quad = \int \frac{d^d \ell}{( 2\pi )^d} \mathcal{P} \mathcal{V}^{\alpha_0 \alpha_1 \ldots \alpha_{2n-1}} \nonumber \\
    &\quad = \int \frac{d^d \ell}{( 2\pi )^d} \left( n - 1 \right)!\int_0^{1} \left( \prod_{i=0}^{n-1} \dd{x}_{i} \right) \delta \left( \sum_{i=0}^{n-1} x_{i} - 1  \right) \left( \ell^2 - \Delta^2  \right)^{-n}  \cdot \Bigg( \frac{ \ell^{2n} \left( d - 2 \right)!! \left( 2n -1 \right)!!}{\left( d  - 2 + 2n \right)!!} g^{\alpha_0 \ldots \alpha_{2n-1}} \nonumber \\
    &\quad \qquad + \frac{\ell^{2n-2} \left( d - 2 \right)!! \left( 2n-3 \right)!!}{\left( d - 4+ 2n  \right)!!}\sum_{a=0}^{2n-1} \sum_{b>a}^{2n-1} g^{\{ \alpha \neq \alpha_a, \alpha_b \}}   f^{\alpha_a} \left( x,q,a \right)  f^{\alpha_{b}} \left( x,q,b \right)  + \mathcal{O}\left( \ell^{2n-4} \right) \Bigg) \label{Dn_Itilde_intermediate}
,\end{align}
and evaluate them with the known integral
\begin{align}
    \int \frac{\dd{^{d}\ell}}{\left( 2\pi \right)^{d}} \frac{\ell^{2\beta}}{\left( \ell^2 - \Delta^2 \right)^{\alpha}} &= i \left( -1 \right)^{\alpha + \beta} \frac{\Gamma \left( \beta + \frac{d}{2} \right) \Gamma \left( \alpha - \beta - \frac{d}{2} \right) }{\left( 4\pi \right)^{\frac{d}{2}} \Gamma \left( \alpha \right) \Gamma \left( \frac{d}{2} \right)  } \Delta^{2 \left( \frac{d}{2} - \alpha + \beta \right) }
.\end{align}
We see that matching the powers of $\ell$ between \cref{Dn_Itilde_intermediate} and the known integral, the two numerator terms in \cref{Dn_Itilde_intermediate} have $\beta = n$ and $\beta = n - 1$ respectively. The denominator is shared and possesses $\alpha = n$. From this form it is clear to see that $\ell^{2n-4}$ terms and lower powers of $\ell$ are finite as $d \to 2$. 
\begin{proof}
    Such terms have $\beta = n -2$ which would contain $\Gamma \left( \alpha - \beta - \frac{d}{2} \right) $ terms of the form
    \begin{align}
    \Gamma \left( \left(n  \right) - \left( n - 2 \right) - \frac{d}{2} \right) &= \Gamma \left( 2 - \frac{d}{2} \right)  \\
    &= \left( 1 - \frac{d}{2} \right) \Gamma \left( 1 - \frac{d}{2} \right)  \\
    &= \left( 1 - \frac{d}{2} \right) \left( -\frac{d}{2} \right)  \Gamma \left( -\frac{d}{2} \right),
    \intertext{which using $d = 2 + 2\epsilon$ and the pole expansion of the $\Gamma$ function we see that}
    &= -\left( -\epsilon \right) \left( 1 + \epsilon \right)  \left( -\frac{1}{\epsilon} - \gamma + 1 \right) \\
    &= \left( 1 + \epsilon \right)  \left( -1 - \epsilon\gamma + \epsilon \right)\\
    &= -1 + \mathcal{O}\left( \epsilon \right) 
,\end{align}
which are finite and thus discarded. Lower powers of $\ell$ lead to larger $\alpha - \beta - \frac{d}{2}$ values which still contain the $1 - \frac{d}{2} \to -\epsilon$ factor which cancels the poles.\\
\end{proof}

As such focusing on the divergent terms, we have that upon applying the known integral, the $\ell^{2n}$ term with $\beta = n$ evaluates to
\begin{align}
    \vb{C}_{2n}^{\alpha_0 \ldots \alpha_{2n -1}} &\equiv \frac{i \left( d - 2 \right)!! \left( 2n -1 \right)!!}{ \left( d - 2 + 2n \right)!!} g^{\alpha_0 \cdots \alpha_{2n-1}}  \frac{\Gamma \left( n + \frac{d}{2} \right)  }{\left( 4\pi \right)^{\frac{d}{2}} \Gamma \left( n \right) \Gamma \left( \frac{d}{2} \right)} \Gamma \left( -\frac{d}{2} \right)\Delta^{d},
     \end{align}
and the $\ell^{2n-2}$ term has $\beta = n - 1$ and evaluates to
\begin{align}
    \vb{D}_{2n}^{\alpha_0 \ldots \alpha_{2n-1}} &\equiv \frac{i \left( d - 2 \right)!! \left( 2n-3 \right)!!}{\left( d - 4 + 2n \right)!!}\sum_{a=0}^{2n-1} \sum_{b>a}^{2n-1} g^{\{ \alpha \neq \alpha_a, \alpha_b \}}   f^{\alpha_a} \left( x,q,a \right)  f^{\alpha_{b}} \left( x,q,b \right)\nonumber \\
    &\qquad \cdot \frac{\Gamma \left( n-1+\frac{d}{2} \right) }{\left( 4\pi \right)^{\frac{d}{2}} \Gamma \left( n \right) \Gamma \left( \frac{d}{2} \right)  } \Gamma \left( 1 - \frac{d}{2} \right) \Delta^{d - 2}
,\end{align}
where we can thus quote the whole $n$ vertex diagram result as
\begin{align}
     \left( \mathcal{I}_n \right)^{\alpha_0 \ldots \alpha_{2n-1}} = \frac{\left( n - 1 \right)!}{n}\int_0^{1} \left( \prod_{i = 0}^{n-1} \dd{x}_{i} \right) \delta \left( \sum_{i=0}^{n - 1} x_{i} - 1 \right) \left[ \vb{C}_{2n}^{\alpha_0 \ldots \alpha_{2n-1}} + \vb{D}_{2n}^{\alpha_0 \ldots \alpha_{2n-1}} \right]
.\end{align}

In both $\vb{C}_{2n}$ and $\vb{D}_{2n}$ we have a polynomial in $x_{i}$ of order $d$ and a total power of $q^{d}$. In $\vb{C}_{2n}$, this $q^{d}$ momentum dependence is contained within $\Delta^{d}$, and for $\vb{D}_{2n}$, we have two uncontracted factors $q^{\alpha_a}$ and $q^{\alpha_b}$ on top of the $q^{d-2}$. Therefore, we conclude that in the limit of $d \to 2$, all such $n$ vertex diagrams have the same general structure as the $2$ vertex diagram with 
\begin{align}
    \vb{C}_{2n}^{\alpha_0 \ldots \alpha_{2n-1}} &\propto \frac{1}{\epsilon} q^2 g^{\alpha_0 \ldots \alpha_{2n-1}}  \\
    \vb{D}_{2n}^{\alpha_0 \ldots \alpha_{2n-1}} &\propto \frac{1}{\epsilon}\sum_{a=0}^{2n-1} \sum_{b>a}^{2n-1} q^{\alpha_a} q^{\alpha_b} g^{\{\alpha \neq \alpha_a, \alpha_b\} }
,\end{align}
a familiar $\left( q^2 g^{\mu \nu} - q^{\mu} q^{\nu} \right)$-like dependence, where the $\frac{1}{\epsilon}$ arises from the $\Gamma \left( -\frac{d}{2} \right) $ and $\Gamma \left( 1 - \frac{d}{2} \right) $ factors respectively. This momentum structure is shrouded in the index symmetrization and $x_{i}$ dependence in the full expression. As expected, while this momentum dependence is familiar, the classical field dependence that arises from the external vertices is not present in the original Lagrangian. As each $n$ vertex diagram is divergent and has a different external field dependence (i.e. an extra external field $\tensor{P}{_{\mu\nu}^{ij}}\left( x \right) $), one would need to add an infinite number of terms to the Lagrangian to subtract each of these divergences. This property of requiring an infinite number of quantities to be fixed to obtain finite results is referred to as \textit{non-renormalizability}. Such a theory cannot be predictive at all energy scales and is thus not fundamental.

Nonetheless, ModMax and its two dimensional analogue can still be considered on the quantum domain as \textit{effective field theories} where one necessarily identifies a maximum energy scale of applicability. We call such a theory effective, as it is a valid low energy description, but not the true fundamental theory of the system. For example, fermions have mass in the standard model, but this is in fact an effective low energy description where the Higgs interaction which generates such mass terms has been integrated out \cite{Schroeder}.

As such, we notice that all such $n$ vertex diagrams are proportional to $\tanh^{n} \gamma \sim  \gamma^{n}$. Motivated by the small experimental bound on $\gamma$, we seek to characterize all diagrams contributing up to order $\mathcal{O}\left( \gamma^2 \right)$, which is exactly all two vertex diagrams.

\section{Two Vertex $n$-loop Diagrams}

Until this point, we have only expanded perturbatively in increasing number of loops. However, it is equally valid to expand and truncate in powers of $\gamma$ which is equivalent to the number of vertices in the diagram.

If we want to truncate at order $\gamma^2$ (with all orders in loops), which corresponds to diagrams with two vertices, unfortunately there is still an infinite family of diagrams that satisfy this constraint. Namely, in the expansion of the square root, we have
\begin{align}
    \sqrt{2P_C^2 - S_C^2 + \mathcal{Q}} &= \sum_{n=0}^{\infty} \binom{\frac{1}{2}}{n} \frac{\mathcal{Q}^{n}}{\left( 2P_C^2 - S_C^2 \right)^{n - \frac{1}{2}}}
,\end{align}
where $\mathcal{Q}$ can contain up to quartic terms in the quantum field $Q_{i}$. As such we absorb the complexities of this expansion into an object, $\tensor{P}{^{j_{i} \cdots j_{n}}_{\mu_1 \cdots \mu_{n}}}$ which captures all combinations that lead to $n$ $Q_{i}$'s in a given term with
\begin{align}
    \sqrt{2P_C^2 - S_C^2 + \mathcal{Q}} &= \sum_{n=0}^{\infty} \tensor{P}{^{j_{i} \cdots j_{n}}_{\mu_1 \cdots \mu_{n}}} \prod_{i=1}^{n} \partial^{\mu_{i}} Q_{j_i}
,\end{align}
from this expression, we read off the Feynman rules, treating $\tensor{P}{^{j_{1} \cdots j_{n}}_{\mu_1 \cdots \mu_{n}}}$ as a composite object representing the entire effect of the classical field (as it contains only classical fields and derivatives). As such we find that this $(n+1)$-valent vertex has vertex factor

\begin{align} % 3-arm diagram with ... between 2 arms
\vcenter{\hbox{
%\begin{tikzpicture}
%    \begin{feynman}
%    \diagram [small,horizontal=a to d] {
%        e [dot] -- [gluon,rmomentum'=$q\,\,\mathrm{=}-\Sigma_i p^{\mu_i}$] c [crossed dot,label=1:$k_1\text{,}\cdots\text{,}k_n\text{;~}\nu_1\text{,}\cdots\text{,}\nu_n$];
%        a -- [momentum'=$j_1\text{, }p^{\mu_1}$] e,  
%        b -- [momentum=$j_n\text{, }p^{\mu_n}$] e, 
%    };
%    \end{feynman}
%    % \ldots in arc between a and b using tikz
%    \draw [decorate,decoration={brace,amplitude=5pt},xshift=-4pt,yshift=0pt]
%    ($(a.north west) + (0,0.2)$) -- ($(b.south west) + (0,0.2)$) node [black, midway, above=5pt] {\footnotesize $n$}
%    node [black, midway, below=4pt] {\footnotesize $\cdots$};
%\end{tikzpicture}
        \includegraphics{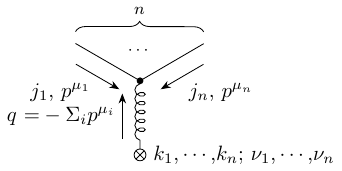}
}}
    &= \int \dd{^{d}q}\tensor{P}{^{k_{i} \cdots k_{n}}_{\nu_1 \cdots \nu_{n}}}\left( q \right) \prod_{i=1}^{n}  p^{\nu_i}
.\end{align}

Therefore, we see that the $(n-1)$ loop and two vertex diagram that we label $\mathcal{D}_{n-1,2}$, (with $n$ propagators) is given by
\begin{align}
    \mathcal{D}_{n-1,2} &= \vcenter{\hbox{
%                \begin{tikzpicture}
%                    \begin{feynman}
%                        \vertex (a);
%                        \vertex [below right=of a] (b);
%                        \vertex [below left=of b] (c);
%                        \vertex [right=of b] (z);
%                        \vertex [above left=of c] (d);
%                        \vertex [above right=of d] (e);
%                        \vertex [left=of d] (x);
%                        \diagram* {
%                            (b) -- [quarter left] (d),
%                            (b) -- [quarter right] (d),
%                            (a) [dot] -- [quarter left] (b) -- [quarter left] (c) [dot] -- [quarter left] (d) -- [quarter left] (e),
%                            (z) [crossed dot] -- [rmomentum=$q$, gluon] (b), (d) -- [rmomentum=$q$, gluon] (x) [crossed dot],
%                        };
%                        % make line invisible but node not
%                        \draw [draw = none, ] (c.south west) -- (a.north west) node [pos=0.5] {$n$} node [pos=0.2] {$\vdots$} node [pos=0.9] {$\vdots$};
%                    \end{feynman}
%                \end{tikzpicture} 
            \includegraphics{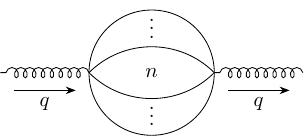}
        }}\label{diagram:2_vertex_n_propagators} \\
        \mathcal{D}_{n-1,2} &= \frac{\sinh^2 \gamma}{\cosh^{n} \gamma}\int \dd{^{d}q}\left( \prod_{i=1}^{n}  \dd{^{d}\ell_{i}} \right) \tensor{P}{^{k_{1} \cdots k_{n}}_{\nu_1 \cdots \nu_{n}}}\left( q \right) \left( \prod_{i=1}^{n}  \frac{p_{i}^{\nu_i} p_{i}^{\tau_i}}{p_{i}^2} \right) \tensor{P}{^{k_{1} \cdots k_n}_{\tau_{i} \cdots \tau_n}}\left( -q \right)  
,\end{align}
where $p_1 = q - \ell_1$, $p_{i} = \ell_{i-1} - \ell_{i}$ for $1 < i < n$ and $p_n = -\ell_{n-1}$ such that
\begin{align}
    \sum_{i= 1}^{n} p_{i} &= q
.\end{align}

\begin{note}
    Diagrams where a loop begins and ends on the same vertex do not contribute, as they vanish in dimensional regularization. We can see this as the momentum circulating around such a loop, $\ell$, will appear multiplicatively in the vertex factor in the form $\ell^{\mu} \ell^{\nu}$, and in the propagator in the form $\frac{1}{\ell^2}$. Factoring this dependence, we see the familiar
    \begin{align}
        \int \frac{\dd{^{d}\ell}}{\left( 2\pi \right)^{d}} \frac{\ell^{\mu} \ell^{\nu}}{\ell^2} &= \frac{g^{\mu \nu}}{d}\int \frac{\dd{^{d}\ell}}{\left( 2\pi \right)^{d}}1
    ,\end{align}
    which we have seen vanishes with dimensional regularization in the limit $d \to 2$, as desired.
\end{note}

Therefore it remains to consider the structure of the divergence of diagrams with $m$ loops between two vertices. To achieve this, we proceed with the evaluation of the diagram $\mathcal{D}_{m, 2}$ (i.e. one with $m+1$ propagators). It is convenient to isolate the part of the integrand which depends on the loop momenta and evaluate it separately. To do this, let us define
\begin{align}
    \left( \mathfrak{L}_{m, 2} \right)_{\{ i j \}}^{\{ \mu \nu \}} &= \int \, \left(  \prod_{i=1}^{m}  \frac{\dd{^d \ell_{i} }}{ ( 2 \pi  )^d } \right) \left( \frac{1}{\prod_{j=1}^{m + 1} p^\mu_j p_{\mu j} } \right) \left( \prod_{k=1}^{m + 1}  p_{i_k}^{\nu_k} p_{j_k}^{\mu_k} \right) \label{frak_L_defn}
.\end{align}
Here we use $\{ i j \}$ as a shorthand for the multi-index $\{ i_1 \ldots i_{m+1} j_1 \ldots j_{m+1} \}$ and $\{ \mu \nu \}$ for $\{ \mu_1 \ldots \mu_{m+1} \nu_1 \ldots \nu_{m+1} \}$. We will sometimes suppress these multi-indices in writing $\mathfrak{L}_{m, 2}$ for convenience. The quantity $\mathfrak{L}_{m, 2}$ determines the value of the diagram $\mathcal{D}_{m, 2}$ as
\begin{align}
    \mathcal{D}_{m, 2} &= \frac{\sinh^2 ( \gamma )}{\cosh^{m + 1} ( \gamma ) }\int \frac{\dd{^{d}q}}{( 2 \pi )^d} \tensor{P}{^{ i_1 }^{ \cdots }^{ i_{m + 1} }_{ \nu_1 }_{ \cdots }_{\nu_{m + 1}}}(q) \left( \mathfrak{L}_{m, 2} \right)_{\{ i j \}}^{\{ \mu \nu \}} \tensor{P}{^{j_1}^{\cdots}^{j_{m + 1}}_{\mu_1}_{\cdots}_{\mu_{m + 1}}} \left( -q \right)
,\end{align}
so to understand the divergences in $\mathcal{D}_{m, 2}$, it suffices to understand those in $\mathfrak{L}_{m, 2}$.

One can evaluate the divergence structure of these diagrams by evaluating each loop integral in succession, beginning with $\ell_{m}$ and proceeding backwards. See \cref{sec:scalar_field_n_loop_2_vertex_calc} for a full demonstration of this process which yields a result proportional to $\Gamma \left( -\frac{d}{2} \right) \ell^{d}_{m-1}$.

Therefore applying this argument recursively, each successive loop integral gains an additional $\ell^{d}_{i}$ factor, resulting after all $m$ integrals, in an external momenta $q$ dependence of
\begin{align}
    \left( \mathfrak{L}_{m, 2} \right)_{\{ i j \}}^{\{ \mu \nu\} } \propto \Gamma\left( -\frac{d}{2} \right)^{m} q^{dm} \label{final_n_loop_2_vertex_dimreg}
.\end{align}

However, as each integral yields $6$ different symmetrizations of the external indices, the exact form of an $m$ loop diagram contains $6^{m}$ different symmetrizations and is thus challenging to write explicitly in a general form. Regardless, we can comment on the structure of the divergences present, as they are the central object of interest.

While proceeding in dimensional regularization facilitated the characterization of the divergence, it is useful to quote the dependence of such divergences on an external characteristic momenta scale, $\Lambda$. Namely, one can show that the presence of a $\frac{1}{\epsilon}$ divergence in dimensional regularization is equivalent to a logarithmic divergence of the form
\begin{align}
    \frac{1}{\epsilon} \sim \log \left( \Lambda  \right)
.\end{align}
Therefore, as each $m$ loop diagram leads to a divergence of the form $\Gamma \left( -\frac{d}{2} \right)^{m}$, we observe logarithmic divergences of the form
\begin{align}
    \left( \frac{1}{\epsilon} \right)^{m} q^{2m} \sim \left( \log \left( \Lambda \right) \right) ^{m} 
.\end{align}
This is a different structure of divergence for each vertex as observed for the 1 loop $m$ vertex diagrams. However, if a pattern is hidden within these series of divergences which allows one to recover a finite number of terms by summing the series, then the theory would be more amenable to renormalization. While the background field method has proved effective in obtaining these divergences, any possible pattern is obscured by the complex index symmetrization arising from the Feynman rules in this scheme. Thus, having successfully characterized the effective action and its divergences, we look towards an alternative quantization approach, and evaluate the possibly of hidden structures.

\chapter{Auxiliary Fields}

In our approach so far to quantizing ModMax, we have relied heavily on the background field method formalism. The background field method is powerful and effective in that we are able to consider various additional constraints on the theory such as constant classical field strength with ease. It similarly is compatible with the Taylor expansion of the square root present in ModMax and facilitates the truncation at second order in the quantum field. However, as I have demonstrated in my analysis, expanding order by order quickly becomes unfeasible.

In tackling the nonlinearity present in ModMax, the only other approach is the introduction of \textit{auxiliary fields}. Such fields are not physical, but rather are defined in terms of the physical fields to capture some aspect of the nonlinearity. When one considers auxiliary fields, it is simple to show that such an alternative representation of the theory is equivalent to the original at the classical level. However, equivalence at the quantum level is highly nontrivial as any number of classical symmetries may be broken in either the auxiliary or original theory.

We introduce the method of auxiliary fields in the context of describing a relativistic point particle. We then outline an auxiliary field representation of ModMax, and detail the classical behaviour of this theory. We find that this alternative approach alleviates the need for a Taylor expansion of the square root, but necessitates explicitly breaking Lorentz symmetry to proceed in the quantization. This approach is not the focus of this thesis, but is included to demonstrate the comparative effectiveness of the background field method.

\section{Relativistic Point Particle Action}

Contrary to the field theory used throughout this thesis, we now consider the action of a relativistic point particle described by a position vector $x^{\mu} = \left( t, \vec{x} \right) $. Such an action should extremize the \textit{proper time} between events. It can be shown that the choice
\begin{align}
    S = -m\int_{\text{Worldline}} \dd{l}
,\end{align}
achieves this, where $m$ is the mass of the particle and $\dd{l}$ is the proper time between two infinitesimally separated events $x^{\mu}$ and $x^{\mu} + d x^{\mu}$. 

\begin{definition}
    A particle's \textbf{worldline} $X^{\mu}$ is a timelike path in spacetime which the particle follows. 
\end{definition}

\begin{figure}[h]
%\begin{tikzpicture}[scale=1.5]
%    \draw[->] (0,0) -- (3,0) node[right] {$t$};
%    \draw[->] (0,0) -- (0,3) node[above] {$x$};
%    \draw[color=blue,thick,-latex] plot [domain=0:1.5] (\x,{sin(4*\x r)*0.25 + 2*\x }) node[right,above] {$X^{\mu}\left( \tau \right) $};
%    \draw[dotted] plot [domain=0:3] (\x,{\x});
%\end{tikzpicture}
    \centering
    \includegraphics[width=\figwidth/2]{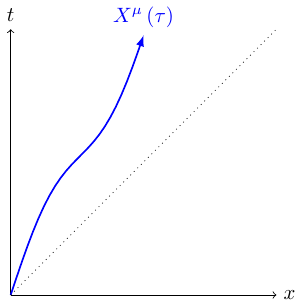}
    \caption{A worldline $X^{\mu}\left( \tau \right) $ parametrized by a variable $\tau$. Notice that it is a timelike curve as it is above the null diagonal line.}
    \label{fig:worldline}
\end{figure}

In Minkowski space, we have that
\begin{align}
    \dd{l}^2 &= -\eta_{\mu \nu} \dd{x^{\mu}} \dd{x^{\nu}} \nonumber\\
    &= \dd{t}^2 - \dd{\vec{x}}^2,
    \intertext{which with the parametrization $\vec{x}\left( t \right) $, we can write as}
   \dd{l}^2 &= \dd{t}^2 - \left( \dv{\vec{x}}{t} \right)^2 \dd{t}^2 \nonumber \\
    &= \dd{t}^2  \left( 1- \left( \dv{\vec{x}}{t} \right)^2 \right) \nonumber \\
   \dd{l} &= \dd{t} \sqrt{ 1- \left( \dv{\vec{x}}{t} \right)^2}, 
   \intertext{where we define $\displaystyle v^2 = \left( \dv{x}{t} \right)^2$ yielding}
\dd{l} &= \dd{t} \sqrt{ 1- v^2}
.\end{align}

This allows us to write the action as
\begin{align}
   S  &= -m \int_{\text{Worldline}} \dd{t} \sqrt{ 1- v^2}
.\end{align}

However, this action is not clearly Lorentz invariant as we desire for a relativistic point particle. Rather than parametrizing our particle's worldline by a time $t$, we consider the worldline to be parameterized by a variable $\tau$. Thus we can write the worldline of the particle as $X^{\mu}\left( \tau \right) $. Therefore on the worldline we can write
\begin{align}
    \dd{X}^{\mu} &= \dv{X^{\mu}}{\tau} \dd{\tau},
    \intertext{which leads to }
    \dd{l^2} &= -\eta_{\mu\nu} \dd{X^{\mu}} \dd{X^{\nu}} \nonumber \\
    &= -\eta_{\mu \nu} \dv{X^{\mu}}{\tau} \dv{X^{\nu}}{\tau} \dd{\tau}^2 \nonumber \\
   \dd{l} &=\dd{\tau} \sqrt{-\eta_{\mu \nu} \dv{X^{\mu}}{\tau} \dv{X^{\nu}}{\tau}}
,\end{align}
which we can insert into the action to obtain
\begin{align}
    S &= -m \int \dd{\tau} \sqrt{-\eta_{\mu \nu} \dv{X^{\mu}}{\tau} \dv{X^{\nu}}{\tau}} \label{eq:covariant_action}
.\end{align}
This is manifestly Lorentz invariant. This form of the action is nonlinear however, and defies traditional quantization techniques in the same fashion as ModMax. Therefore, we consider a representation in terms of the auxiliary field $e\left( \tau \right)$ referred to as the \textit{Einbein action},
\begin{align}
    S = \frac{1}{2} \int \dd{\tau}  \left( e\left( \tau \right)^{-1}\eta_{\mu \nu} \dv{X^{\mu}}{\tau} \dv{X^{\nu}}{\tau} - e\left( \tau \right) m^2 \right) \label{eq:einbein}
.\end{align}

For nonzero mass $m$, the equation of motion for $e\left( \tau \right) $ is
\begin{align}
    \pdv{\mathcal{L}}{e} &= 0 \\
    \implies e\left( \tau \right) &= \frac{1}{m} \sqrt{-\eta_{\mu\nu}\pdv{X^{\mu}}{\tau} \pdv{X^{\nu}}{\tau}} 
,\end{align}
which when reinserted into the Einbein action \cref{eq:einbein} recovers the manifestly Lorentz invariant action in \cref{eq:covariant_action}. This implies the theories are classically equivalent, but makes no predictions about their equivalence after quantization.

\begin{note}
    All such actions are invariant under Lorentz transformations, translations and spatial rotations, as well as reparametrization of $\tau \to \tilde{\tau}$. 
    While under a reparametrization of $\tau \to \tilde{\tau}$, $X^{\mu}\left( \tau \right) $ transforms to $\tilde{X}^{\mu}\left( \tilde{\tau} \right) $, the auxiliary field $e\left( \tau \right) $ transforms as a density such that
    \begin{align}
        \tilde{e}\left( \tilde{\tau} \right) = \left( \dv{\tilde{\tau}}{\tau} \right)^{-1} e\left( \tau \right) 
    .\end{align}
    Therefore, if we view this invariance as a gauge symmetry \cite{string_theory}, then picking a fixed value of $e\left( \tau \right) $ corresponds to fixing a gauge. Choosing $e\left( \tau \right) = 1$, the action becomes
    \begin{align}
        S = \frac{1}{2} \int \dd{\tau} \left( \eta_{\mu \nu} \dv{X^{\mu}}{\tau} \dv{X^{\nu}}{\tau} - m^2 \right) 
    ,\end{align}
    where the equation of motion for $e\left( \tau \right) $ can no longer be imposed, but instead becomes a constraint equation
    \begin{align}
        1 = e\left( \tau \right) &= -\frac{1}{m} \eta_{\mu \nu} \dv{X^{\mu}}{\tau} \dv{X^{\nu}}{\tau} \\
        \implies -m^2 &= \eta_{\mu \nu} \dv{X^{\mu}}{\tau} \dv{X^{\nu}}{\tau}
    ,\end{align}
    which with the identification of momenta $\displaystyle p^{\mu} = \pdv{X^{\mu}}{\tau}$, this constraint can be written as
    \begin{align}
        p^2 &= -m^2
    ,\end{align}
    the familiar mass-shell condition. This form of the action is much more approachable when considering the quantization of the theory, with the only difficulty arising in imposing the constraint.
\end{note}

\section{Auxiliary ModMax}

We seek to apply this approach of introducing auxiliary fields, and then removing them by imposing their equation of motion as a constraint instead.

The ModMax Lagrangian has an auxiliary field representation \cite{Lechner_2022} of
\begin{align}
    \mathcal{L} &= \cosh \gamma S + \sinh \gamma \left( S \phi_1 + P \phi_2 \right) - \frac{1}{2}\rho^2 \left( \phi_1^2 + \phi_2^2 - 1 \right) 
,\end{align}
which has equations of motion for the scalar fields,
\begin{align}
    \sinh \gamma S = \rho^2 \phi_1, && \sinh \gamma P = \rho^2 \phi_2, && \rho \left( \phi_1^2 + \phi_2^2 - 1 \right) = 0, \\
    S = \frac{\rho^2 \phi_1}{\sinh \gamma}, && P = \frac{\rho^2 \phi_2}{\sinh \gamma}. && 
\end{align}
Notice that if $\rho = 0$, then $\phi_1$ and $\phi_2$ are unconstrained and the equations of motion yield $S = P = 0$. The Lagrangian then reduces to the well studied Bialynicki-Birula theory \cite{bb1988, Lechner_2022}.

For finite $\rho$, substituting the equations of motion for $\phi_1$ and $\phi_2$ back into the Lagrangian, we can obtain
\begin{align}
    \mathcal{L} = \cosh \gamma S + \frac{\sinh^2\gamma}{2}\rho^{-2} \left( S^2 + P^2 \right) +\frac{1}{2}\rho^2 
,\end{align}
which has equation of motion for $\rho$
\begin{align}
    \rho^{4} &= \sinh^2 \gamma \left( S^2 + P^2 \right) \label{eq:rho_eom}
.\end{align}

\begin{note}
    Contrary to the relativistic point particle, this action does not have reparametrization invariance. Therefore, to remove the auxiliary fields we must make use of the gauge symmetry already present. However, as $S$ and $P$ are gauge invariants, one does not have freedom to fix the gauge to impose the equation of motion for $\rho$. Nonetheless, we notice that our Lagrangian is not a function of $\partial_0 A_0$ as
\begin{align}
    S &= -\frac{1}{2}\left( \partial_\mu A_{\nu} \partial^{\mu} A^{\nu} - \partial_\mu A_{\nu} \partial^{\nu} A^{\mu} \right) \nonumber\\
    &= -\frac{1}{2} \left( \partial_0 A_i \partial^{0} A^{i} + \partial_j A_i \partial^{j} A^{i} + \partial_i A_0 \partial^{i} A^{0} - \partial_0 A_i \partial^{i} A_0 - \partial_j A_i \partial^{i} A^{j} - \partial_i A_0 \partial^{t} A^{i}  \right),
    \intertext{and}
    P &= E^{i} B_i = -\left( \partial_0 A_i - \partial_i A_0 \right) \epsilon^{ijk} \partial_j A_k
,\end{align}
are functions of $A_0$ but not $\partial_0 A_0$. Therefore, the canonical conjugate momenta
\begin{align}
    \Pi^{0} = \pdv{\mathcal{L}}{\left( \partial_0 A_0 \right) } = 0
,\end{align}
vanishes, suggesting $A_0$ is an ill suited canonical variable. For Maxwell's theory, this is identically, the case, and can be treated by introducing Coulomb gauge, in which one fixes $A_0 = 0$, treating it as a non-dynamical variable.

We posit that there exists a generalized Coulomb gauge, fixing $A_0 = \omega \left( x \right) $ for some scalar function $\omega \left( x \right) $ that imposes the constraint
\begin{align}
    \sinh^2 \gamma \left( S^2 + P^2 \right) &= 1\label{eq:constraint}
,\end{align}
at the level of the equations of motion. Namely, this implies $\rho = 1$ such that $\rho$ becomes fixed and non-dynamical as well. Fixing $A_0$ in this fashion explicitly breaks Lorentz symmetry.
\end{note}

Thus, imposing the resulting constraint on the Lagrangian
\begin{align}
    \mathcal{L} = \cosh \gamma S + \frac{\sinh^2\gamma}{2} \left( S^2 + P^2 \right) + \frac{1}{2}
,\end{align}
is a classically equivalent form of ModMax. This constraint appears quite abstract, and unintuitive in comparison to the $p^2 = -m^2$ constraint that we saw for the Einbein action.

\section{Symmetry Preservation}
As ModMax has EM-duality at the level of the equations of motion, so should this equivalent formalism. Namely, in general, electromagnetic duality is expressed as
\begin{align}
   \mqty( -2\pdv{\mathcal{L}\left( F' \right) }{F'_{\mu\nu}} \\ \widetilde{F}'_{\mu \nu} ) &=  \mqty(	\cos \alpha & \sin \alpha \\	-\sin \alpha & \cos \alpha \\)\mqty( -2\pdv{\mathcal{L}\left( F' \right) }{F'_{\mu\nu}} \\ \widetilde{F}_{\mu \nu} )
,\end{align}
where
\begin{align}
    G^{\mu \nu} \equiv -2\pdv{\mathcal{L}}{F_{\mu\nu}} &= \cosh \gamma F^{\mu \nu} + \rho^{-2}\sinh^2 \gamma \left( S F^{\mu \nu} + P \widetilde{F}^{\mu\nu} \right)
.\end{align}

An equivalent statement of electromagnetic duality \cite{Lechner_2022} is if the theory satisfies
\begin{align}
    G_{\mu \nu} \widetilde{G}^{\mu \nu} &= F_{\mu \nu} \widetilde{F}^{\mu \nu}
.\end{align}

\begin{note}
    We have
    \begin{align}
        \widetilde{G}^{\mu \nu} = \cosh \gamma \widetilde{F}^{\mu \nu} + \rho^{-2}\sinh^2 \gamma \left( S \widetilde{F}^{\mu \nu} - P F^{\mu \nu}\right) 
    ,\end{align}
    which reveals that
    \begin{align}
        G_{\mu \nu} \widetilde{G}^{\mu \nu} &= \cosh^2 \gamma F_{\mu \nu} \widetilde{F}^{\mu \nu} + \rho^{-2}\cosh \gamma \sinh^2 \gamma \times \nonumber\\
        &\phantom{=}~\left[  \widetilde{F}_{\mu \nu} \left( S F^{\mu \nu} + P\widetilde{F}^{\mu \nu} \right) + F_{\mu \nu} \left( S \widetilde{F}^{\mu \nu} - P F^{\mu \nu} \right)  \right]\nonumber  \\
        &\phantom{=}+ \rho^{-4}\sinh^4 \gamma \left( S F_{\mu \nu} + P\widetilde{F}_{\mu \nu} \right) \left( S \widetilde{F}^{\mu \nu} - P F^{\mu \nu} \right) \nonumber \\
        &= -4\cosh^2 \gamma P - 4\rho^{-2}\cosh \gamma \sinh^2 \gamma \left[ 2SP - 2SP \right] \nonumber\\
        &\phantom{=}~+ 4\rho^{-4}\sinh^{4} \gamma \left[ -S^2 P + S^2 P + S^2P + P^3 \right]   \nonumber \\
        &= -4\cosh^2 \gamma P + 4\rho^{-4}\sinh^{4} \gamma P\left[ S^2 + P^2 \right],
        \intertext{where if we impose the equation of motion for $\rho$,}
        G_{\mu \nu} \widetilde{G}^{\mu \nu} &= -4\cosh^2 \gamma P + 4\sinh^{4} \gamma P\left[ \frac{1}{\sinh^2 \gamma} \right]    \nonumber\\
        &= 4\left( \sinh^2 \gamma - \cosh^2 \gamma \right)  P,
        \intertext{where the hyperbolic identity $\sinh^2 \gamma - \cosh^2 \gamma = -1$ provides}
       G_{\mu \nu} \widetilde{G}^{\mu \nu} &= -4P,
        \intertext{and as we have}
        F_{\mu \nu} \widetilde{F}^{\mu \nu} &= -4P\nonumber  \\
\implies G_{\mu \nu} \widetilde{G}^{\mu \nu} &= F_{\mu \nu} \widetilde{F}^{\mu \nu}
    ,\end{align}
    this theory is electromagnetically dual when the constraint is imposed as expected.
\end{note}

To check conformal invariance, we see that the stress energy tensor is given by
\begin{align}
    T_{\mu \nu} &= -2 \left( \pdv{\mathcal{L}}{S} \pdv{S}{g^{\mu \nu}} +  \pdv{\mathcal{L}}{P} \pdv{P}{g^{\mu \nu}}\right) + g_{\mu \nu} \mathcal{L}
,\end{align}
where
\begin{align}
    \pdv{\mathcal{L}}{S} = \cosh \gamma + \rho^{-2}\sinh^2 \gamma S, &&
    \pdv{\mathcal{L}}{P} = \rho^{-2}\sinh^2 \gamma P,
    \intertext{and}
    \pdv{S}{g^{\mu \nu}} = -\frac{1}{2} \tensor{F}{_{\mu}^{\rho}} F_{\nu \rho}, &&
    \pdv{P}{g^{\mu \nu}} = -\frac{1}{4} \left( \tensor{F}{_{\mu}^{\rho}}\widetilde{F}_{\mu \rho} + \tensor{F}{_{\nu}^{\rho}} \widetilde{F}_{\mu \rho} \right) 
,\end{align}
which lead to a trace of the form
\begin{align}
    \tensor{T}{^{\mu}_{\mu}} &= -4 \left( S \pdv{\mathcal{L}}{S} + P \pdv{\mathcal{L}}{P} - \mathcal{L} \right) \nonumber \\
    &= -4 \bigg( S \cosh \gamma + \rho^{-2}S^2\sinh^2 \gamma  + \rho^{-2}P^2 \sinh^2 \gamma - S \cosh \gamma \nonumber\\
    &\phantom{=}\quad \quad- \rho^{-2}\frac{\sinh^2 \gamma}{2} \left( S^2 + P^2 \right) - \frac{1}{2} \rho^2 \bigg)  \nonumber\\
    &= -4  \left( \rho^{-2}\frac{\sinh^2 \gamma}{2} \left( S^2 + P^2 \right) - \frac{1}{2} \rho^2  \right),
    \intertext{where we see the equation of motion \cref{eq:rho_eom} causes this term to exactly vanish}
    \tensor{T}{^{\mu}_{\mu}} &= 0
,\end{align}
as desired. 

Therefore, we have that this auxiliary field representation of ModMax maintains both of the notable symmetries of the original theory at the classical level. Note that this makes no comment on the equivalence at the quantum level, which is significantly inhibited by the breaking of Lorentz symmetry by fixing $A_0$. This digression suggests that one cannot integrate out $\rho$ without explicitly breaking Lorentz symmetry (or purely recovering ModMax itself). This contrasts with the effectiveness of the background field method, in which we were able to preserve Lorentz invariance in our quantization procedure. Alternative auxiliary field representations are a natural extension of this work, however are likely less elucidating of underlying structures than the background field method developed above.

%%%%% Conclusion
\chapter{Conclusion}

The central aim of this project was to perturbatively quantize ModMax by obtaining the effective action. While ModMax's nonlinearity poses a great resistance to traditional quantization techniques, the background field method and dimensional regularization proved highly effective in quantizing this theory. As such, I achieved the central aim of this project by characterizing the effective action arising from quantum corrections in both a static and varying classical background field. I obtained the effective action by evaluating all one loop Feynman diagrams and all two loop diagrams containing up to two vertices.

This effective action provides corrections to the classical theory arising from the quantum domain. I showed that these corrections exactly vanish when the background field is static. This suggests that under this restriction, there are no quantum corrections to this theory, a novel result undiscovered in the literature. However, when the background field was allowed to vary, divergent corrections arose which were not of the form of the original Lagrangian. This was also a novel result. While these corrections appear to respect conformal symmetry, as ModMax is the unique nonlinear theory possessing conformal symmetry and electromagnetic duality, this suggests that these corrections must break electromagnetic duality. While this result hints towards the non-renormalizability of ModMax, it is still a valid effective field theory. Further, the discovery of novel conformal field theories is of great theoretical interest. The natural extension of this investigation would be to investigate the properties of the classical conformal theory generated by these quantum corrections.

This result motivated the study of the two dimensional analogue of ModMax, due to the increased predictive power of conformal symmetry in two dimensions. I applied the method I developed to quantize ModMax to its two dimensional analogue theory. As expected, I obtained corrections of an analogous form, also vanishing when the background field is held constant. Allowing the background to vary, I similarly obtained divergent quantum corrections to the classical theory that are not of the form of the initial Lagrangian. 

For both of the theories investigated, while the quantum corrections obtained were novel results, the divergence and new form of these corrections suggest that the theories do not admit well-behaved quantum versions in their current form. Nonetheless, they are valid as effective field theories, and the possibility of generating new conformal field theories through the corrections obtained is promising. In the landscape of nonlinear electrodynamics, this quantization of ModMax serves to demonstrate the possibility of translating such classical nonlinear theories to the quantum domain.

\appendix

%%%%% Projectors
\chapter{Propagator Derivations} \label{sec:propagator}
For QED, the momentum space photon propagator is given by
\begin{align*}
    D^{\nu \rho} &= \frac{-i}{k^2} \left( g^{\mu \nu} + \left( 1 - \xi \frac{k^{\mu} k^{\nu}}{k^2} \right)  \right)
.\end{align*}

This implies that it is the inverse of the quadratic term in the QED Lagrangian,
\begin{align}
    \mathcal{L} &=  A^{\mu} \underbrace{\left( -k^2 g_{\mu \nu} + \left( 1 - \frac{1}{\xi} \right)\partial_{\mu} \partial_{\nu}  \right)}_{\text{quadratic term}} A^{\nu}
.\end{align}

Namely, we have
\begin{align}
    \left( -k^2 g_{\mu \nu} + \left( 1 - \frac{1}{\xi} \right) \partial_{\mu} \partial_{\nu}  \right) D^{\nu \rho} &= i \delta_{\mu}^{\rho}
.\end{align}

\begin{proof}
    Observe that the propagator is the inverse of this term such that
    \begin{align*}
        &\left( -k^2 g_{\mu \nu} + \left( 1 - \frac{1}{\xi} \right) k_\mu k_\nu \right) \frac{-i}{k^2 } \left( g^{\mu \nu} + \left( 1 - \xi \right)\frac{k^{\mu} k^{\nu}}{k^2}   \right)  \\
       &= -\frac{i}{k^2} \left( -k^2 \delta_{\mu}^{\rho} - \left( 1 - \xi \right) k^2 g_{\mu \nu} \frac{k^{\nu} k^{\rho}}{k^2} + \left( 1 - \frac{1}{\xi} \right) k_\mu k_\nu g^{\nu \rho} + \left( 1 - \frac{1}{\xi} \right) \left( 1 - \xi \right) \frac{k_{\mu} k_{\nu} k^{\nu} k^{\rho}}{k^2}   \right)  \\
        &= -\frac{i}{k^2} \left( -k^2 \delta_{\mu}^{\rho} - \left( 1 - \xi \right) k_{\mu} k^{\rho}  + \left( 1- \frac{1}{\xi} \right) k_\mu k^\rho + \left( 1 - \frac{1}{\xi} \right) \left( 1 - \xi \right) k_{\mu} k^{\rho} \right)  \\
        &= -\frac{i}{k^2} \left( -k^2 \delta_{\mu}^{\rho} + \xi k_{\mu} k^{\rho}  - \frac{1}{\xi} k_\mu k^\rho + \left( 1 - \xi - \frac{1}{\xi} + 1 \right) k_{\mu} k^{\rho} \right)  \\
        &= -\frac{i}{k^2 } \left( -k^2 \delta_{\mu}^{\rho} + 2 k_{\mu} k^{\rho} \right)  \\
        &= -\frac{i}{k^2 } \left( -k^2 \delta^{\rho}_{\mu} \right)  \\
        &= i\delta_{\mu}^{\rho}
    ,\end{align*}
    as desired.
\end{proof}

For ModMax, as the term quadratic in the quantum fields is entirely analogous up to a $\cosh \gamma$ factor and a background dependent term that can be absorbed through gauge choice
\begin{align}
    \mathcal{L}_\text{ModMax} &= \frac{\cosh \gamma}{2}  a_{\nu} \left( -k^2 g^{\mu \nu} + \left( 1 - \frac{1}{\xi} - \tanh \gamma S_C \right) k^{\mu} k^{\nu}  \right) a_{\mu} + \text{~non-quadratic~terms}
,\end{align}
the propagator is identical, with the addition of division by the $\cosh \gamma$ factor with
\begin{align}
    D^{\nu \rho} &= \frac{1}{\cosh \gamma} \frac{-i}{k^2} \left( g^{\mu \nu} + \left( 1 - \left( \xi + \frac{1}{S_C \tanh \gamma} \right)  \frac{k^{\mu} k^{\nu}}{k^2} \right)  \right)
,\end{align}
where throughout my thesis I have chosen $\xi$ such that the propagator simplifies to
\begin{align}
    D^{\nu \rho} &= \frac{1}{\cosh \gamma} \frac{-i g^{\mu \nu}}{k^2}
.\end{align}

\chapter{ModMax 2-Loop Diagram Calculation} \label{sec:2_loop_2_vertex_calc}

We proceed in dimensional regularization with $d \neq 4$ in which the diagram thus evaluates to
\begin{align}
    \mathcal{D}_{\text{2 Loops}} &= \vcenter{\hbox{
%\begin{tikzpicture}
%    \begin{feynman}
%    \diagram [horizontal=b to c] {
%        c [dot] -- [photon, half right, momentum'=$p$] b [dot, label=0:$ $], 
%        b -- [photon,momentum'=$q$] c, 
%        c -- [photon, half left, rmomentum=$p-q$] b, 
%};
%    \end{feynman}
%\end{tikzpicture}
            \includegraphics{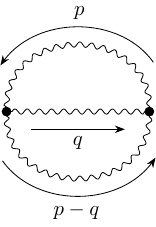}
}} \nonumber\\
&= \frac{\sinh^2 \gamma}{6} \tensor{B}{^{\mu_1}_{\nu_1}^{\rho_1}_{\tau_1}^{\alpha_1}_{\beta_1}} \tensor{B}{^{\mu_2}_{\nu_2}^{\rho_2}_{\tau_2}^{\alpha_2}_{\beta_2}} \times\nonumber \\
&\phantom{=}~\int \frac{\dd{^{d}p}\dd{^{d}q}}{\left( 2\pi \right)^{2d}} p_{\mu_1} q_{\rho_1} \left( p -q \right)_{\alpha_1} p_{\mu_2} q_{\rho_2} \left( p -q \right)_{\alpha_2} D^{\nu_1 \nu_2} D^{\tau_1 \tau_2} D^{\beta_1 \beta_2} \nonumber \\
&= \left( -i \right)^3\frac{\sinh^2 \gamma}{6 \cosh^3 \gamma} \tensor{B}{^{\mu_1}_{\nu_1}^{\rho_1}_{\tau_1}^{\alpha_1}_{\beta_1}} \tensor{B}{^{\mu_2}^{\nu_1}^{\rho_2}^{\tau_1}^{\alpha_2}^{\beta_1}} \times \nonumber\\
&\phantom{=}~\int \frac{\dd{^{d}p}\dd{^{d}q}}{\left( 2\pi \right)^{2d}} \frac{p_{\mu_1} q_{\rho_1} \left( p -q \right)_{\alpha_1} p_{\mu_2} q_{\rho_2} \left( p -q \right)_{\alpha_2}}{p^2 q^2 \left( p - q \right)^2}
.\end{align}

The numerator has terms with either $2,3$ or $4$ factors of $p$ in it (and symmetrically for $q$). Inspecting the $p_{\mu_1} p_{\mu_2}$ term of the integral, we have
\begin{align}
    \mathcal{D}_{\text{2 Loops}} \propto  \mathcal{I}_{\text{2 Loops}}&=\int \frac{\dd{^{d}p} \dd{^{d}q}}{\left( 2\pi \right)^{2d}} \frac{p_{\mu_1} p_{\mu_2}q_{\rho_1} q_{\rho_2} q_{\alpha_1} q_{\alpha_2}}{p^2 q^2 \left( p^2 -2p^{\mu}q_{\mu} + q^2 \right)},
    \intertext{where focusing on the $p$ integral, we write}
   \mathcal{I}_{\text{2 Loops}} &=\int \frac{\dd{^{d}q}}{\left( 2\pi \right)^{d}} \frac{q_{\rho_1} q_{\rho_2} q_{\alpha_1} q_{\alpha_2}}{q^2}\int \frac{\dd{^{d}p}}{\left( 2\pi \right)^{d}} \frac{p_{\mu_1} p_{\mu_2}}{p^2 \left( p^2 -2p^{\mu}q_{\mu} + q^2 \right)},
    \intertext{and introduce a Feynman integral over $x$}
   \mathcal{I}_{\text{2 Loops}} &= \int \frac{\dd{^{d}q}}{\left( 2\pi \right)^{d}} \frac{q_{\rho_1} q_{\rho_2} q_{\alpha_1} q_{\alpha_2}}{q^2}\int \frac{\dd{^{d}p }}{\left( 2\pi \right)^{d}} \int_0^1 \dd{x} \frac{p_{\mu_1} p_{\mu_2}}{\left[ p^2 x +  \left( p^2 -2p^{\mu}q_{\mu} + q^2 \right) \left( 1 - x \right) \right]^2} \nonumber \\
    &= \int \frac{\dd{^{d}q}}{\left( 2\pi \right)^{d}} \frac{q_{\rho_1} q_{\rho_2} q_{\alpha_1} q_{\alpha_2}}{q^2}\int \frac{\dd{^{d}p }}{\left( 2\pi \right)^{d}} \int_0^1 \dd{x} \frac{p_{\mu_1} p_{\mu_2}}{\left[ p^2+  \left( -2p^{\mu}q_{\mu} + q^2 \right) \left( 1 - x \right) \right]^2},
    \intertext{where we complete the square in the denominator}
   \mathcal{I}_{\text{2 Loops}} &= \int \frac{\dd{^{d}q}}{\left( 2\pi \right)^{d}} \frac{q_{\rho_1} q_{\rho_2} q_{\alpha_1} q_{\alpha_2}}{q^2}\int \frac{\dd{^{d}p}}{\left( 2\pi \right)^{d}} \int_0^1 \dd{x} \frac{p_{\mu_1} p_{\mu_2}}{\left[ \left( p_{\mu} - q_{\mu}\left( 1 -x \right)  \right)^2 - q^2\left( 1 - x \right)^2 \right]^2},
    \intertext{and make the translation $p_{\mu} \to p_{\mu} + q_{\mu}\left( 1- x \right) $,}
   \mathcal{I}_{\text{2 Loops}} &=  \int \frac{\dd{^{d}q}}{\left( 2\pi \right)^{d}} \frac{q_{\rho_1} q_{\rho_2} q_{\alpha_1} q_{\alpha_2}}{q^2}\int \frac{\dd{^{d}p }}{\left( 2\pi \right)^{d}} \int_0^1 \dd{x} \frac{\left( p_{\mu_1} + q_{\mu_1}\left( 1- x \right) \right)  \left( p_{\mu_2} + q_{\mu_2}\left( 1- x \right) \right)}{\left[ p^2 - q^2\left( 1 - x \right)^2 \right]^2},
    \intertext{where here the cross terms in the integrand with one $p$ will vanish, and thus we are left with}
   \mathcal{I}_{\text{2 Loops}} &=  \int \frac{\dd{^{d}q}}{\left( 2\pi \right)^{d}} \frac{q_{\rho_1} q_{\rho_2} q_{\alpha_1} q_{\alpha_2}}{q^2}\int \frac{\dd{^{d}p }}{\left( 2\pi \right)^{d}} \int_0^1 \dd{x} \frac{p_{\mu_1} p_{\mu_2} + q_{\mu_1} q_{\mu_2} \left( 1 - x \right)^2}{\left[ p^2 - q^2\left( 1 - x \right)^2 \right]^2},
    \intertext{where we make use of the known integral (A.4 in \cite{Renormalization})}
    &\int \frac{\dd{^{d}p}}{\left( 2\pi \right)^{d}} \frac{p^{2\beta}}{\left( p^2 - \Delta^2 \right)^{\alpha}} = i \left( -1 \right)^{\alpha + \beta} \frac{\Gamma \left( \beta + \frac{d}{2} \right) \Gamma \left( \alpha - \beta - \frac{d}{2} \right) }{\left( 4\pi \right)^{\frac{d}{2}} \Gamma \left( \alpha \right) \Gamma \left( \frac{d}{2} \right)  } \Delta^{2 \left( \frac{d}{2} - \alpha + \beta \right) } \label{known_integral},
    \intertext{with $\alpha = 2$, $\beta = 0,1$ from $p_{\mu_1}p_{\mu_2}\to \frac{p^2}{4}g_{\mu_1 \mu_2}$ and $\Delta^2 = q^2\left( 1 - x \right)^2$ to obtain}
   \mathcal{I}_{\text{2 Loops}} &=  \int \frac{\dd{^{d}q}}{\left( 2\pi \right)^{d}} \frac{q_{\rho_1} q_{\rho_2} q_{\alpha_1} q_{\alpha_2}}{q^2} \int_0^1 \dd{x}i\bigg[ q_{\mu_1} q_{\mu_2} \left( 1 - x \right)^2 \frac{\Gamma\left( \frac{d}{2} \right) \Gamma \left( 2-\frac{d}{2} \right)}{\left( 4\pi \right)^{\frac{d}{2}}\Gamma\left( \frac{d}{2} \right) } \Delta^{d-4} \nonumber\\
    &\phantom{=}\hspace{5cm}- \frac{g_{\mu_1 \mu_2}}{4} \frac{\Gamma\left( 1+\frac{d}{2} \right)  \Gamma\left( 1-\frac{d}{2} \right)}{\left( 4\pi \right)^{\frac{d}{2}} \Gamma\left( \frac{d}{2} \right) } \Delta^{d-2} \bigg],
    \intertext{where substituting in $\Delta^2$ and using $\Gamma\left( x+1 \right) = x \Gamma \left( x \right) \implies \Gamma \left( 1 + \frac{d}{2} \right) = \frac{d}{2} \Gamma \left( \frac{d}{2} \right)$ leads to}
   \mathcal{I}_{\text{2 Loops}} &=  \int \frac{\dd{^{d}q}}{\left( 2\pi \right)^{d}} \frac{q_{\rho_1} q_{\rho_2} q_{\alpha_1} q_{\alpha_2}}{q^2} \int_0^1 \dd{x}\frac{i}{\left( 4\pi \right)^{\frac{d}{2}}  }\bigg[ q_{\mu_1} q_{\mu_2} q^{d-4} \left( 1 - x \right)^{2+d-4} \Gamma\left( 2-\frac{d}{2} \right) \nonumber \\
    &\phantom{=} \hspace{5cm}- \frac{d}{2}\frac{g_{\mu_1 \mu_2}}{4} \Gamma\left( 1 - \frac{d}{2} \right) q^{d-2} \left( 1 - x \right)^{d-2} \bigg],
    \intertext{where evaluating the Feynman integral gives $\frac{1}{d-1}$ and thus}
   \mathcal{I}_{\text{2 Loops}} &=  \int \frac{\dd{^{d}q}}{\left( 2\pi \right)^{d}} \frac{q_{\rho_1} q_{\rho_2} q_{\alpha_1} q_{\alpha_2}}{q^2} \frac{i}{\left( 4\pi \right)^{\frac{d}{2}} \left( d-1 \right)  }\bigg[ q_{\mu_1} q_{\mu_2} q^{d-4}  \Gamma\left( 2-\frac{d}{2} \right) \nonumber \\
   &\phantom{=} \hspace{7cm}- \frac{dg_{\mu_1 \mu_2}}{8} \Gamma\left( 1 - \frac{d}{2} \right) q^{d-2}  \bigg],
    \intertext{where lastly with $\Gamma\left( 2 -\frac{d}{2} \right) = \left( 1 - \frac{d}{2} \right) \Gamma \left( 1 - \frac{d}{2} \right) $ we arrive at}
   \mathcal{I}_{\text{2 Loops}} &=  \int \frac{\dd{^{d}q}}{\left( 2\pi \right)^{d}} \frac{q_{\rho_1} q_{\rho_2} q_{\alpha_1} q_{\alpha_2}}{q^2} \frac{i\Gamma\left( 1 - \frac{d}{2} \right)}{\left( 4\pi \right)^{\frac{d}{2}} \left( d-1 \right)  }\left[ q_{\mu_1} q_{\mu_2} q^{d-4} \left( 1 - \frac{d}{2} \right) - \frac{dg_{\mu_1 \mu_2}}{8}  q^{d-2}  \right]
,\end{align}
where $\forall d \neq 4$, we are left with a symmetrizable integral over $q$ that will vanish identically. By analytic continuation, in this regularization scheme we conclude that the integrals also vanish at $d = 4$.

The other two possible numerators with $3$ and $4$ factors of $p$ follow similarly as the Feynman integral substitution in the denominator is independent of the momenta in the numerator. Namely, for $3$ factors of $p$, we have
\begin{align}
    \mathcal{D}_{\text{2 Loops}} \propto \mathcal{I}_{p^3\text{~term}} &= \int \frac{\dd{^{d}p} \dd{^{d}q}}{\left( 2\pi \right)^{2d}} \frac{p_{\mu_1} p_{\mu_2} p_{\alpha_1}q_{\rho_1} q_{\rho_2} q_{\alpha_2}}{p^2 q^2 \left( p^2 - 2p^{\mu} q_{\mu} \right) } + \alpha_1 \leftrightarrow \alpha_2
.\end{align}
Focusing on the $p$ integral, the same process and translation $p_{\mu} \to p_{\mu} + xq_{\mu}$ yields
\begin{align}
   \mathcal{I}_{p^3\text{~term}} &=\int \frac{\dd{^{d}q}}{\left( 2\pi \right)^{d}} \frac{q_{\rho_1} q_{\rho_2} q_{\alpha_2}}{q^2} \int \frac{\dd{^{d}p}}{\left( 2\pi \right)^{d}} \frac{\left( p_{\mu_1} + xq_{\mu_1} \right) \left( p_{\mu_2} + x q_{\mu_2} \right) \left( p_{\alpha_1} + x q_{\alpha_1} \right)  }{\left[ p^2 - q^2 x^2 \right]^2 } \nonumber \\
   \mathcal{I}_{p^3\text{~term}} &=\int \frac{\dd{^{d}q}}{\left( 2\pi \right)^{d}} \frac{q_{\rho_1} q_{\rho_2}  q_{\alpha_1}q_{\alpha_2}}{q^2} \int \frac{\dd{^{d}p}}{\left( 2\pi \right)^{d}} \frac{p_{\mu_1} p_{\mu_2} + x^2 q_{\mu_1} q_{\mu_2}}{\left[ p^2 - q^2 x^2 \right]^2 }
,\end{align}
which is in fact exactly as we saw for the $4$ factors of $p$ case. Likewise by symmetry, the $2$ factors will be identical to the $4$ factors under exchange of $p \leftrightarrow q$. Thus it is sufficient to multiply our result by $3$ to account for all terms.

\chapter{Scalar Field $\mathcal{D}_2$ Calculation} \label{sec:scalar_field_1_loop_2_vertex_calc}

Let us first focus on the divergence structure of the diagram $\mathcal{D}_2$ of \cref{D2_diagram}, which we repeat here for convenience:
\begin{align}\label{first_diagram}
    \mathcal{D}_2 \; = \vcenter{\hbox{
%\begin{tikzpicture}
%    \begin{feynman}
%    \diagram [layered layout, horizontal=a to b] {
%        z [crossed dot,label=$\mu\nu;ij$] -- [gluon, momentum=$q$] a [dot]-- [ half left, looseness=1.5, rmomentum=$\ell;k$] b [dot], 
%        a -- [half right, looseness=1.5, momentum'=$\ell + q;l$] b ,
%        b -- [gluon,rmomentum=$q$] c [crossed dot,label=$\rho\tau;mn$];
%    };
%    \end{feynman}
%\end{tikzpicture}
            \includegraphics{figures/Scalar_D2.pdf}
}}
.\end{align}
As we mentioned around \cref{simpler_integral_body}, the value of this diagram can be expressed in terms of the simpler quantity
\begin{align}
    {\mathcal{I}}^{\mu \nu \rho \tau}_2 = \int \frac{\dd{^{d}\ell}}{\left( 2\pi \right)^{d}} \frac{( \ell + q )^{\nu}\ell^{\mu} ( \ell + q )^{\tau} \ell^{\rho} }{\ell^2 \left( \ell + q \right)^2} .
\end{align}
All of the dependence on loop momenta is encoded within ${\mathcal{I}}^{\mu \nu \rho \tau}_2$. From the value of $\mathcal{I}^{\mu \nu \rho \tau}_2$, the original diagram $\mathcal{D}_2$ is recovered from the expression in \cref{I2_to_D2}, which only involves additional dependence on the classical background via the tensor $\tensor{P}{_\mu_\nu^i^j}$ and an additional integration over the momentum $q$. Therefore, in order to study the divergences arising from the loop, it suffices to perform dimensional regularization of the quantity $\mathcal{I}^{\mu \nu \rho \tau}_2$.

Expanding out the products and introducing a Feynman parameter $x$ in order to resolve the denominator, we find
\begin{align}
    \mathcal{I}^{\mu \nu \rho \tau}_2 &= \int \frac{\dd{^{d}\ell}}{\left( 2\pi \right)^{d}} \frac{\ell^{\nu}\ell^{\mu} \ell^{\tau} \ell^{\rho} + \ell^{\nu} \ell^{\mu} q^{\tau} \ell^{\rho} + q^{\nu} \ell^{\mu} \ell^{\tau} \ell^{\rho}  + q^{\nu} \ell^{\mu} q^{\tau} \ell^{\rho}}{\ell^2 \left( \ell + q \right)^2} \nonumber \\
    &=\int \frac{\dd{^{d}\ell}}{\left( 2\pi \right)^{d}} \int_0^{1} \dd{x}\frac{\ell^{\nu}\ell^{\mu} \ell^{\tau} \ell^{\rho} + \ell^{\nu} \ell^{\mu} q^{\tau} \ell^{\rho} + q^{\nu} \ell^{\mu} \ell^{\tau} \ell^{\rho}  + q^{\nu} \ell^{\mu} q^{\tau} \ell^{\rho}}{\left[ \ell^2 \left( 1 - x \right)  +  x\left( \ell + q \right)^2 \right]^2} \nonumber \\
    &= \int \frac{\dd{^{d}\ell}}{\left( 2\pi \right)^{d}} \int_0^{1} \dd{x}\frac{\ell^{\nu}\ell^{\mu} \ell^{\tau} \ell^{\rho} + \ell^{\nu} \ell^{\mu} q^{\tau} \ell^{\rho} + q^{\nu} \ell^{\mu} \ell^{\tau} \ell^{\rho}  + q^{\nu} \ell^{\mu} q^{\tau} \ell^{\rho}}{\left[ \ell^2  + x \left( 2\ell_{\mu} q^{\mu} + q^2 \right) \right]^2} \nonumber \\
    &= \int \frac{\dd{^{d}\ell}}{\left( 2\pi \right)^{d}} \int_0^{1} \dd{x}\frac{\ell^{\nu}\ell^{\mu} \ell^{\tau} \ell^{\rho} + \ell^{\nu} \ell^{\mu} q^{\tau} \ell^{\rho} + q^{\nu} \ell^{\mu} \ell^{\tau} \ell^{\rho}  + q^{\nu} \ell^{\mu} q^{\tau} \ell^{\rho}}{\left[ \left( \ell^{\mu} + xq^{\mu}  \right) ^2  +x \left( 1 - x \right)q^2  \right]^2} .
\end{align}
In the final step, we have completed the square in the denominator by adding and subtracting $q^2 x^2$. We now shift the integration variable from $\ell^\mu$ to
\begin{align}
    \ell^{\prime \mu} = \ell^\mu - x q^\mu ,
\end{align}
which causes the denominator to become even in $\ell^\prime$, and thus terms in the numerator which are odd in $\ell^{\prime \mu}$ will vanish by symmetry. We immediately drop the primes on $\ell^{\prime \mu}$ and write the surviving terms as
\begin{align}
    \mathcal{I}^{\mu \nu \rho \tau}_2 &= \int \frac{\dd{^{d}\ell}}{\left( 2\pi \right)^{d}} \int_0^{1} \dd{x} \Bigg[ \frac{\ell^{\nu} \ell^{\mu} \ell^{\tau} \ell^{\rho}}{\left( \ell^2  + q^2x\left( 1 - x \right)  \right)^2} + \frac{x^2  \ell^{\nu} q^{\mu} \ell^{\tau} q^{\rho} }{\left( \ell^2 + q^2 x \left( 1 - x \right)  \right)^2} + \frac{\left( x^2 - 2x + 1 \right) q^{\nu} \ell^{\mu} q^{\tau} \ell^{\rho}}{\left( \ell^2  + q^2x\left( 1 - x \right)  \right)^2} \nonumber \\
                                      &\phantom{=}~ +  \frac{\left( x^2 - x  \right) \left( q^{\nu} q^{\mu} \ell^{\tau} \ell^{\rho} + q^{\nu} \ell^{\mu} \ell^{\tau} q^{\rho} + \ell^{\nu} \ell^{\mu} q^{\tau} q^{\rho} + \ell^{\nu} q^{\mu} q^{\tau} \ell^{\rho} \right) }{\left[ \ell^2  + q^2x \left( 1 - x \right)  \right]^2}  + \frac{\left( x^{4} -2x^3 + x^2 \right)  q^{\nu} q^{\mu} q^{\tau} q^{\rho}}{\left( \ell^2  + q^2x\left( 1 - x \right)  \right)^2} \Bigg] \nonumber \\
    &= \int \frac{\dd{^{d}\ell}}{\left( 2\pi \right)^{d}} \int_0^{1} \dd{x} \Bigg[ \frac{\ell^{\nu} \ell^{\mu} \ell^{\tau} \ell^{\rho}}{\left( \ell^2  + q^2x\left( 1 - x \right)  \right)^2} + \frac{ x^2  \ell^{\nu} q^{\mu} \ell^{\tau} q^{\rho} }{\left( \ell^2 + q^2 x \left( 1 - x \right)  \right)^2} + \frac{\left( 1 - x \right)^2 q^{\nu} \ell^{\mu} q^{\tau} \ell^{\rho}}{\left[ \ell^2  + q^2x\left( 1 - x \right)  \right]^2}  \nonumber\\
    &\phantom{=}~ + \frac{ x\left( 1 -x \right) \left( q^{\nu} q^{\mu} \ell^{\tau} \ell^{\rho} + q^{\nu} \ell^{\mu} \ell^{\tau} q^{\rho} +  \ell^{\nu} \ell^{\mu} q^{\tau} q^{\rho} + \ell^{\nu} q^{\mu} q^{\tau} \ell^{\rho} \right) }{\left[ \ell^2  + q^2x \left( 1 - x \right)  \right]^2}   + \frac{\left[ x^2 \left( 1 - x \right)^2 \right]  q^{\nu} q^{\mu} q^{\tau} q^{\rho}}{\left[ \ell^2  + q^2x\left( 1 - x \right)  \right]^2} \Bigg] ,
\end{align}
where in the last expression we have factored various polynomials.

By a symmetry argument similar to the one discussed around \cref{eq:sym_rep_used} and \cref{eq:symmetrization}, within the integral we can replace products of loop momenta with symmetrized combinations of metric tensors:
\begin{align}\label{sym_replacements_app}
    \ell^{\mu} \ell^{\nu} &\to \frac{1}{d} \ell^2 g^{\mu \nu} , \nonumber \\
    \ell^{\mu} \ell^{\nu} \ell^{\rho} \ell^{\tau} &\to \frac{1}{d \left( d + 2 \right) } \ell^4 \left( g^{\mu \nu} g^{\rho \tau} + g^{\mu \rho} g^{\nu \tau} + g^{\mu \tau} g^{\nu \rho} \right) .
\end{align}
Applying the replacements \cref{sym_replacements_app}, the integral $\mathcal{I}^{\mu \nu \rho \tau}_2$ becomes
\begin{align}\label{app_A1_intermediate}
    \mathcal{I}^{\mu \nu \rho \tau}_2 &= \int \frac{\dd{^{d}\ell}}{\left( 2\pi \right)^{d}} \int_0^{1} \dd{x} \Bigg( \frac{\ell^{4}}{d \left( d + 2 \right) }\frac{g^{\mu \nu} g^{\rho \tau} + g^{\mu\tau} g^{\nu \rho} + g^{\mu \rho} g^{\nu \tau}  }{\left[ \ell^2  + q^2x\left( 1 - x \right)  \right]^2} + \frac{\ell^2}{d}\frac{x^2 g^{\nu \tau}q^{\mu} q^{\rho} }{\left[ \ell^2  + q^2x \left( 1 - x \right)  \right]^2} \nonumber\\
    &\qquad + \frac{\ell^2}{d}\frac{ \left( 1 - x \right)^2  g^{\mu \rho} q^{\nu} q^{\tau} }{\left[ \ell^2  + q^2x\left( 1 - x \right)  \right]^2} + \frac{\ell^2}{d}\frac{x\left( 1 -x \right) \left( q^{\nu} q^{\mu} g^{\tau\rho} + g^{\mu \tau} q^{\nu} q^{\rho} + g^{\nu \mu} q^{\tau} q^{\rho} + g^{\nu \rho}q^{\mu} q^{\tau} \right) }{\left[ \ell^2  + q^2x \left( 1 - x \right)  \right]^2} \nonumber\\
    &\qquad + \frac{ x^2 \left( 1 - x \right)^2  q^{\nu} q^{\mu} q^{\tau} q^{\rho}}{\left[ \ell^2  + q^2x\left( 1 - x \right)  \right]^2} \Bigg) .
\end{align}
It will be convenient to make use of the standard result
\begin{align}\label{first_anselmi}
    \int \frac{\dd{^{d}\ell}}{\left( 2\pi \right)^{d}} \frac{\ell^{2\beta}}{\left( \ell^2 - \Delta^2 \right)^{\alpha}} = i \left( -1 \right)^{\alpha + \beta} \frac{\Gamma \left( \beta + \frac{d}{2} \right) \Gamma \left( \alpha - \beta - \frac{d}{2} \right) }{\left( 4\pi \right)^{\frac{d}{2}} \Gamma \left( \alpha \right) \Gamma \left( \frac{d}{2} \right)  } \Delta^{2 \left( \frac{d}{2} - \alpha + \beta \right) } ,
\end{align}
which can be found, for instance, in equation (A.4) in \cite{Renormalization}. Using \cref{first_anselmi} with
\begin{align}\label{delta_sq_defn}
    \Delta^2 = -q^2 x\left( 1 - x \right) 
\end{align}
in \cref{app_A1_intermediate}, we find
\begin{align}
    \mathcal{I}^{\mu \nu \rho \tau}_2 &= \frac{i}{\left( 4\pi \right)^{\frac{d}{2}} \Gamma \left( 2 \right) \Gamma \left( \frac{d}{2} \right)  } \int_0^{1} \dd{x} \Bigg( \frac{\Delta^{d}}{d \left( d + 2 \right) }\Gamma \left( 2 + \frac{d}{2} \right) \Gamma \left( -\frac{d}{2} \right) \left( g^{\mu \nu} g^{\rho \tau} + g^{\mu \rho} g^{\nu \tau} + g^{\mu \tau} g^{\nu \rho} \right)  \nonumber\\
    &\quad +\frac{\Delta^{d - 2}}{d} \Gamma \left( 1 + \frac{d}{2} \right) \Gamma \left( 1 - \frac{d}{2} \right)  x \left( 1 - x \right)  \left( q^{\nu} q^{\mu} g^{\tau \rho} + g^{\mu \tau} q^{\nu} q^{\rho} +  g^{\nu \mu} q^{\tau} q^{\rho} + g^{\nu \rho} q^{\mu} q^{\tau} \right) \nonumber\\
    &\quad + \frac{\Delta^{d-2}}{d} \Gamma \left( 1 + \frac{d}{2} \right) \Gamma \left( 1 - \frac{d}{2} \right) x^2 g^{\nu \tau} q^{\mu} q^{\rho} +\frac{\Delta^{d-2}}{d} \Gamma \left( 1 + \frac{d}{2} \right) \Gamma \left( 1 - \frac{d}{2} \right) \left( 1 - x \right)^2  g^{\mu \rho} q^{\nu} q^{\tau}\nonumber \\
    &\quad +\Delta^{d-4}\Gamma \left( \frac{d}{2} \right) \Gamma \left( 2 - \frac{d}{2} \right)  \left[ x^2 \left( 1 - x \right)^2  \right] q^{\nu} q^{\mu} q^{\tau} q^{\rho} \Bigg) .
\end{align}
Using gamma function identities and some algebra, one can simplify this to
\begin{align}
    \mathcal{I}^{\mu \nu \rho \tau}_2 &= \frac{i\Gamma \left( -\frac{d}{2} \right)}{\left( 4\pi \right)^{\frac{d}{2}}} \int_0^{1} \dd{x} \Bigg( \frac{\Delta^{d}}{4} \left( g^{\mu \nu} g^{\rho \tau} +   g^{\mu \rho} g^{\nu \tau} + g^{\mu \tau} g^{\nu \rho} \right) -\frac{d\Delta^{d-2}}{4} x^2 g^{\nu \tau} q^{\mu} q^{\rho}  \nonumber \\
    &\quad -\frac{d\Delta^{d-2}}{4} \left( 1 - x \right)^2 g^{\mu \rho} q^{\nu} q^{\tau} - \frac{d\Delta^{d - 2}}{4} x \left( 1 - x \right) \left( q^{\nu} q^{\mu} g^{\tau \rho} + g^{\mu\tau} q^{\nu} q^{\rho} + g^{\nu \mu} q^{\tau} q^{\rho} + g^{\nu \rho} q^{\mu} q^{\tau} \right) \nonumber\\
    &\quad  + \frac{d\left( d - 2 \right) \Delta^{d-4}}{4} x^2 \left( 1 - x \right)^2 q^{\nu} q^{\mu} q^{\tau} q^{\rho} \Bigg) .
\end{align}
After substituting in for $\Delta^2$ using the definition \cref{delta_sq_defn}, we can now evaluate the resulting integrals using the formula
\begin{align}
    \int_0^{1} \dd{x} x^{\alpha - 1} \left( 1 - x \right)^{\beta - 1} &= \frac{\Gamma \left( \alpha \right) \Gamma \left( \beta \right) }{\Gamma \left( \alpha + \beta \right) } = \mathrm{B} ( \alpha, \beta ) , 
\end{align}
which we recognize as the definition of the beta function $\mathrm{B} ( \alpha , \beta )$. By doing this, we find
\begin{align}\label{appendix_intermediate_two}
    \mathcal{I}^{\mu \nu \rho \tau}_2 &= \frac{i\Gamma \left( -\frac{d}{2} \right)}{4\left( 4\pi \right)^{\frac{d}{2}}} \int_0^{1} \dd{x} \Bigg( \nonumber \\
                                      &\phantom{=+}~  d \left( d - 2 \right) q^{d-4}\left[ -x\left( 1 - x \right)   \right]^{\frac{d}{2}} q^{\nu} q^{\mu} q^{\tau} q^{\rho} - dq^{d-2}  \left( -x \right)^{\frac{d}{2}-1}\left( 1 - x \right)^{\frac{d}{2}+1} g^{\mu \rho} q^{\nu} q^{\tau} \nonumber\\
    &\quad + q^{d} \left[ -x\left( 1 - x \right) \right]^{\frac{d}{2}} \left( g^{\mu \nu} g^{\rho \tau} + g^{\mu \rho} g^{\nu \tau} + g^{\mu \tau} g^{\nu \rho} \right)    -dq^{d-2} \left[ \left( -x \right)^{\frac{d}{2}+1} \left( 1 - x \right)^{\frac{d}{2}-1}\right] g^{\nu \tau} q^{\mu} q^{\rho}  \nonumber \\
    &\quad +dq^{d - 2} \left[-x \left( 1 - x \right)  \right]^{\frac{d}{2}} \left( q^{\nu} q^{\mu} g^{\tau \rho} + g^{\mu \tau} q^{\nu} q^{\rho} + g^{\nu \mu} q^{\tau} q^{\rho} + g^{\nu \rho} q^{\mu} q^{\tau} \right)   \Bigg) \nonumber \\
    &= \frac{i \left( -1 \right)^{\frac{d}{2}}\Gamma \left( -\frac{d}{2} \right) }{4\left( 4\pi \right)^{\frac{d}{2}}} \Bigg( q^{d} \frac{\Gamma \left( \frac{d}{2} + 1 \right)^2}{\Gamma \left( d + 2 \right) } \left( g^{\mu \nu} g^{\rho \tau} + g^{\mu \rho} g^{\nu \tau} + g^{\mu \tau} g^{\nu \rho} \right)   \nonumber \\
    &\quad +dq^{d-2} \frac{\Gamma \left( \frac{d}{2} \right) \Gamma \left( \frac{d}{2} + 2 \right) }{\Gamma \left( d + 2 \right) } \left[ g^{\mu \rho} q^{\nu} q^{\tau} + g^{\nu \tau} q^{\mu} q^{\rho} \right] + d\left( d - 2 \right) q^{d - 4} \frac{\Gamma\left( \frac{d}{2} + 1 \right)^2}{\Gamma \left( d + 2 \right)} q^{\nu} q^{\mu} q^{\tau} q^{\rho} \nonumber \\
    &\quad +dq^{d-2} \frac{\Gamma \left( \frac{d}{2} + 1 \right)^2}{\Gamma \left( d + 2 \right) } \left( q^{\nu} q^{\mu} g^{\tau \rho} + g^{\mu \tau} q^{\nu}q^{\rho} +  g^{\nu \mu} q^{\tau} q^{\rho} + g^{\nu \rho} q^{\mu} q^{\tau} \right)  \Bigg) .
\end{align}
Note that each term in \cref{appendix_intermediate_two} scales as $q^d$, as expected. Factoring out the gamma functions, we have found
\begin{align}
    \mathcal{I}^{\mu \nu \rho \tau}_2 &= \frac{i \left( -1 \right)^{\frac{d}{2}}\Gamma \left( -\frac{d}{2} \right) }{4\left( 4\pi \right)^{\frac{d}{2}}} \frac{\Gamma \left( \frac{d}{2} + 1 \right)^2}{\Gamma \left( d + 2 \right) } \Bigg[ q^{d}  \left( g^{\mu \nu} g^{\rho \tau} + g^{\mu \rho} g^{\nu \tau} + g^{\mu \tau} g^{\nu \rho} \right) + d \left( d - 2 \right) q^{d - 4}q^{\nu} q^{\mu} q^{\tau} q^{\rho} \nonumber \\
    &\quad + dq^{d-2} \left( q^{\nu} q^{\mu} g^{\tau \rho} + g^{\mu \tau} q^{\nu} q^{\rho} + g^{\nu \mu} q^{\tau} q^{\rho} + g^{\nu \rho} q^{\mu} q^{\tau} \right)  +\left( d + 2 \right) q^{d-2} \left( g^{\mu \rho} q^{\nu} q^{\tau} + g^{\nu \tau} q^{\mu} q^{\rho} \right)  \Bigg] .
\end{align}
Finally, to perform dimensional regularization, we set the spacetime dimension to $d = 2 + 2 \epsilon$ and take $\epsilon \to 0$ using the limiting behavior
\begin{align}
    \Gamma \left( -1-\epsilon \right) = \frac{1}{\epsilon} - \gamma + 1 +  \mathcal{O}\left( \epsilon \right)
,\end{align}
for the gamma functions. Keeping only divergent terms, we arrive at the final expression
\begin{align}\label{two_vertex}
    \mathcal{I}^{\mu \nu \rho \tau}_2 &= \left( \frac{1}{\epsilon} \right) \frac{-i}{24\left( 4\pi \right)} \bigg[ q^{2}  \left( g^{\mu \nu} g^{\rho \tau}  + g^{\mu \rho} g^{\nu \tau} + g^{\mu \tau} g^{\nu \rho} \right) +2 \left( q^{\nu} q^{\mu} g^{\tau \rho} + g^{\mu \tau} q^{\nu} q^{\rho}+ g^{\nu \mu} q^{\tau} q^{\rho} + g^{\nu \rho} q^{\mu} q^{\tau} \right) \nonumber \\
    &\qquad \qquad \qquad \qquad + 4\left( g^{\mu \rho} q^{\nu} q^{\tau} + g^{\nu \tau} q^{\mu} q^{\rho} \right) \bigg] .
\end{align}
This completes the evaluation of the divergent contribution from $\mathcal{I}^{\mu \nu \rho \tau}_2$, which justifies the result \cref{divergent_final_body} which was quoted in the body of the paper.

\chapter{Scalar Field $m$-Loop 2-Vertex Calculation} \label{sec:scalar_field_n_loop_2_vertex_calc}

In this appendix, we will show how to evaluate the integral over one of the $m$ loop momenta $\ell_i$ which appear in the expression for the two-vertex, $m$-loop diagram of \cref{diagram:2_vertex_n_propagators}. It suffices to integrate over the final momentum $\ell_{m}$, since the result may then be iterated to evaluate the other $m-1$ integrals.

Notice that the momentum dependence of $\ell_m$ can be isolated such that only $\ell_{m-1}$ is involved in the expression. Specifically, we will compute the quantity

\begin{align}\label{one_integral_Lm2_defn}
    L_{m, 2} &= \int \dd{^{d}\ell_{m}} \frac{\ell_{m}^{\nu_{m+1}} \ell_{m}^{\tau_{m+1}} \left( \ell_{m} - \ell_{m-1} \right)^{\nu_{m}} \left( \ell_{m} - \ell_{m-1} \right)^{\tau_{m}}}{\ell_{m}^2 \left( \ell_{m} - \ell_{m-1} \right)^2}  .
\end{align}

This object $L_{m, 2}$ is proportional to the remaining integrand that one finds by performing the integral over $\ell_m$ in the definition of $\mathfrak{L}_{m, 2}$. As we will see, after obtaining an expression for $L_{m,2}$, this result can be used recursively to evaluate $\mathfrak{L}_{m, 2}$ itself.

We notice the integral in \cref{one_integral_Lm2_defn} is exactly of the form of the one appearing in the $1$-loop, $2$-vertex diagram which we evaluated in \cref{sec:scalar_field_1_loop_2_vertex_calc}. Proceeding in the same way, we introduce a Feynman parameter $x$ to write

\begin{align}
    L_{m, 2} &= \int \dd{^{d}\ell_{m}} \dd{x} \frac{\ell_{m}^{\nu_{m+1}} \ell_{m}^{\tau_{m+1}} \left( \ell_{m} - \ell_{m-1} \right)^{\nu_{m}} \left( \ell_{m} - \ell_{m-1} \right)^{\tau_{m}}}{\left[ \left( 1 - x \right) \ell_{m}^2  + x \left( \ell_{m} - \ell_{m-1} \right)^2 \right]^2} \nonumber \\
    &= \int \dd{^{d}\ell_{m}} \dd{x} \frac{\ell_{m}^{\nu_{m+1}} \ell_{m}^{\tau_{m+1}} \left( \ell_{m} - \ell_{m-1} \right)^{\nu_{m}} \left( \ell_{m} - \ell_{m-1} \right)^{\tau_{m}}}{\left[ \ell_{m}^2  + x \left(  2\ell_{m} \cdot \ell_{m-1}- \ell_{m-1}^2 \right) \right]^2} \nonumber \\
    &= \int \dd{^{d}\ell_{m}} \dd{x} \frac{\ell_{m}^{\nu_{m+1}} \ell_{m}^{\tau_{m+1}} \left( \ell_{m} - \ell_{m-1} \right)^{\nu_{m}} \left( \ell_{m} - \ell_{m-1} \right)^{\tau_{m}}}{\left[ \left( \ell_{m} + x \ell_{m-1} \right)^2 - x^2 \ell_{m-1}^2  + x \ell_{m-1}^2 \right]^2} ,
\end{align}
or after shifting the integration variable as $\ell_{m} \to \ell_{m} - x \ell_{m-1}$,
\begin{align}\hspace{-30pt}
    L_{m, 2} &= \int \dd{^{d}\ell_{m}} \dd{x} \frac{\left( \ell_{m} - x\ell_{m-1} \right)^{\nu_{m+1}} \left( \ell_{m} - x\ell_{m-1} \right) ^{\tau_{m+1}} \left( \ell_{m} - \left( 1 - x \right) \ell_{m-1} \right)^{\nu_{m}} \left( \ell_{m} - \left( 1 - x \right) \ell_{m-1} \right)^{\tau_{m}}}{\left[ \ell_{m}^2 + x \left( 1 - x \right)  \ell_{m-1}^2  \right]^2} .
\end{align}

We may keep only even powers of $\ell_{m}$ in the integrand, as odd powers vanish by symmetry:

\begin{align}
    L_{m, 2} &= \int \dd{^{d}\ell_{m}} \dd{x} \Bigg( \frac{  \ell_{m}^{\nu_{m+1}} \ell_{m}^{\tau_{m+1}} \ell_{m}^{\nu_{m}} \ell_{m}^{\tau_{m}}  + x^2 \ell_{m-1}^{\nu_{m+1}} \ell_{m-1}^{\tau_{m+1}} \ell_{m}^{\nu_{m}} \ell_{m}^{\tau_{m}} }{\left[ \ell_{m}^2 + x \left( 1 - x \right)  \ell_{m-1}^2  \right]^2} \nonumber \\
    &\phantom{=}~+ \frac{x \left( 1 - x \right)  \ell_{m}^{\nu_{m+1}} \ell_{m-1}^{\tau_{m+1}} \ell_{m}^{\nu_{m}} \ell_{m-1}^{\tau_{m}} + x \left( 1 - x \right)  \ell_{m-1}^{\nu_{m+1}} \ell_{m}^{\tau_{m+1}} \ell_{m-1}^{\nu_{m}} \ell_{m}^{\tau_{m}}}{\left[ \ell_{m}^2 + x \left( 1 - x \right)  \ell_{m-1}^2  \right]^2} \nonumber \\
    &\phantom{=}~+ \frac{ \left( 1 - x \right)^2 \ell_{m}^{\nu_{m+1}} \ell_{m}^{\tau_{m+1}} \ell_{m-1}^{\nu_{m}} \ell_{m-1}^{\tau_{m}} + x^2 \left( 1 - x \right)^2 \ell_{m-1}^{\nu_{m+1}} \ell_{m-1}^{\tau_{m+1}} \ell_{m-1}^{\nu_{m}} \ell_{m-1}^{\tau_{m}} }{\left[ \ell^2_{m} + x\left( 1 - x \right)\ell_{m-1}^2 \right]^2} \Bigg) .
\end{align}

We now replace products of $\ell_m^\mu$ with powers of $\ell_m$ and symmetrized metric factors, following the generalized symmetrization rule \cref{eq:symmetrization}, which yields

\begin{align}
    L_{m, 2} &= \int \dd{^{d}\ell_{m}} \dd{x} \Bigg( \frac{ \frac{\ell_{m}^{4}}{d \left( d + 2 \right) }  g^{(\nu_{m+1} \tau_{m+1}} g^{\nu_{m} \tau_{m} )}  + x^2 \frac{\ell_{m}^2}{d} g^{\mu_{m} \nu_{m}}\ell_{m-1}^{\nu_{m+1}} \ell_{m-1}^{\tau_{m+1}}  }{\left[ \ell_{m}^2 + x \left( 1 - x \right)  \ell_{m-1}^2  \right]^2} \nonumber \\
             &\phantom{=}~+ \frac{ x \left( 1 - x \right)  \frac{\ell_{m}^2}{d} g^{\nu_{m+1} \nu_{m}} \ell_{m-1}^{\tau_{m+1}}  \ell_{m-1}^{\tau_{m}} + x \left( 1 - x \right)  \frac{\ell_{m}^2}{d} g^{\tau_{m+1} \nu_{m}}\ell_{m-1}^{\nu_{m+1}} \ell_{m-1}^{\tau_{m}} }{\left[ \ell_{m}^2 + x \left( 1 - x \right)  \ell_{m-1}^2  \right]^2} \nonumber \\
    &\phantom{=}~+ \frac{\left( 1 - x \right)^2 \frac{\ell_{m}^2}{d}  g^{{\nu_{m+1}}{\tau_{m+1}}} \ell_{m-1}^{\nu_{m}} \ell_{m-1}^{\tau_{m}} + x^2 \left( 1 - x \right)^2 \ell_{m-1}^{\nu_{m+1}} \ell_{m-1}^{\tau_{m+1}} \ell_{m-1}^{\nu_{m}} \ell_{m-1}^{\tau_{m}} }{\left[ \ell^2_{m} + x\left( 1 - x \right)\ell_{m-1}^2 \right]^2} \Bigg)  .
\end{align}

Splitting the numerator up, we can once again apply the standard formula \cref{known_integral} with 

\begin{align}\label{deltasq_app_defn}
    \Delta^2 = - x \left(  1 - x \right) \ell_{m-1}^2 ,
\end{align}
which gives
\begin{align}
    L_{m, 2} &= \frac{i}{\left( 4\pi \right)^{\frac{d}{2}} \Gamma \left( \frac{d}{2} \right) }\int_{0}^{1} \dd{x} \Bigg( \frac{1}{d \left( d + 2 \right) } g^{(\nu_{m+1} \tau_{m+1}} g^{\nu_{m} \tau_{m})}  \Gamma \left( 2 + \frac{d}{2} \right) \Gamma \left( -\frac{d}{2} \right) \Delta^{d}  \nonumber \\
    &\qquad \qquad \qquad + \frac{x^2}{d} g^{\mu_{m} \nu_{m}} \ell_{m-1}^{\nu_{m+1}} \ell_{m-1}^{\tau_{m+1}} \Gamma \left( 1 + \frac{d}{2} \right) \Gamma \left( 1 - \frac{d}{2} \right) \Delta^{d-2}  \nonumber\\
    &\qquad \qquad \qquad + \frac{x\left( 1 - x \right) }{d} g^{\nu_{m+1} \nu_{m}} \ell_{m-1}^{\tau_{m+1}} \ell_{m-1}^{\tau_{m}} \Gamma \left( 1 + \frac{d}{2} \right) \Gamma \left( 1 - \frac{d}{2} \right) \Delta^{d-2}  \nonumber\\
    &\qquad \qquad \qquad + \frac{x\left( 1 - x \right) }{d} g^{\tau_{m+1} \nu_{m}} \ell_{m-1}^{\nu_{m+1}} \ell_{m-1}^{\tau_{m}} \Gamma \left( 1 + \frac{d}{2} \right) \Gamma \left( 1 - \frac{d}{2} \right) \Delta^{d-2}  \nonumber\\
    &\qquad \qquad \qquad + \frac{\left( 1 - x \right)^2 }{d} g^{\nu_{m+1} \tau_{m+1}} \ell_{m-1}^{\nu_{m}} \ell_{m-1}^{\tau_{m}} \Gamma \left( 1 + \frac{d}{2} \right) \Gamma \left( 1 - \frac{d}{2} \right) \Delta^{d-2}  \nonumber\\
    &\qquad \qquad \qquad + x^2 \left( 1 - x \right)^2  \ell_{m-1}^{\nu_{m+1}} \ell_{m-1}^{\tau_{m+1}} \ell_{m-1}^{\nu_{m}} \ell_{m-1}^{\tau_{m}} \Gamma \left( \frac{d}{2} \right) \Gamma \left( 2 - \frac{d}{2} \right) \Delta^{d-4} \Bigg) .
\end{align}

Replacing $\Delta$ using its definition in \cref{deltasq_app_defn}, we see that each term contains $d$ overall factors of loop momenta:

\begin{align}
    L_{m, 2} &= \frac{i}{\left( 4\pi \right)^{\frac{d}{2}} \Gamma \left( \frac{d}{2} \right) }\int_{0}^{1} \dd{x} \Bigg( \frac{1}{d \left( d + 2 \right) }  g^{(\nu_{m+1} \tau_{m+1}} g^{\nu_{m} \tau_{m})}  \Gamma \left( 2 + \frac{d}{2} \right) \Gamma \left( -\frac{d}{2} \right) \left( x\left( 1 - x \right) \right)^{\frac{d}{2}} \ell_{m-1}^{d}   \nonumber \\
    &\qquad \qquad \qquad + \frac{1}{d} g^{\mu_{m} \nu_{m}} \ell_{m-1}^{\nu_{m+1}} \ell_{m-1}^{\tau_{m+1}} \Gamma \left( 1 + \frac{d}{2} \right) \Gamma \left( 1 - \frac{d}{2} \right) x^{\frac{d}{2}+1}\left( 1 - x \right)^{\frac{d}{2} - 1} \ell_{m-1}^{d-2} \nonumber\\
    &\qquad \qquad \qquad + \frac{1}{d} g^{\nu_{m+1} \nu_{m}} \ell_{m-1}^{\tau_{m+1}} \ell_{m-1}^{\tau_{m}} \Gamma \left( 1 + \frac{d}{2} \right) \Gamma \left( 1 - \frac{d}{2} \right) \left( x\left( 1- x \right)  \right)^{\frac{d}{2}} \ell_{m-1}^{d-2} \nonumber\\
    &\qquad \qquad \qquad + \frac{1}{d} g^{\tau_{m+1} \nu_{m}}  \Gamma \left( 1 + \frac{d}{2} \right) \Gamma \left( 1 - \frac{d}{2} \right) \left( x\left( 1- x \right)  \right)^{\frac{d}{2}} \ell_{m-1}^{\nu_{m+1}} \ell_{m-1}^{\tau_{m}}\ell_{m-1}^{d-2} \nonumber\\
    &\qquad \qquad \qquad + \frac{1}{d} g^{\nu_{m+1} \tau_{m+1}}  \Gamma \left( 1 + \frac{d}{2} \right) \Gamma \left( 1 - \frac{d}{2} \right) x^{\frac{d}{2}-1}\left( 1- x \right)^{\frac{d}{2}+1}  \ell_{m-1}^{\nu_{m}} \ell_{m-1}^{\tau_{m}}\ell_{m-1}^{d-2} \nonumber\\
    &\qquad \qquad \qquad +  \Gamma \left( \frac{d}{2} \right) \Gamma \left( 2 - \frac{d}{2} \right) \left( x\left( 1- x \right)  \right)^{\frac{d}{2}} \ell_{m-1}^{\nu_{m+1}} \ell_{m-1}^{\tau_{m+1}} \ell_{m-1}^{\nu_{m}} \ell_{m-1}^{\tau_{m}}\ell_{m-1}^{d-4} \Bigg)  .
\end{align}

We now use the identity $\Gamma \left( 1 + x \right) = x \Gamma \left( x \right)$ to factor and cancel the gamma functions, and finally evaluate the Feynman integrals, giving the result

\begin{align}
    L_{m, 2} &= \frac{i\Gamma \left( -\frac{d}{2} \right)}{\left( 4\pi \right)^{\frac{d}{2}} \Gamma \left( d + 2 \right) } \Bigg( \frac{ \left( 1 + \frac{d}{2} \right) \frac{d}{2}}{d \left( d + 2 \right) }   g^{(\nu_{m+1} \tau_{m+1}} g^{\nu_{m} \tau_{m})}   \Gamma \left( \frac{d}{2}+1 \right)^2 \ell_{m-1}^{d}   \nonumber \\
    &\qquad \qquad \qquad + \frac{\left( \frac{d}{2} \right) \left( -\frac{d}{2} \right)  }{d} g^{\mu_{m} \nu_{m}} \ell_{m-1}^{\nu_{m+1}} \ell_{m-1}^{\tau_{m+1}}  \Gamma \left( \frac{d}{2} + 2 \right) \Gamma \left( \frac{d}{2} \right) \ell_{m-1}^{d-2} \nonumber\\
    &\qquad \qquad \qquad + \frac{\frac{d}{2}\left( -\frac{d}{2} \right) }{d} g^{\nu_{m+1} \nu_{m}} \Gamma \left( \frac{d}{2} + 1 \right)^2 \ell_{m-1}^{\tau_{m+1}} \ell_{m-1}^{\tau_{m}} \ell_{m-1}^{d-2} \nonumber\\
    &\qquad \qquad \qquad + \frac{\frac{d}{2}\left( -\frac{d}{2} \right) }{d} g^{\tau_{m+1} \nu_{m}}   \Gamma \left( \frac{d}{2} + 1 \right)^2 \ell_{m-1}^{\nu_{m+1}} \ell_{m-1}^{\tau_{m}}\ell_{m-1}^{d-2} \nonumber\\
    &\qquad \qquad \qquad + \frac{\frac{d}{2}\left( -\frac{d}{2} \right) }{d} g^{\nu_{m+1} \tau_{m+1}}   \Gamma \left( \frac{d}{2} + 2 \right) \Gamma \left( \frac{d}{2} \right)  \ell_{m-1}^{\nu_{m}} \ell_{m-1}^{\tau_{m}}\ell_{m-1}^{d-2} \nonumber\\
    &\qquad \qquad \qquad +  \left( 1-\frac{d}{2} \right) \left( -\frac{d}{2} \right)   \Gamma \left( \frac{d}{2} + 1 \right)^2 \ell_{m-1}^{\nu_{m+1}} \ell_{m-1}^{\tau_{m+1}} \ell_{m-1}^{\nu_{m}} \ell_{m-1}^{\tau_{m}}\ell_{m-1}^{d-4} \Bigg) .
\end{align}

We see that all terms scale as $\ell_{m-1}^{d}$ and the only divergent gamma function is $\Gamma \left( -\frac{d}{2} \right) $. This establishes the result used in the body of the paper, in the text above \cref{final_n_loop_2_vertex_dimreg}, which can then be applied iteratively to evaluate the remaining loop integrals.

\backmatter

%%%%% Bibliography, in BibTeX format (the .bib file)
\bibliography{thesis.bib}

\end{document}